\documentclass[12pt]{article}
\usepackage{latexsym}
\usepackage[parfill]{parskip}
\usepackage{amssymb,amsmath, bm}
\usepackage{graphicx}
\usepackage{xr-hyper}
\usepackage{marvosym}
\usepackage{bbm}
\usepackage{booktabs}
\usepackage{setspace}
\usepackage{subfigure}
\usepackage{titlesec}
\usepackage{threeparttable}
\usepackage{mathtools}
\usepackage{float}
\usepackage{dsfont}
\usepackage{xcolor}
\usepackage[colorlinks,citecolor=blue,urlcolor=blue,linkcolor=blue,bookmarks=false]{hyperref}
 \usepackage{enumerate}
 \usepackage{soul} 
\usepackage[title, titletoc]{appendix}
\titlespacing*{\section}
{0pt}{.5ex}{.5ex}
\titlespacing*{\subsection}
{0pt}{.5ex}{.5ex}
\titlespacing*{\subsubsection}
{0pt}{.5ex}{.5ex}
\titlespacing*{\paragraph}
{0pt}{.5ex}{.5ex}

\usepackage{natbib}

\usepackage{capt-of}

\usepackage{dcolumn}
\newcolumntype{.}{D{.}{.}{-1}}
\newcolumntype{d}[1]{D{.}{.}{#1}}

\usepackage{pgf,tikz}
\usetikzlibrary{positioning}
\usetikzlibrary{shapes}
\usepackage{pgfplots}
\pgfplotsset{compat=1.18}
\usetikzlibrary{matrix}       
\usetikzlibrary{arrows,shapes,positioning,shadows,trees}

\usepackage{amsthm}
\newtheorem{theorem}{Theorem}[section]

\newtheorem{proposition}[theorem]{Proposition}

\newtheorem{assumption}[theorem]{Assumption}

\newcommand\independent{\protect\mathpalette{\protect\independenT}{\perp}}
\def\independenT#1#2{\mathrel{\rlap{$#1#2$}\mkern2mu{#1#2}}}

\usepackage{color}

\makeatletter
\newcommand*{\addFileDependency}[1]{
  \typeout{(#1)}
  \@addtofilelist{#1}
  \IfFileExists{#1}{}{\typeout{No file #1.}}
}
\makeatother

\def\spacingset#1{\renewcommand{\baselinestretch}%
{#1}\small\normalsize} \spacingset{1}

\newcommand{\blind}{1}

\begin{document}

\def\spacingset#1{\renewcommand{\baselinestretch}%
{#1}\small\normalsize} \spacingset{1}


\if1\blind
{
  \title{\bf Nonparametric Causal Inference for Optogenetics: Sequential Excursion Effects for Dynamic Regimes}
  \author{Gabriel Loewinger\thanks{G.L. and A.W.L. contributed equally to this work.} \thanks{National Institute of Mental Health, NIH. Email: gloewinger@gmail.com},
\and Alexander W. Levis$^*$\thanks{Postdoctoral Researcher, Carnegie Mellon University. Email: alevis@cmu.edu},
\and
Francisco Pereira\thanks{National Institute of Mental Health, NIH.}
}
  \maketitle
} \fi

\if0\blind
{
  \bigskip
  \bigskip
  \bigskip
  \begin{center}
    {\LARGE\bf Nonparametric Causal Inference for Optogenetics: Sequential Excursion Effects for Dynamic Regimes}
\end{center}
  \medskip
} \fi

\bigskip
\begin{abstract}
Optogenetics is a powerful neuroscience technique for studying how neural circuit manipulation affects behavior. Standard analysis conventions discard information and severely limit the scope of the causal questions that can be probed. To address this gap, we 1) draw connections to the causal inference literature on sequentially randomized experiments, 2) propose non-parametric {\color{black} estimators} for analyzing ``open-loop'' (static regime) optogenetics behavioral experiments, 3) derive extensions of history-restricted marginal structural models for dynamic treatment regimes with positivity violations for ``closed-loop'' designs, and 4) propose a taxonomy of identifiable causal effects that encompass a far richer collection of scientific questions compared to standard methods. From another view, our work {\color{black} builds upon} ``excursion effect'' methods---popularized recently in the mobile health literature---to estimate causal contrasts for {\color{black} general} treatment sequences in the presence of positivity violations. We describe sufficient conditions for identifiability of the proposed causal estimands, and provide asymptotic statistical guarantees for a proposed inverse probability-weighted estimator, a multiply-robust estimator (for two intervention timepoints), a hypothesis testing {\color{black} procedure}, and a computationally scalable implementation. Finally, we apply our methodology to data from a recent neuroscience study and show how it provides insight into causal effects of optogenetics on behavior that are obscured by standard analyses. 
\end{abstract}

\noindent%
{\it Keywords:}  marginal structural models, optogenetics, excursion effects, dynamic treatment regimes, 
micro-randomized trials, sequentially randomized experiments, neuroscience
\vfill

\newpage
\spacingset{1.9} 
\section{Introduction} \label{sec:intro}
Optogenetics is a neuroscience technique to ``turn on/off'' neurons \textit{in vivo} in real-time, with millisecond time resolution. It works by shining lasers on neurons that have been genetically modified through viral infection to express a light-sensitive protein. It is one of the most popular neuroscience assays, with roughly 700 articles referencing it in 2023 alone.\footnote{
699 papers on Web of Science mention ``optogenetics'' published in 2023 and between 526-796 per year in 2015-2023.} 
It is commonly applied in behavioral tasks that are composed of a sequence of timepoints or decision points {\color{black} (referred to as ``trials'' in the neuroscience literature),} $t \in \{1,2,...,T\}$, each of which involves presentation of stimuli and an opportunity for a behavioral response. For example, a {\color{black} timepoint} might begin with a cue (e.g., a light), which indicates that a lever press will trigger delivery of a food reward. Investigators might want to know, for instance, whether applying an optogenetic stimulation {\color{black} treatment} on a random subset of timepoints alters the rate at which mice press the lever. On {\color{black} timepoint} \textit{t}, an animal's {\color{black} longitudinal} outcome, $Y_t$, time-varying covariates, $X_t$, and {\color{black} sequentially randomized} (optogenetic) treatment indicator, $A_t$, are observed. Experiment{\color{black}al designs} often include both treatment ($G=1$) and control ($G=0$) {\color{black} arms (often referred to as ``groups'' in the neuroscience literature \citep{spont_da})}, with animals assigned randomly to {\color{black} one arm at baseline}. While laser stimulation is often applied on a random\footnote{Some studies deterministically set the sequence of laser/no-laser decisions but we focus on ``stochastic'' experimental designs.} subset of {\color{black} timepoints} in both arms {\color{black} (i.e., $\mathbb{P}(A_t=1 \mid G = 1)>0$ and $\mathbb{P}(A_t=1 \mid G = 0)>0$)}, 
only treatment arm animals express the protein that enables the laser to trigger the target neural response. 
{\color{black} This is analogous to a clinical trial with sequential randomization of pill delivery/no-delivery $A_t$ (at each $t$), where treatment arm participants ($G=1$) receive the active treatment on $A_t=1$ timepoints, and control arm subjects ($G=0$) receive an inactive ``sugar pill'' on $A_t=1$ timepoints.} 
To answer the question above, investigators often estimate the effect of optogenetic manipulation through  comparisons such as $\psi_t = \mathbb{E} [Y_{t} \mid G = 1] - \mathbb{E} [Y_{t} \mid G = 0]$, i.e., the difference between the average outcomes in the treatment and control arms.
%
It is common to test whether $\psi_t = 0$ at specific timepoints like the end of the study ($t = T$), or to conduct inference on summaries (e.g., $\bar{\psi} =\frac{1}{T} \sum_t \psi_t$). These \textit{between-arm} comparisons assess the intervention impact based on long-term, ``macro''/``global'' longitudinal effects. 
Such macro effects marginalize over the \textit{observed} treatment sequence, $\{A_t\}_{t=1}^T$. When $A_t$ is randomized at each $t$, the sequences $\{A_t\}_{t=1}^T$ may be vary highly across subjects in the timing (i.e., $\{t:A_t=1\}$) or number (i.e.,  $\sum_t A_t$) of stimulations.
In studies with stochastic experimental regimes, \textit{within-arm} comparisons between laser and no-laser {\color{black} timepoints} are also common (e.g., $\tilde{\psi} = \sum_t\left\{\mathbb{E} [ Y_{t} \mid A_t=1,~G = 1] - \mathbb{E} [ Y_{t} \mid A_t=0,~G = 1]\right\}$).
Importantly, such comparisons do not probe
\textit{within-arm} ``micro’’/``local'' longitudinal effects related to specific stimulation sequence patterns. For example, one might ask whether there is a dose-dependent relationship between the outcome and the number of stimulations in the last five {\color{black} timepoints}.
Figures~\ref{fig:estimands}A-C 
show some representative micro effects that are identifiable in many optogenetics studies, yet typically are not explored. Such effects may be present even in studies in which one fails to detect the macro effects commonly tested. 
However, no formal causal inference framework has been applied to these studies, resulting in analysis conventions that limit the scope of questions researchers can ask.

Furthermore, certain experimental designs can complicate the use and interpretation of even the standard analysis approaches. In ``dynamic regime'' designs (referred to as ``closed-loop'' in the neuroscience literature), stimulation is applied depending on the behavior of the animal.
For example, say a study tests if lever pressing for food, $Y_t$, decreases if optogenetic treatment is applied ($A_t = 1$), with positive probability, only when animals approach the lever ($X_t = 1$). Since $A_t$ is randomized conditional on $X_t$, any valid analysis must incorporate $X_t$. However, standard strategies like including $X_t$ as a covariate in a regression can bias causal effect estimates because 
1) $X_t$ influences the probability of {\em both} future outcomes and treatment, {\color{black}$A_t$}, and thus can be cast as a time-varying confounder, and 2) $X_t$ \textit{also} mediates the effect of prior treatments, {\color{black} $\{A_j : j < t\}$} \citep{robins1986, hernan_causal_2023} (see the associated DAG in Figure~\ref{fig:estimands}D).
In other words, closed-loop designs induce ``treatment–confounder feedback'' \citep{hernan_causal_2023}, which can lead to bias
with standard analyses. We include an example in Appendix~\ref{app:conf-ill} to show how, when {\color{black} $A_t$} has opposing effects on $Y_t$ and $X_t$, 
treatment {\color{black} ($G=1$)} and control {\color{black} ($G=0$)}  arms can exhibit \textit{identical} average (observed) outcome levels even if the laser {\color{black} (i.e., $A_t$)} causes a large immediate effect {\color{black} (i.e., on $Y_t$)}. 
Furthermore, standard regression approaches 
can actually induce collider-bias, or block mediators of the treatment effect \citep{hernan_causal_2023}. Finally, if treatment policies deterministically rule out treatment (e.g., $\mathbb{P}(A_t=1 \mid X_t = 0) = 0$), certain effects are not identifiable: the positivity violation inherent to these experimental designs precludes estimation of certain counterfactual distributions. Closed-loop designs therefore require specialized approaches for valid causal inference.

\begin{figure*} 
	\centering
\includegraphics[width=0.99\linewidth]{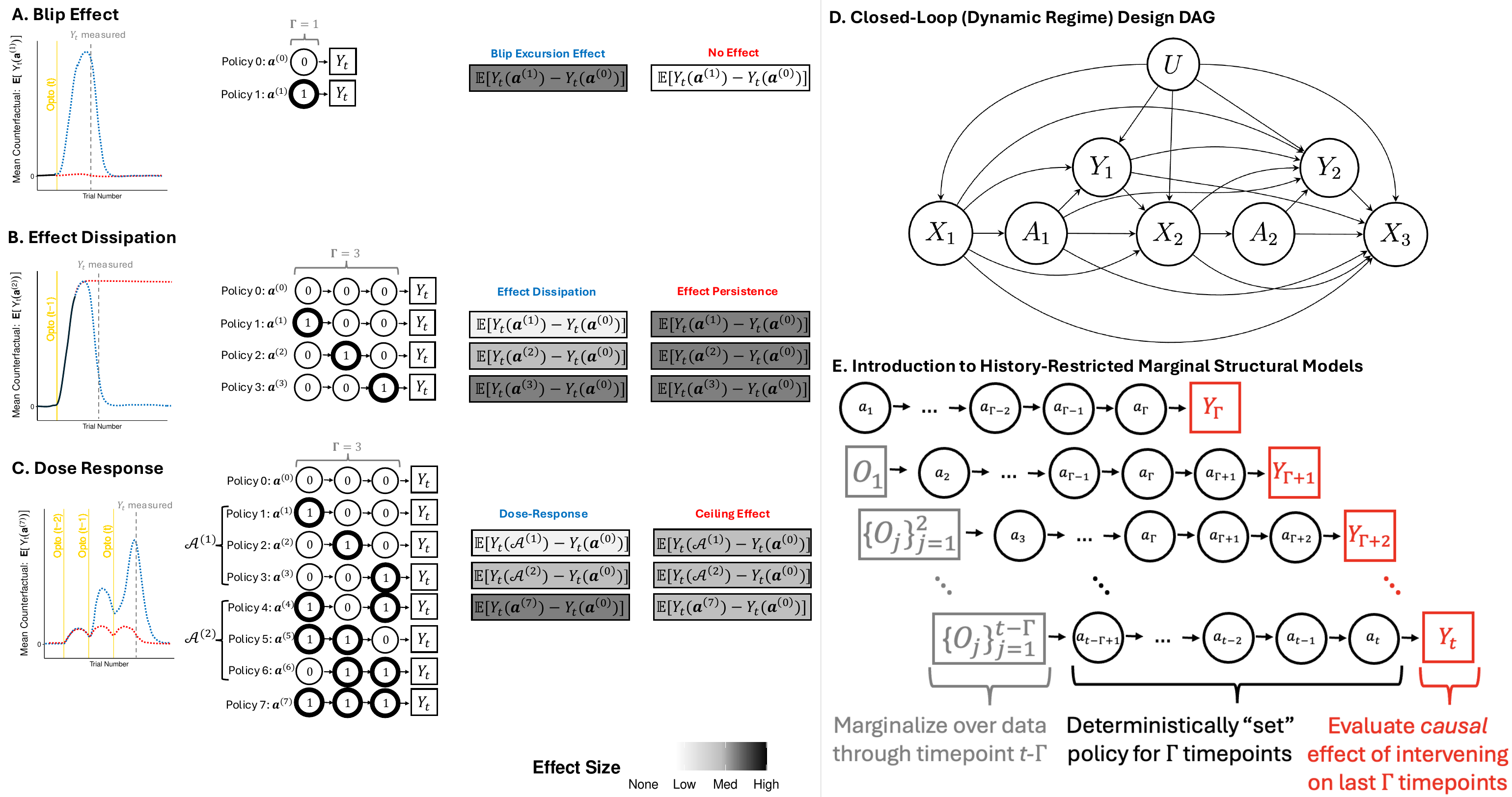}
\caption{\footnotesize\textbf{ \textit{Sequential Excursion Effects}.} [A]-[C] The left panels show one setting where a sequence of laser simulations \color{blue} do \color{black} or \color{black} do not \color{black} have the indicated effect on the outcome. The middle panel shows deterministic static policies that could be used to construct a causal contrast to probe the effect. The right panel shows what the anticipated effect size (darker is larger) of the contrast might be if the effect \color{blue} was \color{black} or \color{black} was not \color{black} present. [A] \textbf{Blip Effect}: the effect of a single stimulation vs. no treatment on a recent {\color{black} timepoint}. [B] \textbf{Effect Dissipation}: Whether the effect of a single stimulation causes an effect that rises and dissipates after a few {\color{black} timepoints}, or persists. [C] \textbf{Dose Response}: Do successive simulations increase the response in a dose-dependent fashion? We denote $\mathcal{A}^{(j)}$ as the set of treatment sequences that have dose $j$ and $\mathbb{E}[Y_t(\mathcal{A}^{(j)})] \coloneqq |\mathcal{A}^{(j)}|^{-1} \sum_{\boldsymbol{a}^{(l)} \in \mathcal{A}^{(j)}} \mathbb{E}[Y_t(\boldsymbol{a}^{(l)})]$.
[D] Closed-loop design DAG for two {\color{black} timepoints}. $U$ represents unmeasured variables. 
[E] HR-MSM illustration inspired by figure in \cite{guo2021}. 
} 
\label{fig:estimands}
\end{figure*}

More broadly, there have been numerous high profile calls for more rigorous 
causal inference in neuroscience \citep{barack2022call, ross2024causation, biswas2022statistical, kording_quasi}. However, existing methods 
\citep{pospisil2023, jiang2023instrumental, lepperod2023} focus on instrumental variable-based approaches to estimate causal effects of optogenetics on neural activity. Unlike our setting, these methods are restricted to datasets with
measurements of the activity of the stimulated neurons, and the neurons they interact with. 
We focus on behavioral outcomes and explicitly deal with sequentially randomized designs, whereas prior work treats each {\color{black} timepoint} as exchangeable.


To fill these gaps, we (1) propose the first, to our knowledge, formal counterfactual-based causal framing of these behavioral {\color{black} experiments}, (2) make study design recommendations by contextualizing optogenetic {\color{black} studies} within the sequentially randomized experiment causal inference literature, (3) propose an analysis {\color{black} approach} based on history-restricted marginal structural models that enables the estimation of ``sequential excursion effects'' to capture the local causal contrasts described above, (4) {\color{black} build upon} excursion effect methodology to accommodate dynamic regimes in cases where experimental designs violate positivity assumptions, and treatment sequences greater than length one, and (5) apply our approach to data from a \textit{Nature} paper, and show it reveals effects obscured by standard methods.

Our paper is organized as follows. We review relevant causal inference literature in Section~\ref{sec:prelim}. We 
present a formal counterfactual framework for causal inference in optogenetic experiments
and introduce our proposed methodology in Section~\ref{sec:prop_estimands}. 
We provide theoretical guarantees for identification and estimation in Section~\ref{sec:ident-est}. Sections~\ref{sec:sims} and~\ref{sec:application} contain simulation results and a real data application. We make recommendations for optogenetics experimental designs in Section~\ref{sec:recs-opto}. Finally, we provide concluding remarks in Section~\ref{sec:discuss}.

\section{Preliminaries} \label{sec:prelim}

\subsection{Notation, Designs and Terminology}\label{sec:notation}
\paragraph{Notation}
Let $\mathcal{O}_{t} = \{X_{t}, A_{t}, Y_{t}\}$ be the vector of \textit{observed} variables for an animal on {\color{black} timepoint} $t$, namely $X_{t}$ denotes time-varying covariates, $A_{t}$ is a binary indicator of (optogenetic) treatment, and $Y_{t}$ denotes the behavioral outcome. We denote $T$ as the number of {\color{black} timepoints} and write 
$[T] = \{1,2,...,T\}$. A sample of subjects $[n]$
is collected but, as subjects are exchangeable, we often suppress indices to reduce notational burden. We express \textit{counterfactual} variables, or \textit{potential outcomes}, with parentheses. For example, $Y_t(\overline{a}_t)$, represents the potential outcome that would be observed at $t$ if a subject received the treatment sequence, $\overline{a}_t = (a_1,\ldots,a_{t})$. Overbars represent all history 
up to and including a specific {\color{black} timepoint}. For instance, $\overline{B}_j = (B_1, \ldots, B_j)$, for any sequence of variables $\{B_t\}_{t = 1}^{T}$, and $j \in [T]$. Finally, we define the history ${H}_t = (\overline{X}_t, \overline{A}_{t-1}, \overline{Y}_{t - 1})$, so $H_t$ includes all information {\em prior} to the treatment ``decision'' at $t$, together with the covariates observed at $t$. 


\paragraph{Optogenetics Experimental Designs}\label{sec:exp_design}

Many optogenetics studies collect data with \textit{stochastic} experimental designs (e.g., treat at {\color{black} timepoint} $t$ with probability $\mathbb{P}(A_t = 1) = 0.5$).
Both \textit{dynamic} and \textit{static} stochastic designs are common. Regimes are considered \textit{dynamic} if they set treatment probabilities as a function of the history $H_t$, and \textit{static} if they assign treatment independent of $H_t$. In the neuroscience literature, dynamic and static regimes are usually referred to as ``closed loop'' and ``open loop,'' respectively.
Dynamic experimental policies are common in neuroscience, as many scientific questions bear expressly on the role of treatment within levels of time-varying covariates. In addition to the behavioral criteria mentioned above, researchers might apply optogenetic manipulation with positive probability only when certain measured neural activity patterns are observed. 

\subsection{Background and Relevant Literature}
In sequentially randomized experiments, \textit{marginal structural models} (MSMs) are often used to model $\mathbb{E}[Y_t(\overline{a}_t)]$ \citep{msm2, msm3}, the expected potential outcome that would be observed at {\color{black} timepoint} $t$ if a subject received the treatment sequence  $\overline{a}_t$.
By placing structure on $\mathbb{E}[Y_t(\overline{a}_t)]$, MSMs can borrow strength across treatment sequences to increase power when there are many timepoints. MSMs are often fit by using generalized estimating equations (GEE) with inverse probability of treatment weighting (IPW). 
Even when treatment probabilities are known, however, IP weights for timepoint $t$ are calculated as the product of $t$ treatment probabilities, and thus can grow unstable computationally \citep{hr_msm2007}. For static treatments sequences, stabilized weights are often used to address this problem \citep{hernan2002}. \citet{orellana2010} and \citet{petersen2014} proposed MSMs for classes of dynamic treatment regimes that may be chosen to be compatible with positivity violations in a given context. Unfortunately, these methods suffer even more from IP weight instability when $T$ is large, as there are fewer options for weight stabilization \citep{petersen2014}.

When the number of timepoints is high, \textit{history-restricted} MSMs \citep{hr_msm2007} (HR-MSMs), can be applied to instead model {\color{black} $\mathbb{E}[Y_t(\overline{A}_{t-\Gamma}, \boldsymbol{a}_{\Gamma, t})]$} for a window of $\Gamma$ timepoints  $\boldsymbol{a}_{\Gamma,t} = (a_{t-\Gamma+1},\ldots,a_t)$, typically with $\Gamma \ll t$. {\color{black} Note that we adopt bold-faced symbols for this window of interventions to distinguish from a \textit{full history} of treatment values (for which we use overbars).}
That is, 
HR-MSMs model the mean counterfactual outcome at time $t$, under an intervention defined on a proximal (often short) treatment sequence. 
As {\color{black}$\mathbb{E}[Y_t(\overline{A}_{t-\Gamma}, \boldsymbol{a}_{\Gamma, t})] = \mathbb{E}[Y_t(\boldsymbol{a}_{\Gamma, t})]$}, by a consistency assumption, these estimands implicitly marginalize over the observed treatment sequence, $\overline{A}_{t-\Gamma}$, prior to the first point of intervention.
We emphasize that Markov-like assumptions are made with our {\color{black} approach} only if they follow directly from the experimental design: in general, HR-MSMs (and thus our proposed methods) allow for $X_t$, $Y_t$ to be causally affected by \textit{all} prior {\color{black} timepoints} (i.e., $\mathcal{O}_j$ for $j \in [t-1]$). 
By placing structure on $\mathbb{E}[Y_t(\boldsymbol{a}_{\Gamma, t})]$,
the HR-MSM can borrow strength across treatment sequences $\boldsymbol{a}_{\Gamma, t}$, which can increase power. Figure~\ref{fig:estimands}E provides a graphical illustration of the causal effects targeted in HR-MSMs. As with traditional MSMs, HR-MSMs are typically fit by using an IP weighted GEE. Conceptually, IPW constructs a pseudo-population in which, under standard causal assumptions, sample averages target mean counterfactuals \citep{hernan_causal_2023}. Conversely, standard regression techniques generally yield biased estimates of causal effects as 1) failure to condition on time-varying confounders, $X_t$, biases estimates, as treatment is randomized conditional on $X_t$ in closed-loop designs, but 2) conditioning on $X_t$ induces confounding, as $X_t$ are colliders on the path between past treatments and subsequent outcomes, through unmeasured confounders, $U$, as shown in the DAG in  Figure~\ref{fig:estimands}E \citep{hernan_causal_2023}. 
HR-MSMs can also incorporate time-varying effect modifiers to test, for example, whether causal effects vary across {\color{black} timepoints}, or covariate levels \citep{hr_msm_review}. While the names \textit{history-restricted} and \textit{history-adjusted} are often used interchangeably in the literature, we adopt the former, following the convention described in \cite{guo2021}. Machine learning methods including causal transformers \citep{melnychuk2022}, counterfactual recurrent networks \citep{bica2020}, and recurrent marginal structural networks \citep{lim2018} are comparable to HR-MSMs that condition on all of $H_{t-\Gamma+1}$. These methods target effects of static treatment sequences and require a positivity assumption, and thus cannot be applied in closed-loop designs. They also do not provide tools for statistical inference.

In stochastic designs, HR-MSMs can be used to estimate the causal effect of specific deterministic treatment sequences $\boldsymbol{a}_{\Gamma,t}$, even if they differ from the observed sequence $\overline{A}_t$ close to {\color{black} timepoint} $t$, as long as they are compatible with the experimental treatment rule (``policy''). Importantly, this enables estimation of interpretable causal parameters, such as the effect of treatment on the most recent {\color{black} timepoint}, $\mathbb{E} \left [ Y_t(a_{t} = 1)  - Y_t(a_{t} = 0) \right ]$. 
These causal contrasts have grown popular recently in the analysis of mobile health studies \citep{Boruvka2018}, where they are referred to as ``excursion effects.'' \cite{guo2021} describe excursion effects as ``1)
a contrast between the distributions of the potential outcomes under two `time-varying treatments [regimes] occurring over an interval of time extending into the future' that deviate from the treatment protocol; and 2) a contrast that is `marginal over prior treatment assignments'.'' However, current methods are restricted to estimating excursion effects for the $\Gamma = 1$ case in experimental designs like ours, and thus preclude estimation of effects defined over longer intervals $\Gamma>1$ (e.g., the micro longitudinal effects in Figures \ref{fig:estimands} and \ref{fig:estimands_appendix}) {\color{black} for \textit{general} treatment sequences. \textit{Reference distributions} have been proposed and employed in the mobile health excursion effect literature \citep{dempsey2020, shi2022, shi2024}, where a sequence of interventions could occur \text{after} the point of excursion from the observational regime. For example, one might consider a contrast of active treatment versus control at a given timepoint, but then fix treatment in both groups to the control condition for a number of timepoints until the outcome is measured. While these treatment sequences may be longer than one timepoint, the resulting estimand is a contrast at only a single timepoint, whereas we consider here models that can capture contrasts for arbitrary \textit{sequences} of interventions (e.g., two values of $\boldsymbol{a}_{\Gamma, t}$) that may differ at multiple timepoints.}
Mobile health studies often include treatment rules with positivity violations. For instance, treatment may be withheld in certain cases  due to ethical or practical constraints (e.g., no delivery of phone notifications while driving). \cite{Boruvka2018} use the notation that treatment is withheld when the time-varying ``availability'' indicator, $I_t$, equals zero. Similarly, in ``closed-loop'' optogenetics experiments, $I_t = 1$ when the conditions are met such that neural manipulation may occur (e.g., when the animal approaches the lever in the example in Section~\ref{sec:intro}). There have been proposals for methods intended to account for such implied positivity violations \citep{Neugebauer_dynamic, Boruvka2018, Qian2021}, such as the availability-conditional estimand \citep{Boruvka2018}: $\mathbb{E} \left [ Y_t(a_{t} = 1)  - Y_t(a_{t} = 0) \mid I_t = 1 \right ]$. However, estimands proposed for these settings are defined only for $\Gamma = 1$. Thus, in the presence of these positivity violations, there are no causal inference methods for longer proximal treatment sequences.  {\color{black} \cite{murphy2001} proposes the broad class of dynamic treatment regimes for MSMs. No work has proposed, however, using specific \textit{sequential} dynamic regimes that account for positivity violations to enable estimation of contrasts of treatment sequences in the HR-MSM setting. Moreover, they have not proposed appropriate estimands, or derived corresponding theory and estimators.} Finally, in some mobile health applications, treatment probabilities may be unknown or recorded with error \citep{shi2023}, and thus robust methodologies with guarantees under treatment model misspecification are needed.

\section{Proposed estimands}\label{sec:prop_estimands}
To fill the gaps identified above, we propose HR-MSMs for proximal sequences of dynamic treatment regimes, designed to be compatible with treatment availability restrictions in optogenetics designs. These estimands can incorporate time-varying effect modifiers, and are defined for any $\Gamma \geq 1$, enabling dissection of intricate effects of treatment \textit{sequences}.

%
\subsection{HR-MSMs for Static Treatment Regimes}\label{sec:hr-msm_static}
For fixed $\Gamma \in \mathbb{N}$, and any $t \geq \Gamma$, let $\boldsymbol{a}_{\Gamma, t} \equiv (a_{t-\Gamma + 1}, \ldots, a_t) \in \{0,1\}^\Gamma$ denote a putative proximal treatment sequence prior to $t$. The counterfactual outcome $Y(\boldsymbol{a}_{\Gamma, t}) \equiv Y(\overline{A}_{t - \Gamma}, \boldsymbol{a}_{\Gamma, t})$ represents the outcome that would occur under an intervention that leaves the natural treatment values $A_1, \ldots, A_{t - \Gamma}$ as they occur in the observed data, then sets treatment according to the deterministic, static treatment sequence $\boldsymbol{a}_{\Gamma, t}$ for the $\Gamma$ timepoints leading up to $t$. To probe the effects of $\boldsymbol{a}_{\Gamma, t}$, one can adopt a working HR-MSM
$m(t, \boldsymbol{a}_{\Gamma, t}, V_{t - \Gamma + 1}; \boldsymbol{\beta}) \approx \mathbb{E}(Y_t(\boldsymbol{a}_{t, \Gamma}) \mid V_{t - \Gamma + 1})$,
which places structure on the means of the counterfactuals of interest, conditional on a set of effect modifiers $V_{t - \Gamma +1} \subseteq H_{t - \Gamma + 1}$ that are measured prior to treatment at $t-\Gamma+1$. We extend these models to account for positivity violations that commonly occur in neuroscience applications. Nonetheless, the above static regime HR-MSMs represent an important special case, and we will see in Section~\ref{sec:ident-est} that the coefficients $\boldsymbol{\beta}$ can typically be estimated in open-loop optogenetic designs. We discuss these models further below, since the proposed dynamic regime HR-MSMs coincide with the static formulation when there are no availability issues (i.e., $I_t = \mathds{1}(\mathbb{P}[A_t = 1 \mid H_t] > 0) \equiv 1$). 

\subsection{HR-MSMs for Dynamic Treatment Regimes}\label{sec:hr-msm}
Adopting the notation from \cite{Boruvka2018}, we define $I_t \coloneqq \mathds{1}(\mathbb{P}[A_t = 1 \mid H_t] > 0)$ as an ``availability indicator'', i.e., $I_t = 0$ if and only if active treatment (e.g., laser stimulation) is prohibited by design.  Define
$\mathcal{D}_t = \{d_t : \mathcal{H}_t \to \{0,1\} \mid d_t(H_t) = 0 \text{ if } I_t = 0\}$, for any $t$,
to be the class of treatment rules at time $t$ compatible with $I_t$. We consider the deterministic rules $\mathcal{D}_t^* = \{d_t^{(0)}, d_t^{(1)}\} \subset \mathcal{D}_t$, where
$d_t^{(0)} \equiv 0, \, d_t^{(1)} \equiv I_t$. In words, $d_t^{(0)}$ fixes 
$A_t = 0$, and $d_t^{(1)}$ sets $A_t$ equal to  $I_t$; {\color{black} similar definitions were (conceptually) introduced in~\citet{Boruvka2018} to interpret single timepoint availability-conditional estimands.} 
The treatment rules $d_t^{(0)}, d_t^{(1)} \in \mathcal{D}_t$ are two close analogues to the static regimes $a_t = 0$, $a_t = 1$. Importantly, the causal effects of $d_t^{(0)}, d_t^{(1)}$ on subsequent outcomes are identifiable under positivity violations inherent to closed-loop designs, unlike the effects of $a_t = 0$, $a_t = 1$. We can combine these time-specific rules to construct multiple time-point analogs of excursion effects compatible with availability restrictions. Specifically, for $\Gamma \in \mathbb{N}$, we let $\overline{\mathcal{D}}_{\Gamma, t}$ be a subset of $\mathcal{D}_{t - \Gamma + 1}^* \times \cdots \times \mathcal{D}_t^*$, taking
$\boldsymbol{d}_{\Gamma, t} = (d_{t - \Gamma + 1}, \ldots, d_{t}) \in \overline{\mathcal{D}}_{\Gamma, t}$ to be a sequence of $\Gamma$ treatment rules for {\color{black} timepoints} $j \in \{t-\Gamma+1,...,t\}$. The counterfactual outcome under this sequence is defined as
\begin{equation}\label{eq:counterfactual}
    Y_t(\boldsymbol{d}_{\Gamma, t}) = Y_t(A_1, \ldots, A_{t - \Gamma}, d_{t - \Gamma + 1}(H_{t - \Gamma + 1}), \ldots, d_t(H_t(\boldsymbol{d}_{\Gamma-1, t-1}))).
\end{equation}
That is, $Y_t(\boldsymbol{d}_{\Gamma, t})$ is potential outcome under an intervention that leaves the observed treatment sequence for the first $t - \Gamma$ {\color{black} timepoints}, then sequentially sets treatment by applying $d_{t - \Gamma + j}$ to $H_{t - \Gamma + j}(\boldsymbol{d}_{j - 1, t - \Gamma + j - 1})$, for $j \in [\Gamma]$, where $\boldsymbol{d}_{j - 1, t - \Gamma + j - 1} = (d_{t - \Gamma + 1}, \ldots, d_{t - \Gamma + j - 1})$.

Letting $V_{t} \subseteq H_t$ be a set of effect modifiers at {\color{black} timepoint} $t$, we seek to estimate $\mathbb{E}[Y_t(\boldsymbol{d}_{\Gamma, t}) \mid V_{t - \Gamma + 1}]$, the counterfactual mean outcome, conditional on effect modifiers that are observed \textit{before} the treatment decision of {\color{black} timepoint} $t-\Gamma + 1$. By construction, these estimands are identifiable under standard causal assumptions (see Section~\ref{sec:ident-est}). The interpretation of these effects is somewhat subtle, and warrants further discussion. When $\Gamma = 1$, we can express a contrast of these estimands in terms of the effect of exposure in a certain subgroup:
\[\mathbb{E}[Y_t(d_t^{(1)}) \mid V_t] - \mathbb{E}[Y_t(d_t^{(0)}) \mid V_t]
            = \mathbb{E}[Y_t(a_t = 1) - Y_t(a_t = 0) \mid V_t, I_t = 1] \mathbb{P}[I_t = 1 \mid V_t].\]
That is, the mean contrast in counterfactual outcomes for $d_t^{(1)}$ versus $d_t^{(0)}$ is the mean effect of $A_t = 1$ versus $A_t = 0$ among those with $I_t = 1$---the availability-conditional estimand proposed by~\citet{Boruvka2018}---diluted by the probability of availability. Even when $\Gamma = 1$, it may not always be clear for whom the availability-conditional estimand generalizes to (i.e., the group $I_t = 1$ may be highly idiosyncratic and not of particular interest). Conversely, our proposed estimands summarize the effects of plausible interventions on the whole population, acknowledging that treatment, $A_t = 1$, is not always possible.

The comparison just described is somewhat akin to the duality in clinical trials of per-protocol (or complier-specific) effects, and intention-to-treat effects. Thus, when $\Gamma = 1$, we recommend assessing both the availability-conditional estimand, as in \citet{Boruvka2018}, as well as our proposed population-level effect. When $\Gamma > 1$, it is not clear whether an analogous availability-conditional estimand exists; our approach is viable for arbitrary $\Gamma$. In general, our estimands have the population-level (possibly conditional on effect modifiers) interpretation of summarizing how outcomes would be affected if the experimental protocol were changed to match $\boldsymbol{d}_{\Gamma, t}$ for the $\Gamma$ timepoints leading up to the outcome. Finally, these estimands are dependent on the treatment protocol~\citep{guo2021}, as it is for all approaches that marginalize over observed treatment values prior to the intervention. 

For ``closed-loop'' designs, the proposed estimands are, by construction, contrasts between sequences of treatment \textit{opportunities}, i.e., $d_{t}^{(1)}(H_t)$  is 1 if and only if $I_t = 1$. This is distinct from static regime sequences where, for instance, $a_t = 1$ always sets treatment to 1 at $t$. While the latter set of contrasts have a cleaner interpretation, we use the former approach since effects of static sequences are not identified under the positivity violations inherent to closed-loop designs. Importantly, the proposed estimands satisfy a sharp null preservation property: when the \textit{treatment} has no effect in any individuals, contrasts under dynamic treatment sequences are also null. We now formalize this property.

\begin{proposition}\label{prop:sharp-null}
    For fixed $\Gamma$, suppose consistency (i.e., Assumption~\ref{ass:consistency}) holds, as well as the sharp null hypothesis that $Y_t(\boldsymbol{a}_{\Gamma, t}) = Y_{t}(\boldsymbol{a}_{\Gamma, t}')$, for any $\boldsymbol{a}_{\Gamma, t}, \boldsymbol{a}_{\Gamma, t}' \in \{0,1\}^{\Gamma}$. Then $\mathbb{E}(Y_{t}(\boldsymbol{d}_{\Gamma, t})) \equiv \mathbb{E}(Y_t)$ for all $\boldsymbol{d}_{\Gamma, t} \in \overline{\mathcal{D}}_{\Gamma, t}$.
\end{proposition}
The proof of this result and all others in this paper can be found in the Appendices.

This result shows that the proposed estimands provide a means of null hypothesis testing for the causal effect of the \textit{treatment} itself, despite being defined in terms of treatment \textit{opportunities}. To provide wider insight into what the proposed contrasts capture---in the general setting where the sharp null does not hold---we express the estimands in terms of observed variables for a couple of special cases in Appendix Section~\ref{app:sharp-null}. 

When $\Gamma > 1$, 
there may be many potential treatment rule sequence combinations. We thus propose to estimate intervention effects with an HR-MSM on the (conditional) means of the counterfactuals~\eqref{eq:counterfactual}:
$m(t, \boldsymbol{d}_{\Gamma, t}, V_{t - \Gamma + 1}; \boldsymbol{\beta}) \approx \mathbb{E}[Y_t(\boldsymbol{d}_{\Gamma, t}) \mid V_{t - \Gamma + 1}]$,
where $m$ is a fixed known function. We aim to conduct inference on the HR-MSM parameters, $\boldsymbol{\beta}$, but do not assume the model is necessarily well-specified; we treat $\boldsymbol{\beta}$ as projections onto the working model $m$~\citep{neugebauer2007, rosenblum2010}: {\small
\begin{equation}\label{eq:msm-proj}
    \boldsymbol{\beta}_0 = \mathrm{arg\,min}_{\boldsymbol{\beta} \in \mathbb{R}^q} \sum_{t = \Gamma}^T \ \sum_{\boldsymbol{d}_{\Gamma, t} \in \overline{\mathcal{D}}_{\Gamma,t}}\mathbb{E}\left(h(t, \boldsymbol{d}_{\Gamma, t}, V_{t - \Gamma + 1}) \left\{Y_t(\boldsymbol{d}_{\Gamma, t}) - m(t, \boldsymbol{d}_{\Gamma, t}, V_{t - \Gamma + 1}; \boldsymbol{\beta})\right\}^2\right),
\end{equation}}%
for some fixed non-negative weight function $h$. The projection parameter formulation shares close connections with the recent literature on ``assumption-lean inference'' \citep{vansteelandt2022}. By defining a valid summary measure in terms of the model $m$, it strikes a balance between an exclusively parametric approach in which one would assume $m$ is correctly specified, and an entirely nonparametric formulation which would not impose strong structure on the causal quantities of interest. By construction, the parameter $\boldsymbol{\beta}_0$ is the parameter that characterizes the ``best fit'' possible with the working model $m$, in terms of $L_2(\mathbb{P})$-distance to the target causal quantities. This framing also relates to robustness and efficiency. When the working model is correctly specified, the projection-based approach will yield valid inference for the true MSM parameters, but may sacrifice efficiency compared to an approach that assumed the correct model. On the other hand, the latter approach may be biased when the model is incorrectly specified, whereas the projection-based approach remains valid. See~\citet{kennedy2019, kennedy2023} for related commentary. 

\section{Identification and Estimation}\label{sec:ident-est}


We now describe the causal assumptions under which the proposed estimands are identified. We then develop an IPW estimator of the HR-MSM parameters, and derive its asymptotic properties. While we use notation for dynamic regime HR-MSMs here, the results apply equally to static regime HR-MSMs  where there are no availability issues (i.e., $I_t \equiv 1$). In that instance, the treatment rules $\boldsymbol{d}_{\Gamma, t}$ reduce to a corresponding static sequence $\boldsymbol{a}_{\Gamma, t}$.

For each $t$, we define $\pi_t(a; H_t) \coloneqq \mathbb{P}[A_t = a \mid H_t]$, and make the following assumptions.

\begin{assumption}\label{ass:consistency}
    Consistency: $Y_t(\boldsymbol{d}_{\Gamma, t}) = Y_t$, whenever $A_j = d_j(H_j)$, for all $j$
\end{assumption}

\begin{assumption}\label{ass:positivity}
    Positivity: For all $t \in \{\Gamma, \ldots, T\}$, $d_t \in \mathcal{D}_{t}^*$, $\pi_t(d_t(H_t); H_t) \geq \epsilon$, w.p. 1
\end{assumption}

\begin{assumption}\label{ass:NUC}
    Sequential randomization: $A_t \independent (Y_s(\boldsymbol{d}_t), X_{s + 1}(\boldsymbol{d}_t), A_{s+1}(\boldsymbol{d}_t))_{s = t}^T \mid H_t$, for all $t \in \{\Gamma, \ldots, T\}$, and all $\boldsymbol{d}_{t} = (d_1, \ldots, d_t)$
\end{assumption}

Consistency (Assumption~\ref{ass:consistency}) states that for any of the regimes $\boldsymbol{d}_{\Gamma, t}$ under study, the counterfactual outcome $Y_t(\boldsymbol{d}_{\Gamma, t})$ equals the observed outcome $Y_t$ when observed treatment values correspond to assignment under $\boldsymbol{d}_{\Gamma, t}$. Positivity (Assumption~\ref{ass:positivity}) states that treatment probabilities are bounded away from zero---this is required for the asymptotic analysis of the proposed estimator later on. By definition of the availability indicator $I_t$, and the regimes $\mathcal{D}_t^*$ in Section~\ref{sec:hr-msm}, we are allowing $\mathbb{P}[A_t = 1 \mid H_t] = 0$ in some cases (i.e., when $I_t = 0$), but Assumption~\ref{ass:positivity} rules out $\mathbb{P}[A_t = 1 \mid H_t] = 1$. This positivity assumption holds in many open- and closed-loop optogenetic studies. In practice, in such experiments, one can ensure that Assumption~\ref{ass:positivity} holds by design when choosing the treatment assignment probabilities. Finally, Assumption~\ref{ass:NUC} says that treatments are randomly assigned at each time $t$, based on all previously measured data $H_t$. In the sequential optogenetic experiments motivating this work, this assumption would hold by design. In observational studies, one will have to assess the plausibility of Assumption~\ref{ass:NUC} (as well as Assumptions~\ref{ass:consistency} and \ref{ass:positivity}) on a case-by-case basis, ideally based on subject matter knowledge; it may be harder to justify Assumption~\ref{ass:NUC} due to the possible presence of unmeasured confounders. The following result says that these three causal assumptions are sufficient for identification of the counterfactual means $\mathbb{E}[Y_t(\boldsymbol{d}_{\Gamma, t}) \mid V_{t - \Gamma + 1}]$, and of the HR-MSM parameters $\boldsymbol{\beta}_0$.

\begin{proposition}\label{prop:ident}
    Under Assumptions~\ref{ass:consistency}--\ref{ass:NUC}, we have
    \[\mathbb{E}(Y_t(\boldsymbol{d}_{\Gamma, t}) \mid V_{t - \Gamma + 1}) =  \mathbb{E}_{\mathbb{P}}\left(\prod_{j = t - \Gamma + 1}^t \frac{\mathds{1}(A_j = d_j(H_j))}{\pi_t(A_t; H_t)} Y_t \, \middle| \, V_{t - \Gamma + 1}\right).\]
    Moreover, assuming the solution to \eqref{eq:msm-proj} is unique, and the working model $m$ is differentiable in $\boldsymbol{\beta}$, the MSM parameters $\boldsymbol{\beta}_0$ are identified through the estimating equation
    \begin{align*}\boldsymbol{0}
        &= \mathbb{E}_{\mathbb{P}}\bigg(\sum_{t = \Gamma}^T \, \, \, \, \sum_{\boldsymbol{d}_{\Gamma, t} \in \overline{\mathcal{D}}_{\Gamma, t}} h(t, \boldsymbol{d}_{\Gamma, t}, V_{t - \Gamma + 1}) M(t, \boldsymbol{d}_{\Gamma, t}, V_{t - \Gamma + 1}; \boldsymbol{\beta}_0) \\
        & \quad \quad \quad \quad \quad \quad \quad \quad \quad \quad \times \left[\prod_{j = t - \Gamma + 1}^t \frac{\mathds{1}(A_j = d_j(H_j))}{\pi_t(A_t; H_t)} \right]\left\{Y_t - m(t, \boldsymbol{d}_{\Gamma, t}, V_{t - \Gamma + 1}; \boldsymbol{\beta}_0)\right\}\bigg),
        \end{align*}
     where $M(t, \boldsymbol{d}_{\Gamma, t}, V_{t - \Gamma + 1}; \boldsymbol{\beta}) = \nabla_{\boldsymbol{\beta}}\,{m}(t, \boldsymbol{d}_{\Gamma, t}, V_{t - \Gamma + 1}; \boldsymbol{\beta})$.
\end{proposition}

The result of Proposition~\ref{prop:ident} is a population inverse probability-weighted estimating equation for the target parameters $\boldsymbol{\beta}_0$. This estimating equation motivates the following IPW estimator: define $\widehat{\boldsymbol{\beta}}$ to be the solution to the empirical IPW estimating equation,
    \begin{align*}
        \boldsymbol{0}
        &= \mathbb{P}_{n}\bigg(\sum_{t = \Gamma}^T \, \, \, \, \sum_{\boldsymbol{d}_{\Gamma, t} \in \overline{\mathcal{D}}_{\Gamma, t}} h(t, \boldsymbol{d}_{\Gamma, t}, V_{t - \Gamma + 1}) M(t, \boldsymbol{d}_{\Gamma, t}, V_{t - \Gamma + 1}; \widehat{\boldsymbol{\beta}}) \\
        & \quad \quad \quad \quad \quad \quad \quad \quad \quad \quad \times \left[\prod_{j = t - \Gamma + 1}^t \frac{\mathds{1}(A_j = d_j(H_j))}{\pi_t(A_t; H_t)} \right]\left\{Y_t - m(t, \boldsymbol{d}_{\Gamma, t}, V_{t - \Gamma + 1}; \widehat{\boldsymbol{\beta}})\right\}\bigg),
    \end{align*}
where $\mathbb{P}_n(f) = \frac{1}{n}\sum_{i =1}^n f(\mathcal{O}_{1,i}, \ldots, \mathcal{O}_{T, i})$ is the empirical mean of $f$. In optogenetics studies, the propensity scores $\pi_t$ are known by design, but we consider incorporation of outcome modeling, and propensity score estimation when they are unknown in Appendix~\ref{sec:mr}. 


In the following result, we prove asymptotic normality of the proposed estimator $\widehat{\boldsymbol{\beta}}$, under mild conditions. We use some additional notation: let $Z_i = \{\mathcal{O}_{t, i}\}_{t = 1}^T$ to be the totality of data observed on subject $i$, and define the estimating function $\phi(Z, \cdot) : \mathbb{R}^q \to \mathbb{R}^q$ via
\begin{align*}\phi(Z, \boldsymbol{\beta}) &= \sum_{t = \Gamma}^T \, \, \, \, \sum_{\boldsymbol{d}_{\Gamma, t} \in \overline{\mathcal{D}}_{\Gamma, t}} h(t, \boldsymbol{d}_{\Gamma, t}, V_{t - \Gamma + 1}) M(t, \boldsymbol{d}_{\Gamma, t}, V_{t - \Gamma + 1}; \boldsymbol{\beta}) \\
        & \quad \quad \quad \quad \quad \quad \quad \quad \quad \quad \times \left[\prod_{j = t - \Gamma + 1}^t \frac{\mathds{1}(A_j = d_j(H_j))}{\pi_t(A_t; H_t)} \right]\left\{Y_t - m(t, \boldsymbol{d}_{\Gamma, t}, V_{t - \Gamma + 1}; \boldsymbol{\beta})\right\}.
        \end{align*}
With this notation, we note that $\widehat{\boldsymbol{\beta}}$ solves $\mathbb{P}_n[\phi(Z, \widehat{\boldsymbol{\beta}})] = \boldsymbol{0}$. Further, we define $\boldsymbol{A}(\boldsymbol{\beta}) = \mathbb{E}[\phi(Z, \boldsymbol{\beta})\phi(Z, \boldsymbol{\beta})^T]$ and $\boldsymbol{B}(\boldsymbol{\beta})= \mathbb{E}[\nabla_{\boldsymbol{\beta}} \, \phi(Z, \boldsymbol{\beta})]$.

\begin{theorem}\label{thm:asymp}
     Suppose Assumptions~\ref{ass:consistency}--\ref{ass:NUC} and the following conditions hold:
    (i) The minimizer $\boldsymbol{\beta}_0$ in \eqref{eq:msm-proj} is unique;
        (ii) $m(t, \boldsymbol{d}_{\Gamma, t}, V_{t - \Gamma + 1}; \boldsymbol{\beta})$ is continuously differentiable at $\boldsymbol{\beta}_0$, uniformly in $V_{t - \Gamma + 1}$;
        (iii) in a neighborhood around $\boldsymbol{\beta}_0$, $\boldsymbol{A}(\boldsymbol{\beta})$ and $\boldsymbol{B}(\boldsymbol{\beta})$ are finite-valued, and $\boldsymbol{B}(\boldsymbol{\beta})$ is non-singular;
        (iv) $\widehat{\boldsymbol{\beta}} \overset{p}{\to} \boldsymbol{\beta}_0$.
    Then
    $\sqrt{n}(\widehat{\boldsymbol{\beta}} - \boldsymbol{\beta}_0) \overset{d}{\to} \mathcal{N}(\boldsymbol{0}, \boldsymbol{V}(\boldsymbol{\beta}_0))$,
    where $\boldsymbol{V}(\boldsymbol{\beta}) = \boldsymbol{B}(\boldsymbol{\beta})^{-1} \boldsymbol{A}(\boldsymbol{\beta})\boldsymbol{B}(\boldsymbol{\beta})^{-1}$.
\end{theorem}

Theorem~\ref{thm:asymp} gives the asymptotic distribution of the estimator $\widehat{\boldsymbol{\beta}}$. Conditions (i) through (iv) are standard conditions for asymptotic normality of M-estimators \citep{huber1964, huber1967}. For condition (ii), we expect the working model $m$ to be differentiable in $\boldsymbol{\beta}$ for most common models. Condition (iii) is satisfied under mild conditions, e.g., if the weight functions $h$, the model $m$ and its derivative $M$, and the outcomes $Y_t$ are uniformly bounded, and no haphazard degeneracy in $\boldsymbol{B}$ exists that could cause singularity. Lastly, condition (iv) is also quite weak, only requiring convergence of $\widehat{\boldsymbol{\beta}}$ at an arbitrarily slow rate, and would hold under some stochastic equicontinuity conditions \citep{newey1991, pollard2012}.

Importantly, Theorem~\ref{thm:asymp} provides a strategy to construct asymptotically valid Wald-based confidence intervals (CIs) for the MSM parameters $\boldsymbol{\beta}_0$. Namely, for any $\boldsymbol{\beta}$ we can take
$\widehat{\boldsymbol{A}}(\boldsymbol{\beta}) = \mathbb{P}_n[\phi(Z, \boldsymbol{\beta})\phi(Z, \boldsymbol{\beta})^T]$, and $\widehat{\boldsymbol{B}}(\boldsymbol{\beta})= \mathbb{P}_n[\nabla_{\boldsymbol{\beta}} \, \phi(Z, \boldsymbol{\beta})]$,
and define $\widehat{\boldsymbol{V}} = \widehat{\boldsymbol{B}}(\widehat{\boldsymbol{\beta}})^{-1} \widehat{\boldsymbol{A}}(\widehat{\boldsymbol{\beta}})\widehat{\boldsymbol{B}}(\widehat{\boldsymbol{\beta}})^{-1}$, which is consistent for $\boldsymbol{V}(\boldsymbol{\beta}_0)$. Then, for $j \in [q]$, an $(1 - \alpha)$-level CI for $\beta_{j,0}$ is given by
$\widehat{\beta}_j \pm z_{1 - \alpha / 2}\sqrt{\frac{\widehat{V}_{jj}}{n}}$,
where $z_{1 - \alpha/2}$ is the $(1 - \alpha/2)$-quantile of the standard normal distribution, and $\widehat{V}_{jj}$ is the $j$-th diagonal element of $\widehat{\boldsymbol{V}}$. More generally, CIs for any linear combination of $\boldsymbol{\beta}$ parameters can be constructed similarly.
{\color{black} \cite{murphy2001} provides identifiability and asymptotic normality results (akin to Proposition~\ref{prop:ident} and Theorem~\ref{thm:asymp}) for ``all timepoint'' dynamic regime MSMs, which we extend to the history-restricted case. Future work might also generalize Theorem~\ref{thm:asymp} to the large $T$, small $n$ settings common in optogenetic applications (e.g., see \citep{shi2024_metalearning}).}
%

We now highlight simplifications that arise in randomized static designs with no availability issues (i.e., $I_t \equiv 1$), $\boldsymbol{a}_{\Gamma, t} \in \{0,1\}^\Gamma$. First, the sum over $\boldsymbol{d}_{\Gamma, t} \in \overline{\mathcal{D}}_{\Gamma, t}$ picks out exactly one sequence $\boldsymbol{a}_{\Gamma, t} = (A_{t - \Gamma + 1}, \ldots, A_t) \eqqcolon \boldsymbol{A}_{\Gamma, t}$:
{\small
\[
        \boldsymbol{0}
        = \mathbb{P}_{n}\left(\sum_{t = \Gamma}^T  \ h(t, \boldsymbol{A}_{\Gamma, t}, V_{t - \Gamma + 1}) M(t, \boldsymbol{A}_{\Gamma, t}, V_{t - \Gamma + 1}; \widehat{\boldsymbol{\beta}}) \frac{Y_t - m(t, \boldsymbol{A}_{\Gamma, t}, V_{t - \Gamma + 1}; \widehat{\boldsymbol{\beta}})}{\prod_{j = t - \Gamma + 1}^t \pi_t(A_t; H_t)}\right).
\]
}%
    This estimator takes the form of a weighted GEE, regressing outcomes on the model $m$. We adopt a working independence correlation structure~\citep{liang1986}, 
    as alternative correlation structures may lead to bias~\citep{tchetgen2012}. 

Second, if, as in a marginally randomized (open-loop) design, the (known) treatment probabilities $\pi_t$ depend only on $A_t$ (and possibly $V_{t - \Gamma + t}$), then the weight function $h$ can be chosen to include terms that match $\prod_{j = t-\Gamma + 1}^t \pi_t$ exactly, leading to the following GEE:
\[    \boldsymbol{0}
        = \mathbb{P}_{n}\left(\sum_{t = \Gamma}^T  \ \widetilde{h}(t, \boldsymbol{A}_{\Gamma, t}, V_{t - \Gamma + 1})M(t, \boldsymbol{A}_{\Gamma, t}, V_{t - \Gamma + 1}; \widehat{\boldsymbol{\beta}}) \left\{Y_t - m(t, \boldsymbol{A}_{\Gamma, t}, V_{t - \Gamma + 1}; \widehat{\boldsymbol{\beta}})\right\}\right),\]
which also employs a working independence assumption across timepoints within a subject, and can be fit with standard software, depending on the choice of $\widetilde{h}$. In sum, when estimating static treatment regime HR-MSMs in marginally randomized designs with no availability issues, a simple unweighted regression of outcomes on the model $m$ yields valid estimates of projection parameters of the form in~\eqref{eq:msm-proj}. Nonetheless, one should use a variance estimator that accounts for the dependence of observations within a subject, such as  $\widehat{\boldsymbol{V}}$ above. We provide an implementation of our framework, described in Appendix~\ref{app:implement}.

We relegate development of a multiply robust estimator to Appendix~\ref{sec:mr}. 
This estimator involves estimation of nuisance models and is based on the nonparametric influence function of $\boldsymbol{\beta}_0$. It can achieve $\sqrt{n}$-convergence even with flexible nuisance model specifications. 

\section{Numerical Experiments}\label{sec:num-exps}
\subsection{Simulation Studies}\label{sec:sims}
\paragraph{Experimental Setup} We sought to assess performance of the proposed HR-MSM $\boldsymbol{\beta}$ estimators, and identify variance estimators that yield nominal coverage in the small $n$ settings common in optogenetics studies. 
To evaluate the accuracy of our framework in estimating mean counterfactuals, we designed the simulations such that the target estimands––contrasts of mean counterfactuals––corresponded to regression coefficients from the true HR-MSM. The data were simulated to mimic closed-loop optogenetics designs with positivity violations: we drew i) $X_0 \sim \mathrm{Bernoulli}(1/2)$; ii) $A_t \mid X_t \sim \mathrm{Bernoulli}(\frac{1}{2}X_t)$, for $t \in \{0, \ldots, T\}$; iii) $X_t \mid A_{t - 1} \sim \mathrm{Bernoulli}(0.4 + 0.4 A_{t - 1})$, for  $t \in [T]$; and iv) $   Y_t \mid X_{t - 1}, A_{t - 1}, X_t, A_t \sim \mathcal{N}(\alpha_1 X_{t - 1} + \alpha_2 A_{t - 1}  + \alpha_3 X_t + \alpha_4 A_t, \sigma_t^2),~\mbox{for}~ t \in [T]$,
%
where $(\alpha_1, \alpha_2, \alpha_3, \alpha_4) = (0.25, 2, 1.75, 0.5)$, and $\sigma_t^2 = 1$ for all $t$. These set availability indicator $I_t \equiv X_t$, for all $t$, and result in marginal probabilities $\mathbb{P}[X_t = 1] = \frac{1}{2}$, $\mathbb{P}[A_t = 1] = \frac{1}{4}$, for all $t$. We obtain a closed form for the parameters of the saturated two time-point dynamic treatment regime HR-MSM: letting $\boldsymbol{d}_{2, t} = (d_{t - 1}, d_t) \in \overline{\mathcal{D}}_{2, t}$ be arbitrary, and defining $J_{t-1} \coloneqq \mathds{1}(d_{t - 1} \equiv d_{t - 1}^{(1)})$, $J_{t} \coloneqq \mathds{1}(d_{t} \equiv d_{t}^{(1)})$, we derive in Appendix~\ref{app:hr_msm_sim} that
$\mathbb{E}(Y_t(\boldsymbol{d}_{2, t})) = \beta_0 + \beta_1 J_{t - 1} + \beta_2 J_{t} + \beta_3 J_{t - 1}J_t$,
where $\beta_0 =  0.5 \alpha_1 + 0.4 \alpha_3$, $\beta_1 = 0.5 \alpha_2 + 0.2 \alpha_3$, $\beta_2 = 0.4 \alpha_4$, and $\beta_3 = 0.2 \alpha_4$. Aggregating $\boldsymbol{\beta} = (\beta_{0}, \beta_1, \beta_2, \beta_3)$, the HR-MSM given by
\begin{equation}\label{eq:msm-sim}
    m(t, \boldsymbol{d}_{2, t}; \boldsymbol{\beta}) = \beta_{0} + \beta_1 J_{t - 1} + \beta_2 J_{t} + \beta_3 J_{t - 1}J_t
\end{equation}
is correctly specified under this data generating process, and we can evaluate the performance of the proposed estimator relative to these true values.

To show we can conduct valid inference on sequential excursion effects, we estimated the three estimands illustrated in Figure~\ref{fig:estimands}: (1) the ``blip'' effect of an additional exposure opportunity at $t$, while keeping treatment at $t-1$ fixed at the control condition ($\beta_2 = \mathbb{E}[Y_t(d_{t-1}^{(0)}, d_{t}^{(1)}) - Y_t(d_{t-1}^{(0)}, d_{t}^{(0)})]$); (2) ``effect dissipation'', comparing the effect of an exposure opportunity at one versus two {\color{black} timepoints} prior to the outcome ($\beta_2 - \beta_1 = \mathbb{E}[Y_t(d_{t-1}^{(0)}, d_{t}^{(1)}) - Y_t(d_{t-1}^{(1)}, d_{t}^{(0)})]$); and (3) the ``dose response'' curve of exposure opportunities, where the two single-opportunity sequences are averaged (the sequence $(\beta_0, ~\beta_0 + \frac{1}{2}\{\beta_1 + \beta_2\},~ \sum_{j=0}^3 \beta_j)$). This setup also illustrates how HR-MSMs are easily specified such that sequential excursion effects can be calculated as linear combinations of the $\boldsymbol{\beta}$ parameters.

Although the HR-MSM \eqref{eq:msm-sim} is correctly specified, we still estimated the $\boldsymbol{\beta}$ coefficients as projection parameters. We proceeded as if we started by defining $\boldsymbol{\beta}$ as the minimizers in \eqref{eq:msm-proj}, with $V_{t - \Gamma + 1} = \emptyset$ (i.e., no effect modifiers), and $h(t, \boldsymbol{d}_{\Gamma, t}) \equiv 1$ (i.e., constant weight function).
We applied the IPW point estimator for the HR-MSM parameters, $\widehat{\boldsymbol{\beta}}$, described in Section~\ref{sec:ident-est}, as well as the MR estimator described in Section~\ref{sec:mr}. To examine the performance of the MR estimator in an ideal scenario, we correctly specified all nuisance functions (derived in Appendix Section~\ref{app:nuisance-deriv}). For both estimators, we assessed the coverage of 95\% CIs constructed with our large sample variance estimator, 
and the sample size-adjusted \texttt{HC}, \texttt{HC2}, and \texttt{HC3} variance estimators 
\citep{sandwich_package}.
We tested performance with sample size $n \in \{6, 10, 30, 100\}$, and {\color{black} timepoints} $T \in \{10, 50, 500\}$ on 1000 simulation replicates. 

\begin{figure*}[t] 
	\centering
\includegraphics[width=0.99\linewidth]{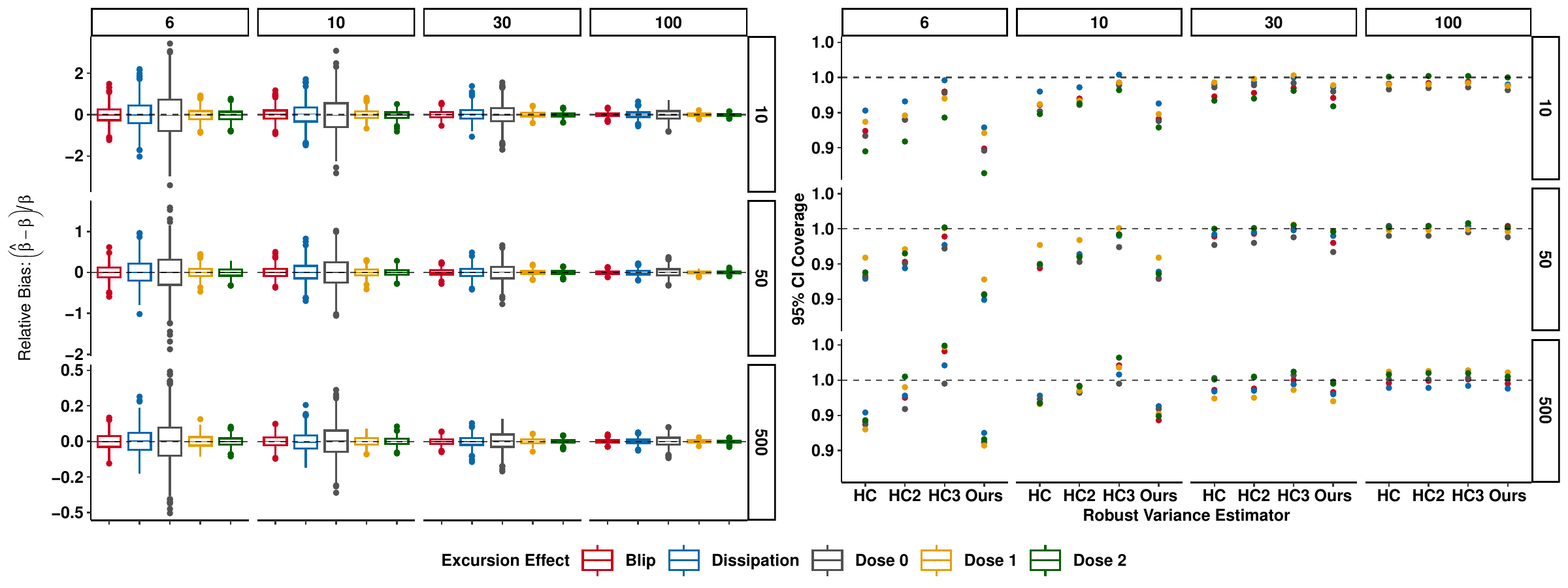}
	\caption{\footnotesize \textbf{Simulation Results} Panel columns indicate sample sizes, $n$, and panel rows indicate number of timepoints, $T$ (cluster sizes). [Left] Relative bias associated with each MSM regression coefficient. These results show that our estimator is unbiased for the target parameters. [Right] 95\% Confidence interval (CI) coverage for the marginal structural model (MSM) parameters corresponding to the sequential excursion effects. The 95\% CIs are constructed using one of three established robust variance estimators and our robust variance estimator. The nominal coverage is reached for large $n$ or large $T$ for all estimators.} 
		\label{fig:sim_results}
\end{figure*}

\paragraph{Results}
 We present the IPW results here and the MR results in Appendix~\ref{app:MR_sims}. We show results in terms of the sequential excursion effects in Figure~\ref{fig:sim_results}, and in terms of the HR-MSM $\widehat{\boldsymbol{\beta}}$ parameters in Appendix~\ref{sec:app_sims}. The Appendix also includes results from a wider range of $n$ and $T$ values. These figures show that our estimators are unbiased for the target sequential excursion effects and HR-MSM $\widehat{\boldsymbol{\beta}}$.
All CIs achieve 95\% coverage when $n$ is large. While sandwich estimators can yield small-sample bias \citep{sandwich_package}, the sample size-adjusted \texttt{HC3}-based CIs achieve close to 95\% coverage even in small $n$ and $T$ settings. Together these results show that we can conduct valid inference in the sample sizes common in neuroscience studies.



\subsection{Application: Optogenetic Study}\label{sec:application} 


\label{opto_background}

In this section, we reanalyze behavioral data from an existing study, \cite{spont_da}. In this experiment, the authors tested whether optogenetically stimulating dopamine (DA) release in the dorsolateral striatum 
while an animal engaged in a specific 
``pose'' (e.g., exploring, rearing, grooming) could ``teach'' mice to exhibit that movement more frequently.
This study was foundational in identifying the role this region plays in learning. To that end, the researchers implanted mice with optogenetics machinery, and filmed them freely-moving in a behavioral chamber. They used 
a pre-trained hidden Markov model to estimate an animal's pose in real-time.
They first measured the animals' target pose frequency on a baseline session without optogenetics. Then, on a subsequent treatment session, they applied the laser on a random subset of the target pose occurrences. They repeated this experiment for six target poses.
The experiment was carried out on animals in both the {\color{black} treatment} (optogenetics) arm, and a control arm where the laser should have no effect. 

To define {\color{black} timepoints (or what they refer to as ``trials'')}, the authors spliced the time-series of estimated pose classifications into intervals of consecutive timepoints with the same pose classification. If mice exhibited the target pose on $t$, they were considered  ``available'' for optogenetic stimulation, $I_t=1$, and were ``unavailable'' otherwise, $I_t=0$. The laser was applied ($A_t = 1$) with the dynamic policy, $\mathbb{P}(A_t = 1\mid I_t ) = 0.75I_t$. Denoting ${Y}^0_t$ and ${Y}_t$ as a binary indicator that an animal engaged in the target pose at $t$ of the \textit{baseline} and \textit{treatment} sessions, respectively, the authors estimated  $\psi = \left( \mathbb{E}[\bar{{Y}}^1 \mid G = 1] - \mathbb{E}[\bar{{Y}}^0 \mid G = 1] \right ) - \left( \mathbb{E}[\bar{{Y}}^1 \mid G = 0] - \mathbb{E}[\bar{{Y}}^0 \mid G = 0] \right )$
where $\bar{Y}^0 = \sum_{t=1}^{T_0} {Y}^0_t$, $\bar{Y}^1 = \sum_{t=1}^T {Y}_t$, and $T, T_0 \in \mathbb{N}$ are the {\color{black} number of timepoints} in treatment and baseline sessions, respectively.\footnote{The authors used a Mann Whitney U Test applied to a summary across poses but, in keeping with the mean counterfactual-based causal estimands, we describe it in terms of means (not medians) and individual poses.} 
There were $n_1=28$ and $n_0=12$ animals in the optogenetics and control arms, respectively. $T$ ranged across animals/sessions from 1207-4876, with a mean of 3612 and IQR = $[3341, 3940]$. The authors reported a (pooled across target poses) positive optogenetics treatment effect estimate akin to $\widehat{\psi}$, suggesting DA stimulation causes an increase in target pose frequency.

Beyond this, there are a number of other questions that researchers could be interested in, but cannot be probed with standard methods. Conceptualizing optogenetics like a ``study drug'', we may ask whether stimulation immediately ``taught'' the animal the target pose, or whether the treatment effect on learning had a lagged onset. Similarly, did the effect of a single stimulation persist or dissipate across {\color{black} timepoints}? Did more treatments lead to more learning monotonically, or is there an antagonistic effect or non-monotonic dose-response curve? We next demonstrate how we can address these questions with our framework.



\subsection{Application Methods}
We applied our methods to assess the questions above. Specifically, we tested the effect of specific sequences of deterministic dynamic policies, $\boldsymbol{d}_{\Gamma, t}$ (occurring on {\color{black} timepoints} $t\in \{t-\Gamma+1,\ldots,t\}$), on the mean counterfactual $\mathbb{E}[Y_t(\boldsymbol{d}_{\Gamma, t})]$. We defined the outcome $Y_t$ as an indicator that the mouse exhibited the target pose on {\color{black} timepoint} $t+2$, the next {\color{black} timepoint} on which mice could exhibit the target pose given availability ($I_t=1$) on {\color{black} timepoint} $t$. 
In Appendix Section~\ref{sec:app-preprocess}, we provide code and analysis details for the HR-MSMs below.



\begin{figure} 
	\centering
\includegraphics[width=0.99\linewidth]{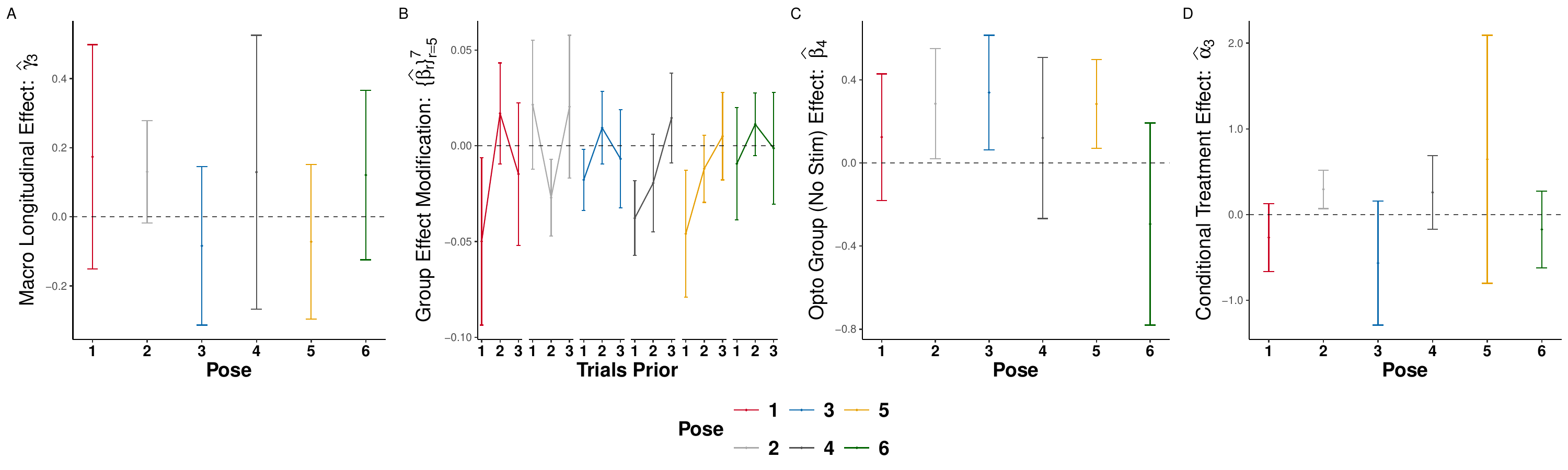}
\caption{\footnotesize
\textbf{Optogenetics Analyses.} 
Plots show coefficient estimates (error bars show 95\% CIs). Columns/colors indicate the target pose. [A] Macro summary analysis
identifies no significant effects. [B] Interaction term between $\texttt{G}$ and sequential excursion effect of a single ``dose'' occuring $r = 1,2,~\mbox{or}~3$ {\color{black} timepoints} prior to the proximal outcome that the ``dose'' occurred on. The excursion effects are significant for poses 1-5 (at, at least, one lag level). [C] Main effect of $\texttt{G}$ under a ``no-recent-treatment opportunity'' policy; this reflects the average causal effect of treatment arm among a population that has received no laser opportunities in the last $\Gamma = 3$ {\color{black} timepoints}. [D] Availability conditional estimate of interaction $\texttt{A} \times \texttt{G}$: laser $\times$ treatment arm interaction.} 		\label{fig:lag_int}
\end{figure}

\subsection{Application Results}
Our first question was whether standard methods reveal significant treatment effects when assessed with the estimands commonly tested in optogenetics studies. We applied a GEE with mean model, $\mbox{log} \left ( \mathbb{E} [\bar{Y}^s \mid G = g, S = s ] \right ) = \gamma_0 + \gamma_1 g + \gamma_2 s + \gamma_3 g \times s$,
where $S = 0,1$ indicates baseline and optogenetics sessions, respectively. 
 The estimate $\widehat{\gamma}_3$, shown in Figure \ref{fig:lag_int}A, provides a treatment effect estimate 
 for the \textit{observed} stochastic dynamic policy in \cite{spont_da}. We adopted a Poisson working model, since \cite{spont_da} analyzed $\bar{Y}^1, \bar{Y}^0 \in \mathbb{N}$, but unlike these authors, we analyzed these macro longitudinal summaries for each pose individually.
The model yielded no significant effects for any pose, in contrast with the results reported in \cite{spont_da} that analyzed data pooled across poses.
We show boxplots in Appendix Figure~\ref{fig:raw_count} of the subject-level summary $\bar{Y}^1 - \bar{Y}^0$ that is compared across arms in this model. Outcome levels are similar across arms for most poses, highlighting how macro summaries can obscure effects.

To assess an analogous local treatment effect using our method, we tested the impact of a single stimulation opportunity, and evaluated whether the effect had a lagged onset and/or dissipated across {\color{black} timepoints}. We included arm, $G$, as an effect modifier, to test whether the effect of treatment opportunity differed across arms. Setting $\Gamma = 3$, and restricting regimes to those with at most one treatment opportunity ``dose'', $\boldsymbol{d}_{3,t} \in \{(d_{t - 2}, d_{t-1}, d_t) : \sum_{j = t-2}^t \sigma_j(d_j) \leq 1\}$, where $\sigma_j(d_j) = \mathds{1}(d_j = d_j^{(1)})$, we fit the HR-MSM
\begin{equation} \label{dose_group_int}
\resizebox{.94\hsize}{!}{
    $\mbox{logit}\left ( \mathbb{E} [Y_t(\boldsymbol{d}_{\Gamma, t}) \mid  G = g] \right ) = \beta_0 + \sum_{r=0}^2 \beta_{r+1}  \sigma_{t - r}(d_{t - r}) + \beta_4 g + \sum_{r=0}^2 \beta_{5+r} g \times \sigma_{t - r}(d_{t - r}).$
    }
    \end{equation}
Thus, $\widehat{\beta}_r$ with $r \in [3]$ is an estimate of the log odds ratio comparing the mean counterfactual of $Y_t$ under a treatment sequence with a single dose (on $r = 1,2$, or $3$ {\color{black} timepoints} prior)
vs. a treatment sequence with zero dose, among control arm animals. Figure \ref{fig:estimands}B illustrates the analogous effect under a static regime.
The interaction terms, $\{\widehat{\beta}_r\}_{r=5}^7$ quantify how these causal effects of a recent treatment opportunity differ between the two arms.



The results from our model show that stimulation opportunities in the treatment arm tend to \textit{reduce} the odds of the outcome, compared to the control arm. As shown in Figure \ref{fig:lag_int}B, these effects are significantly negative for at least one lag level in five out of six target poses. In personal communications, the authors of \cite{spont_da} stated that this result appeared consistent with their finding that animal exploration increased right after stimulation (quantified as higher pose ``entropy''). The results also reveal that the laser's effect tends to dissipate across {\color{black} timepoints} in both optogenetics and control arms (shown in Appendix Figure~\ref{fig:lag-main-effects}). Figure \ref{fig:lag_int}C shows the main effect of arm under a treatment sequence of dose zero. This is essentially an estimate of the ``long-term'' effect of DA stimulation: $\widehat{\beta}_4$ is the log odds ratio of treatment arm under a ``no–recent–stimulation'' policy. We present a sensitivity analysis in Appendix~\ref{sec:gamma-sensitivity} that shows these results are stable $\Gamma$ values. 

Next, we fit the analogous model for the availability-conditional estimand  \citep{Boruvka2018} to determine whether current excursion effect methods (i.e., confined to $\Gamma = 1$) identify the same treatment effects: $\mbox{logit}\left ( \mathbb{E} [Y_t(a_t) \mid I_{t} = 1, G = g] \right ) = \alpha_0 + \alpha_1 a_t+ \alpha_2 g + \alpha_3 g \times a_t$. Figure \ref{fig:lag_int}D shows that effect estimates $\widehat{\alpha}_3$ are significant in only one pose. These results highlight how our approach can uncover a greater number of effects that are obscured when estimated with analogous availability-conditional estimands  confined to $\Gamma=1$ regimes.

\subsection{Dose-Response Excursion Effects}\label{sec:dose-resp}
\paragraph{History-Restricted MSM} We next asked whether optogenetic stimulation monotonically increases learning, or whether there can be ``too much of a good thing'' in learning. 
To test this, we fit an HR-MSM within the treatment arm ($G=1$) to estimate the causal effect of ``dose'', the number of treatment opportunities in the previous $\Gamma = 5$ {\color{black} timepoints}:


\begin{equation} \label{eq:dose_trt}
    \mbox{logit}\left ( \mathbb{E} [Y_t(\boldsymbol{d}_{\Gamma, t}) \mid  G=1] \right ) = \beta_0 + \sum_{r=1}^3 \beta_r \mathds{1} \bigg( \sum_{j = t - \Gamma + 1}^ t \sigma_j(d_j) = r \bigg),
    \end{equation}

where $\sigma_j(d_j) = \mathds{1}(d_j = d_j^{(1)})$. The coefficient $\widehat{\beta}_r$ is an estimate of the log odds ratio comparing the mean counterfactual of $Y_t$ for a treatment sequence of doses $r = 1, 2, 3$ compared to one of dose zero (see Figure \ref{fig:estimands}C for an illustration of the static regime analogue). A dose of three is the maximum feasible dose for $\Gamma = 5$, since a pose cannot occur on two consecutive {\color{black} timepoints}. Appendix Section \ref{lag_appendix} has analysis results for lagged outcomes to show that our framework can incorporate outcome sequences for general {\color{black}$\tilde{Y}^{(\Delta)}_t = f(Y_t, Y_{t+1},...,Y_{t+\Delta})$}.

The dose-response effect estimates, $\{\widehat{\beta}_r \}_{r=1}^3$, from HR-MSM \eqref{eq:dose_trt} are shown in Appendix Figure \ref{fig:dose_optoDA}A. This illustrates the capacity of our approach to identify a clear dose-response effect: within the past $\Gamma = 5$ {\color{black} timepoints}, each additional opportunity for a stimulation \textit{causes} an increase in the odds of engaging in the target pose on the next {\color{black} timepoint}. 
The effects are significant for at least one dose value in all but two target poses. In Appendix~\ref{sec:dose-resp-time}, we analyze whether these dose-response effects evolve across timepoints, $t$.

\paragraph{Conditional Excursion Effect}
We next estimate an availability-conditional estimand \citep{Boruvka2018}, to determine if existing excursion effect methods have the capacity to reveal the effects identified with our method. We estimate this in the MSM
\begin{equation} \label{eq:dose_cond}
    \mbox{logit}\left ( \mathbb{E} [Y_t(a_t) \mid I_{t} = 1, G=1] \right ) = \alpha_0 + \alpha_1 a_t.
    \end{equation}
 Appendix Figure~\ref{fig:dose_optoDA}B shows the availability-conditional treatment effect estimates, $\widehat{\alpha}_1$ estimated in model \eqref{eq:dose_cond}. It identifies no significant effects for any target pose. The conditional estimand, often referred to as a ``blip effect'' (see Figure \ref{fig:estimands}A for an illustration) is only defined for the effect of applying the laser on the most recent {\color{black} timepoint} (i.e., a dose of 1), and thus cannot estimate dose-response profiles. In contrast, our approach can test \textit{sequential} excursion effects (i.e., for policies with $\Gamma > 1$), enabling the estimation of a dose-response profile that reveals treatment effects here. Importantly, the effect estimates $\widehat{\beta}_1$ and $\widehat{\alpha}_1$ have different interpretations because $\widehat{\beta}_1$ reflects a causal effect of a single treatment opportunity for any of the last $\Gamma=5$ {\color{black} timepoints}, and $\widehat{\beta}_1$ is not interpreted as conditional on availability. These results suggest that the optogenetic stimulation altered learning monotonically. We also show the same analysis conducted in the control arm ($G=0$) in Appendix Section~\ref{sec:dose-resp-ctrl}. More generally, the above analyses show we can reliably estimate sequential excursion effects that enable testing a range of scientific questions.



Finally, in Appendix Section~\ref{sec:app-secondary}, we asked whether optogenetic stimulation increased pose learning equally across all animals, or whether improvement depended on how much they were inclined to do that pose initially.
We tested this by adding an interaction term between ``natural'' target pose frequency (the total target pose counts from the baseline session) and ``dose'' to HR-MSM \eqref{eq:dose_trt}. We found that  the interaction was significant for two poses (for at least one dose). For one pose, the laser increased target pose frequency more with high baseline pose counts, while it decreased for the other target pose. This suggests that the effects of DA depend on the movement being learned and on initial subject affinity for the pose. More broadly, this shows how our framework can incorporate effect modifiers.

\subsection{Application Conclusions}
Together these results illustrate how our proposed \textit{sequential} excursion effect framework can reveal effects that are missed by both standard ``macro'' longitudinal effects, and existing conditional one time-point excursion effect methods. 
Our finding that macro summaries 
show almost no differences between arms highlights how ``treatment–confounder'' feedback can obscure strong treatment effects in closed-loop designs, even when inspecting simple averages of observed outcomes.
Intuitively, this is because the observed mean outcomes comprise the total effect of the specific optogenetic stimulation sequences that \textit{happened} to occur in the study, mediated through (and marginalized over) effects on availability $I_t$ at subsequent {\color{black} timepoints} \citep{hernan_causal_2023}.\footnote{We provide an intuitive description of how these can dilute (or exaggerate) effects, potentially as a result of treatment–confounder feedback in Appendix Section~\ref{trt_conf_feedback}.} The latter is often considered a nuisance when interest lies in comparing different treatment sequence patterns, and is only properly accounted for with causal methods (e.g., IPW, $g$-formula \citet{robins1986}).

Sequential excursion effects, on the other hand, provide estimates of mean counterfactuals under specific deterministic dynamic treatment regimes, and thus do not suffer from these drawbacks. Indeed, our analyses reveal immediate negative effects of DA stimulation (detectable at the next {\color{black} timepoint}), and slow positive effects (i.e., in treatment relative to control arms). We also find the control arm exhibits positive, off-target effects of the laser, even though it should not directly cause neural firings. Together the opposing signs of these ``fast''/``slow'' and on/off-target causal effects may dilute the magnitude of macro summaries of the outcome (e.g., total pose counts). By enabling estimation of sequential excursion effects (i.e., $\Gamma > 1$), we can reveal effect profiles (e.g., dose-response) not possible with availability-conditional estimands whose definition is confined to $\Gamma = 1$ regimes. As we observed, the optogenetics arm sometimes exhibits an excursion effect not present in the control arm, providing analysts a tool to disentangle laser effects, caused by increased neural firing, from off-target effects. When off-target effects are not a concern, our framework enables estimation of \textit{within-arm} causal effects without having to collect data in a control arm, thereby reducing the number of animals required. Finally, by including effect-modifiers, one can analyze how causal effects evolve across {\color{black} timepoints}.

\section{Proposals for the Optogenetics Communitiy}\label{sec:recs-opto}


We recommend that investigators use open-loop experimental designs, unless the scientific question necessitates a closed-loop protocol. This is because the effects of static treatment sequences can be estimated under open-loop (static) designs, these effects are often easier to interpret, and, unlike in the closed-loop setting, one can directly estimate the effects of different optogenetic stimulation patterns without considering the availability-respecting dynamic treatment regimes proposed in this paper. Moreover, closed-loop (dynamic) experimental protocols also make between-arm contrasts inherently difficult to interpret, because subjects in each arm have diverging treatment histories that can differentially interact with subsequent treatments. Standard macro summaries, sequential excursion effects, and availability-conditional excursion effects are all affected by this phenomenon. For example, in the optogenetics arm, \textit{both} on-target and off-target effects of stimulation at each {\color{black} timepoint} are potentially interacting with treatment history. Unfortunately, this is \textit{not} accounted for by ``subtracting off'' effects through, for example, 
comparing $(d_{t-1}^{(0)},d_t^{(0)})$ vs. $(d_{t-1}^{(1)},d_t^{(1)})$ (i.e., the always–treat–when–available vs. never-treat) regimes across treatment and control arms. Similarly, the macro summaries accumulate these ``differential history-interactions'' across {\color{black} timepoints} as a byproduct of the design. 
We therefore recommend caution when interpreting 1) between-arm macro summaries, or 2) interactions between arm and excursion effects (sequential or availability-conditional)
in closed-loop designs.

We recommend designs with stochastic experimental policies, unless there is a scientific reason to implement a deterministic policy. 
This ensures that 
positivity violations are avoided, and one can estimate excursion effects. 
In the absence of design constraints, we also recommend measuring behavior on as many {\color{black} timepoints} as possible (i.e., large $T$), with as few conditions as necessary, setting $\mathbb{P}(A_t = 1\mid I_t = 1) = 0.5$ to increase power.



Next, we emphasize that $\mathbb{P}(A_t \mid I_t)$ of one subject should be set independently of data from all other subjects. In some longitudinal designs (e.g., \citet{DA_adapts}), authors have adjusted $\mathbb{P}(A_t \mid I_t)$ on each session in an effort to maintain comparable ``total stimulation doses'' (throughout the experiment) across arms. 
Setting treatment probabilities using data from any other subjects in the experiment can invalidate causal inference, as we anticipate this would induce interference across subjects in the treatment distributions. Excursion effects, on the other hand, can be parameterized to provide a ``dose-controlled'' treatment effect estimate. We believe this is more advisable than manually adjusting treatment probabilities with the hope that the final ``total dose'' is comparable across arms. Given that the effect of adjusting treatment probabilities may elicit both interference and unanticipated impacts on ``macro summaries'' for dynamic stochastic regimes (due to treatment–confounder feedback), we recommend designs that fix the treatment probability across {\color{black} timepoints}. 
We encourage investigators
to document the $\mathbb{P}(A_t \mid I_t)$ at each $t$, as
the proposed causal framework benefits from knowing these probabilities. 

While one should select $\Gamma$ on the basis of subject-matter knowledge and the scientific question, we found that setting $2 \leq \Gamma \leq 5$ was sufficient to estimate many characteristics of the treatment effect profile and yielded reasonable variance estimates in the small sample sizes common in optogenetics. Critically, one should never include time-varying variables occurring after {\color{black} timepoint} $t-\Gamma+1$ as an effect modifier in the HR-MSM. 
When time-varying experimental variables are randomized (e.g., cue-type), one could include these in the HR-MSM (even for variables measured on timepoints $\{t-\Gamma+1,...,t\}$) by defining a combination of this variable and optogenetic stimulation as a compound treatment; the extension of our methods to general discrete treatments is straightforward. We make additional interpretation and modeling recommendations in Appendix Section~\ref{app:model_interp}.

\section{Discussion}\label{sec:discuss}
We have proposed a non-parametric excursion effect {\color{black} approach} for 
optogenetics studies that {\color{black} builds upon} the conditional estimands proposed in \cite{Boruvka2018} to longitudinal policies in the presence of positivity violations. {\color{black} We leverage dynamic regimes from the MSM \citep{murphy2001} and excursion effect \citep{Boruvka2018} literatures, and extend to optogenetics settings to estimate neuroscience-relevant effects.} Our results permit any flexible HR-MSM specification, where the coefficients have a valid interpretation even without assuming correct model specification. We proposed an IPW estimator, and proved its consistency and asymptotic normality under mild assumptions. Moreover, we provide asymptotically valid standard error estimates and confidence intervals.

By leveraging the randomization in typical optogenetics experimental designs, and properly accounting for the time-varying confounding by availability status, the proposed estimands represent \textit{causal} effects and contrasts between exposure opportunity sequences, averaged over the population under study.
For this reason, the effects have causal interpretations instead of, for example, merely capturing correlations between availability status and outcomes over time (e.g., when $I_t$ captures some target behavior, and the outcome, $Y_t$ is some summary of this behavior in the future). We showed in Section~\ref{sec:hr-msm} that, when the sharp null of no causal effect of treatment (e.g., optogenetically increasing dopamine) holds, contrasts of availability-respecting estimands (for any $\Gamma$) would also be null. Thus, although effects for the availability-respecting dynamic policies should be interpreted in terms of treatment \textit{opportunities}, our framework can test for causal effects of the \textit{treatment}, itself.

In analyzing optogenetics studies with HR-MSMs, one should consider that the wider class of excursion effects are sensitive to the experimental design used to collect the data. While this has garnered some criticism within the causal inference literature \citep{guo2021}, this property may actually be desirable in neuroscience studies. Indeed, the causal effects of any specific optogenetic protocol are typically interpreted within the context of the experimental design, since the intervention is usually just a means to study the \textit{natural} (unstimulated) role of a given neural pathway in behavior. Optogenetic stimulation is not akin to a treatment a patient would actually take, and thus estimating its population effects \textit{outside} the study's context may not be of great interest. Finally, if multiple excursion effects are pre-specified as the principal target parameters, we recommend adjusting for multiple comparisons. Otherwise, we envision our framework being used for hypothesis-generating secondary analyses, where multiple comparison adjustment may be less necessary.

For applications in which treatment probabilities are unknown, or measured with error \citep{shi2023}, and sample sizes are larger, we anticipate multiply-robust analogues of our framework may be useful. 
We present an initial exploration of these approaches for the two-timepoint case, but the extension to the more general case should be a focus of future work when dealing with applications that call for more intervention timepoints. 


Although we focus on optogenetics here, the analysis framework and statistical methods proposed are relevant for a wide range of neuroscience and psychology experiments for which the ``local/micro'' longitudinal structure is of scientific interest, and that exhibit similar experimental considerations. Indeed, closed-loop designs are common in behavioral studies in human neuroimaging, psychiatry and cognitive sciences \citep{closed_loop_neurmod2, debettencourt2015closed, jangraw2023highly}. We hope our methods constitute a useful methodological contribution to the causal inference literature, and that they will help applied researchers exploit the rich information contained in their experiments.







\singlespacing


\newpage

\begin{appendices}

\section{Micro Longitudinal Effects}
In Appendix Figure~\ref{fig:estimands_appendix}, we illustrate additional micro longitudinal effects that can be probed with our sequential excursion effect framework. This figure has the same layout as Figure~\ref{fig:estimands}A-C.

\begin{figure*}
	\centering
		\begin{subfigure} 
		\centering
\includegraphics[width=0.7\linewidth]{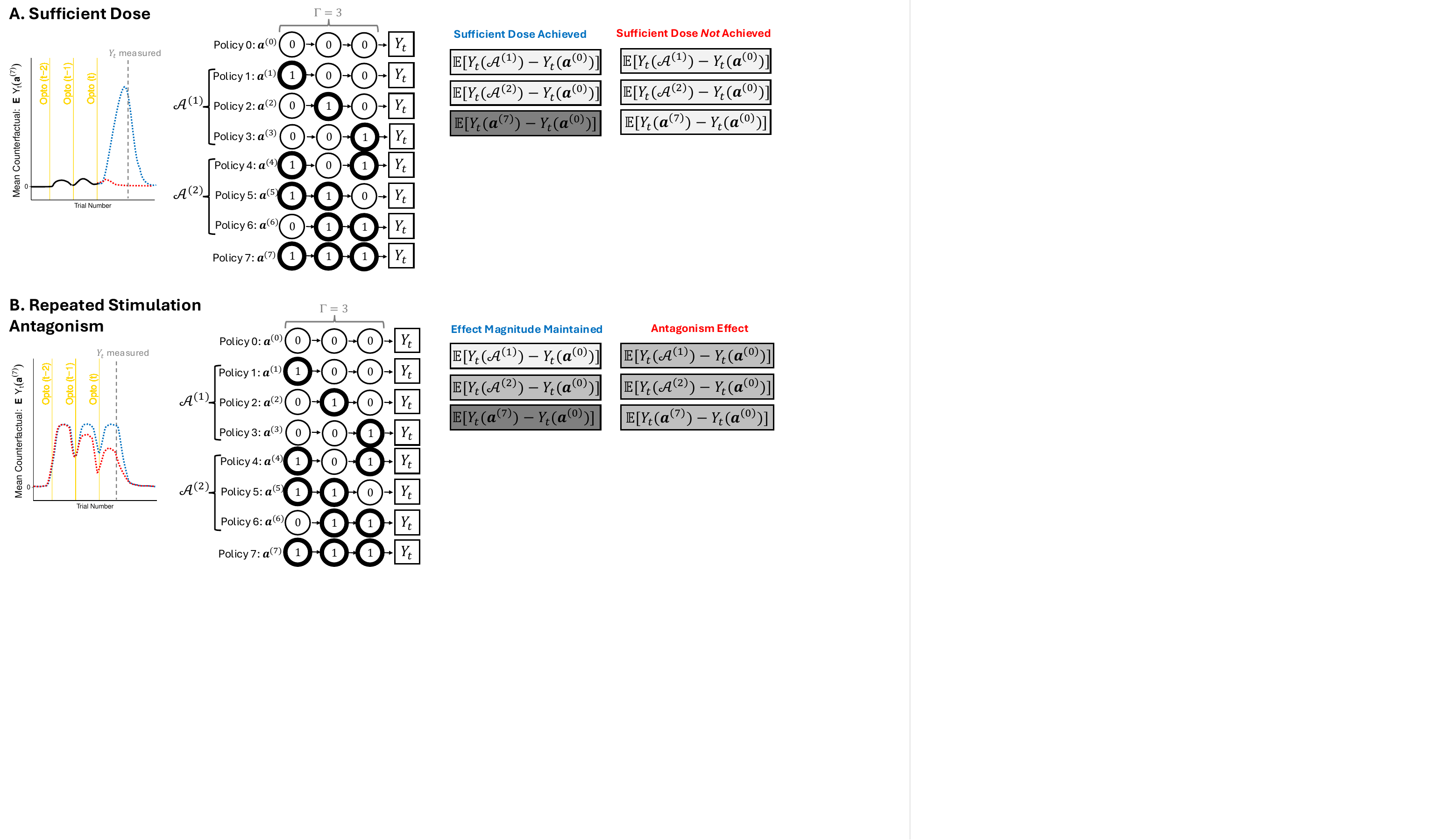}
	\end{subfigure}
	\begin{subfigure}
		\centering
\includegraphics[width=0.7\linewidth]{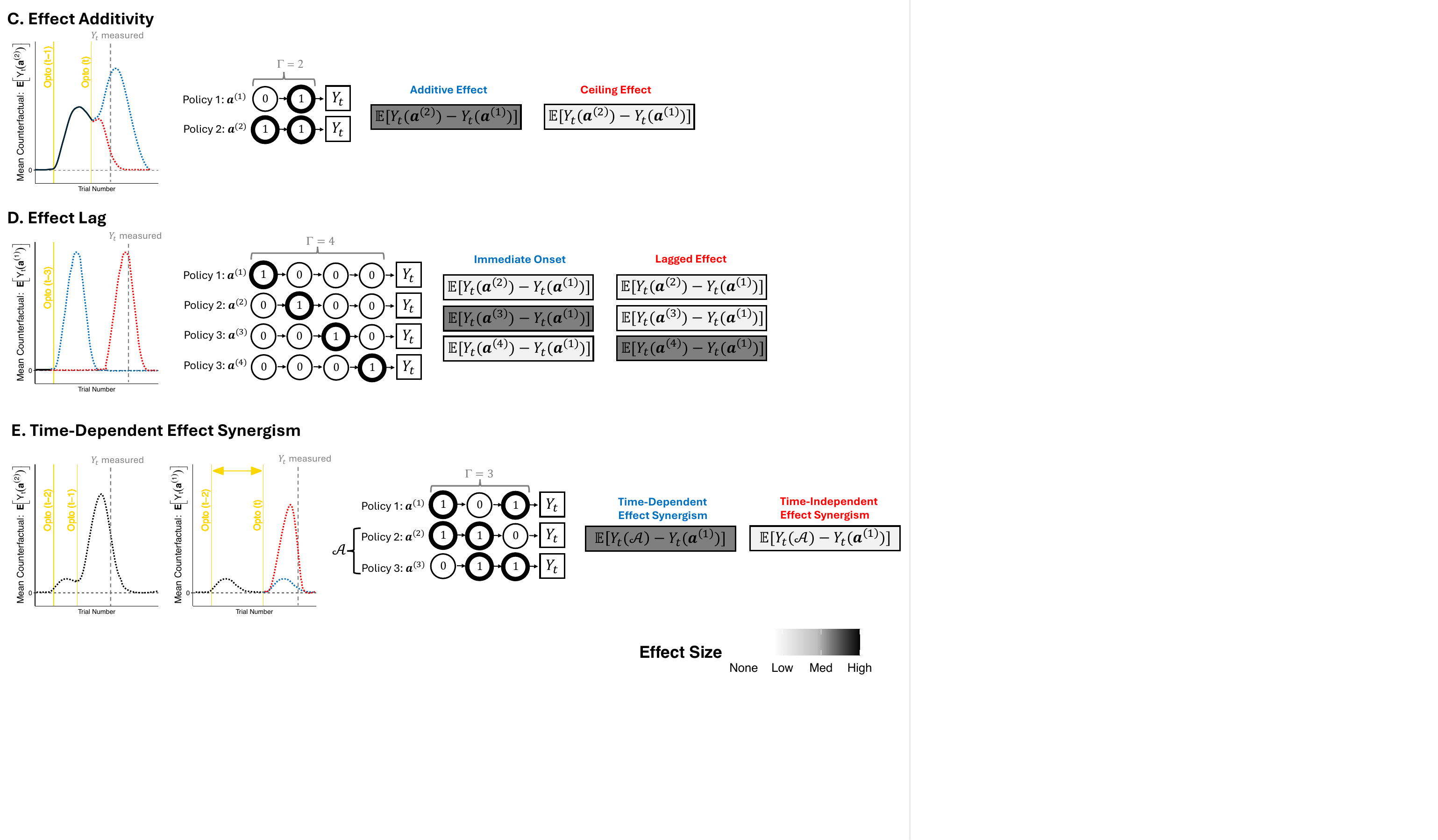}
	\end{subfigure}
\caption{\footnotesize\textbf{Example \textit{Sequential Excursion Effects}.}  The left panels show one setting where a sequence of laser simulations \color{black} do \color{black} or \color{blue} do not \color{black} have the indicated effect on the outcome. The middle panel shows deterministic static policies that could be used to construct a causal contrast to probe the effect. The right panel shows what the anticipated effect size (darker is larger) of those contrasts might be if there \color{black} is \color{black} or \color{blue} is not \color{black} the indicated effect profile. [A] Sufficient dose. The red line shows how three successive stimulations is required to trigger a large effect, whereas the effect profile in blue shows that the sufficient dose has not been reached. [B] Repeated stimulation anatagonism. The red line shows a negative dose-response, and the blue line shows a stable effect size. [C] Effect additivity. The red line shows a second stimulation triggers a larger response, whereas the blue shows that the second stimulation does not increase the response substantially beyond that of the first stimulation. [D] Effect Lag. The red line shows that the causal effect of stimulation is not visible until after a lag period. The blue line shows a setting where the effect is immediate. [E] Time-dependent effect synergism. The red line shows a setting where the effect is additive provided the stimulations occur close enough together (red line), but if stimulations occur far apart, this synergism does not occur (blue line).
} \label{fig:estimands_appendix}
\end{figure*}

\section{Illustration of Treatment-Confounder Feedback} \label{app:conf-ill}
We provide in this section a synthetic example in which two independent groups exhibit identical mean outcome patterns over time, but where the treatment (e.g., turning on laser in the brain) has a substantial effect in one group but not the other. As our construction will demonstrate, this phenomenon manifests due to treatment-confounder feedback leading to effects canceling out. In a similar fashion, one can similarly construct scenarios where effects are exaggerated.

Suppose $G \in \{0,1\}$ represents an experimentally manipulable marker (e.g., animals expressing opsin in the brain), and counterfactual outcomes under $G = g$ are denoted $Y_t^g$. We will suppose that potential outcomes generated in the active setting ($G = 1$) are given by
\[Y_t^1 \sim \mathcal{N}(\gamma_{0t} + \gamma_1 X_{t - 1} + \gamma_2 A_{t - 1} + \gamma_3 X_t + \gamma_4 A_t, \sigma_t^2),\]
and potential outcomes in the control condition ($G = 0$) are given by
\[Y_t^0 \sim \mathcal{N}(\gamma_{0t} + \gamma_1 X_{t - 1} + \gamma_3 X_t, \sigma_t^2),\]
i.e., the treatment (e.g., laser) has an effect when $G = 1$, but not when $G = 0$. 

Suppose further that a behavior $X_t$ is measured at all time points $t$, and determines whether or not treatment will be administered with positive probability. Like the outcomes, this behavior will be affected by the laser only when $G = 1$:
\[X_t^g \sim \mathrm{Bernoulli}(0.7 - 0.5 \, A_{t - 1}\, g), \text{ for } t \in \{1, \ldots, T\},\]
and $X_0^g \sim \mathrm{Bernoulli}(\frac{1}{2})$ at baseline.

Now we consider a study where animals are randomly assigned at baseline to either $G = 1$ or $G = 0$. At each time point $t$, the behavior $X_t$ is measured, and treatment is then drawn according to $A_t \sim \mathrm{Bernoulli}(0.8 X_t)$. By induction, $\mathbb{E}(A_t \mid G = 1) = 0.4$ and $\mathbb{E}(X_t \mid G = g) = 0.7 - 0.2 g$, for all $t$. It follows that
\begin{align*}
    &\mathbb{E}(Y_t \mid G = g) \\
    & = \gamma_{0t} + \gamma_1 \mathbb{E}(X_{t - 1} \mid G = g) + \gamma_2 \mathbb{E}(A_{t - 1} \mid G = 1) g + \gamma_3 \mathbb{E}(X_t \mid G = g) + \gamma_4 \mathbb{E}(A_t \mid G = 1) g \\
    &= \{\gamma_{0t} + 0.7(\gamma_1 + \gamma_3)\} + \{-0.2 (\gamma_1 + \gamma_3) + 0.4 (\gamma_2 + \gamma_4)\}g.
\end{align*}
Thus, the ``macro''/``global'' between-group mean difference trajectory is given by
\[\mathbb{E}(Y_t \mid G = 1) - \mathbb{E}(Y_t \mid G = 0) = -0.2 (\gamma_1 + \gamma_3) + 0.4 (\gamma_2 + \gamma_4),\]
which will be null if $\gamma_2 + \gamma_4 = 0.5(\gamma_1 + \gamma_3)$. Notice that this cancellation is possible even if the immediate effect of treatment on the outcome is quite strong, say if $\gamma_2$ and $\gamma_4$ are large and positive. The cancellation is made possible through the opposing effects of treatment on the intermediate behavior and the outcome: when $G = 1$, $A_{t - 1}$ negatively impacts $X_t$ but positively impacts $Y_t$. More generally, these $X_t$-$A_t$ feedback loops can lead to dilution or exaggeration of the actual effect of treatments when only analyzing observed mean outcomes.

We note that in the data generating scenario described in this appendix, the proposed dynamic treatment regime HR-MSM methodology would pick out non-null effects of treatment within the active group ($G = 1$), and show differing effects between groups, even if the condition above held such that observed mean outcomes were identical. This example thus serves to illustrate both the challenges with closed-loop designs and, despite these challenges, the ability of the proposed methodology to elucidate effects. \newline

\paragraph{Example Analysis on Synthetic Data}
To illustrate the above, we provide an example on a simulated dataset, taking $n = 100$, $T=500$, $\gamma_1 = \gamma_3 = 1$, $\gamma_2 = \gamma_4 = 0.5$.

\begin{figure}[H] 
		\centering
\includegraphics[width=0.99\linewidth]{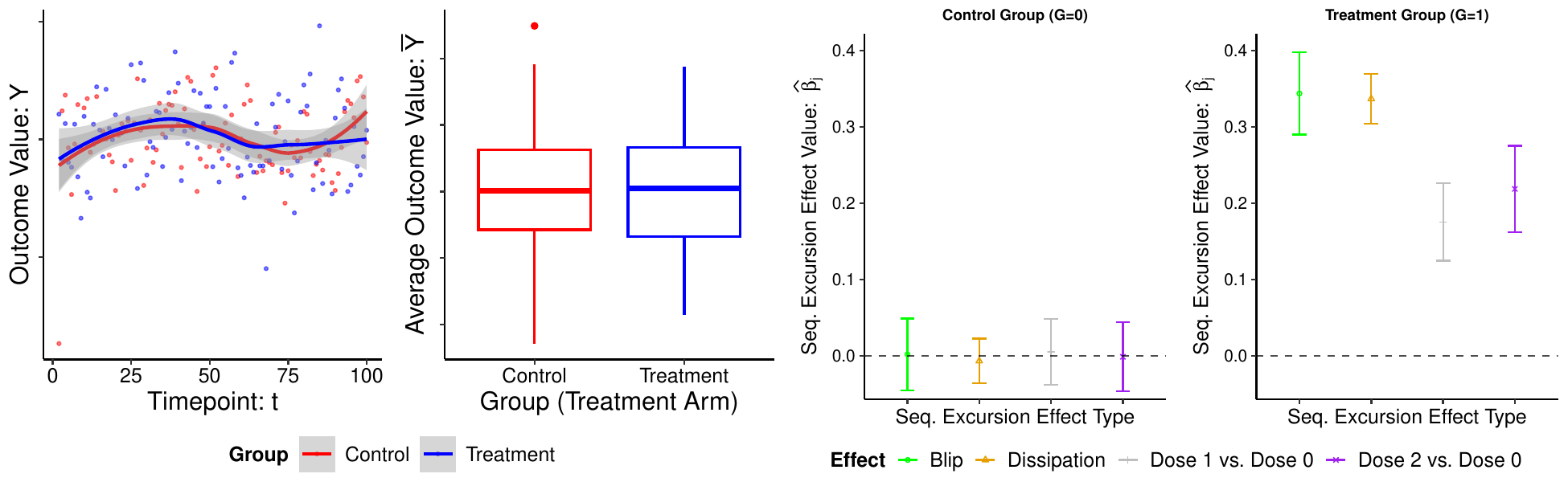}
	\caption{ \textbf{Treatment–Confounder Feedback Example} Sequential excursion effects reveal causal effects obscured in ``macro'' summaries. Analysis results from a simulated dataset following the argument above (in Appendix \ref{app:conf-ill}), taking $n = 100$, $T=500$, $\gamma_1 = \gamma_3 = 1$, $\gamma_2 = \gamma_4 = 0.5$. [A] Each dot is an outcome value, $Y^G_{i,t}$, for subject $i$ at timepoint $t$ from ``control'' ($G = 0$), or from ``treatment'' ($G = 1$) arms (groups). Lines are timepoint-specific means (averaged across subjects), estimated using a linear smoother (\texttt{loess}). (B) Same data as (A), but each point in boxplot is a subject's mean outcome value (averaged across timepoints). In (A)-(B), ``macro'' summaries show no differences due to treatment–confounder feedback: mean outcome values (averaged across subjects or timepoints) are nearly identical in both arms (groups). (C)-(D) Point estimates and 95\% CIs (error bars) of sequential excursion effects reveal ``local'' causal effects (in Treatment arm only), obscured in ``macro'' summaries (shown in (A)-(B)). } 
\label{fig:sim_results_rebuttal}
\end{figure}
\newpage

\section{Multiply Robust Estimation}\label{sec:mr}
In some optogenetics studies, treatment probabilities vary across time but may not be recorded (e.g., \cite{DA_adapts}). Similarly, in mobile health studies, the actual randomization probabilities may be unknown, or possibly not match those that were initially intended \citep{shi2023}. In such cases, one can proceed as before, but will have to construct models $\widehat{\pi}_t$ for each $t$. Note that Theorem~\ref{thm:asymp} assumes that these treatment probabilities are known by design, and thus does not account for uncertainty in estimating $\pi_t$. 

Parametric modeling of the treatment probabilities represents one simple option, though valid inference would rely on correct specification of these models. On the other hand, more flexible methods (e.g., machine learning algorithms) may result in consistent estimates of $\pi_t$, but the resulting estimate $\widehat{\boldsymbol{\beta}}$ would inherit the potentially slow (i.e., typically slower than $\sqrt{n}$) convergence rates of the flexible regression model. In order to alleviate these issues, we present here an estimator built on the ``doubly robust machine learning'' (DRML) paradigm \citep{chernozhukov2018double,kennedy2022semiparametric}. More concretely, we develop an estimator based on the nonparametric influence function of the projection parameter $\boldsymbol{\beta}_0$. Due to the second-order bias of this influence function, the convergence rate of the resulting estimator will depend on the \textit{product} of convergence rates for the underlying nuisance functions, and thus $\sqrt{n}$-convergence will be possible even with flexible nuisance model specifications.

We describe the estimator for $\Gamma = 2$, and leave the general case for future research. For any $t$ and $\boldsymbol{d}_{2, t} \equiv (d_{t -1}, d_t) \in \overline{\mathcal{D}}_{2, t}$, define $b_1^{d_t}(H_t) = \mathbb{E}_{\mathbb{P}}\left(Y_t \mid H_t, A_t = d_t(H_t)\right)$, and
$b_2^{d_{t-1}, d_t}(H_{t - 1}) = \mathbb{E}_{\mathbb{P}}\left(b_1^{d_t}(H_t) \mid H_{t-1}, A_{t-1} = d_{t-1}(H_{t - 1})\right)$. Moreover, let $\boldsymbol{b}_1 \equiv \{b_1^{d_t}: d_t \in \mathcal{D}_t^*, t \in \{2, \ldots, T\}\}$, and $\boldsymbol{b}_2 \equiv \{b_2^{d_{t - 1}, d_t}: (d_{t - 1}, d_t) \in \overline{\mathcal{D}}_{2,t}, t \in \{2, \ldots, T\}\}$ collect all $b_1^{d_t}$ and $b_2^{d_{t - 1}, d_t}$ models, and similarly let $\boldsymbol{\pi} = \{\pi_t : t \in [T]\}$. Then the following estimating function (when plugging in $\boldsymbol{\beta} = \boldsymbol{\beta}_0$) is proportional to the uncentered nonparametric influence function of $\boldsymbol{\beta}_0$:
{\small
\begin{align*}
            \psi(Z; \boldsymbol{\beta}, \boldsymbol{b}_1, \boldsymbol{b}_2, \boldsymbol{\pi})
            &= \sum_{t = 2}^T \, \, \, \, \sum_{\boldsymbol{d}_{2, t} \in \overline{\mathcal{D}}_{2, t}} h(t, \boldsymbol{d}_{2, t}, V_{t - 1})M(t, \boldsymbol{d}_{2, t}, V_{t - 1}; \boldsymbol{\beta})\bigg\{b_2^{d_{t-1}, d_t}(H_{t - 1}) - m(\boldsymbol{d}_{2, t}, t, V_{t - 1}; \boldsymbol{\beta}) \\
            & \quad \quad \quad \quad \quad \quad +\frac{\mathds{1}(A_{t - 1} = d_{t - 1}(H_{t - 1}))}{\pi_{t - 1}(A_{t - 1}; H_{t - 1})}\left(b_1^{d_t}(H_t) - b_2^{d_{t-1}, d_t}(H_{t - 1})\right) \\
            & \quad \quad \quad \quad \quad \quad +\frac{\mathds{1}(A_{t - 1} = d_{t - 1}(H_{t - 1}))\mathds{1}(A_{t} = d_{t}(H_{t}))}{\pi_{t - 1}(A_{t - 1}; H_{t - 1})\pi_{t}(A_{t}; H_{t})}\left(Y_t - b_1^{d_t}(H_t)\right)\bigg\},
\end{align*}
}

It is now straightforward to define an estimator based on $\psi$: fitting $(\widehat{\boldsymbol{b}}_1, \widehat{\boldsymbol{b}}_2, \widehat{\boldsymbol{\pi}})$ on separate independent data, let $\widehat{\boldsymbol{\beta}}_{\mathrm{mr}}$ be the solution to $\boldsymbol{0} = \mathbb{P}_n\left[\psi(Z; \boldsymbol{\beta}, \widehat{\boldsymbol{b}}_1, \widehat{\boldsymbol{b}}_2, \widehat{\boldsymbol{\pi}})\right]$, in $\boldsymbol{\beta}$. In practice when one has access to only one sample, one can perform sample splitting and cross-fitting to achieve the same performance~\citep{bickel1988, robins2008, zheng2010, chernozhukov2018double}. For simplicity, we assume a single data split, but the extension to the cross-fitted case is straightforward. In the following result, we write $\mathbb{P}(f) = \int f(z) \, d\mathbb{P}(z)$ and $\lVert f \rVert = \left\{\int f(z)^2 \, d\mathbb{P}(z)\right\}^{1/2}$ for the mean and $L_2(\mathbb{P})$-norm, respectively, of any function $f$, possibly dependent on the training data, and we write $\boldsymbol{a}^{\otimes 2} = \boldsymbol{a} \boldsymbol{a}^T$ for any vector $\boldsymbol{a} \in \mathbb{R}^q$.

\begin{theorem}\label{thm:asymp-mr}
    Suppose Assumptions~\ref{ass:consistency}--\ref{ass:NUC} hold, and $\mathbb{P}[\widehat{\pi}_t(d_t(H_t); H_t) \geq \epsilon] = 1$, as well as $\mathbb{P}[ h(t, \boldsymbol{d}_{2, t}, V_{t - 1})M(t, \boldsymbol{d}_{2, t}, V_{t -1 })\leq M] = 1$ for all $\boldsymbol{d}_{2,t} = (d_{t - 1}, d_t) \in \overline{\mathcal{D}}_{2, t}$, for all $t \in [T]$. 
    Moreover, assume that
    \begin{enumerate}[(i)]
        \item The class $\{\psi(\, \cdot \, ; \boldsymbol{\beta}, \boldsymbol{b}_1, \boldsymbol{b}_2, \boldsymbol{\pi}): \boldsymbol{\beta} \in \mathbb{R}^q\}$ is Donsker in $\boldsymbol{\beta}$, for each fixed $(\boldsymbol{b}_1, \boldsymbol{b}_2, \boldsymbol{\pi})$.
        \item $\boldsymbol{\beta}_{\mathrm{mr}} \overset{\mathbb{P}}{\to} \boldsymbol{\beta}_0$, $\lVert \widehat{\pi}_{t} - \pi_{t}\rVert = o_{\mathbb{P}}(1)$, $\lVert \widehat{b}_1^{ d_t} - b_1^{d_t}\rVert = o_{\mathbb{P}}(1)$, $\lVert \widehat{b}_2^{d_{t-1}, d_t} - b_2^{d_{t - 1}, d_t}\rVert = o_{\mathbb{P}}(1)$, for all $t$ and $\boldsymbol{d}_{2, t} = (d_{t - 1}, d_t) \in \overline{\mathcal{D}}_{2,t}$.
        \item The function $\boldsymbol{\beta} \mapsto \mathbb{P}\left(\psi(Z ; \boldsymbol{\beta}, \boldsymbol{b}_1, \boldsymbol{b}_2, \boldsymbol{\pi})\right)$ is differentiable at $\boldsymbol{\beta}_0$, uniformly in $(\boldsymbol{b}_1, \boldsymbol{b}_2, \boldsymbol{\pi})$, and $U(\boldsymbol{\beta}_0, \boldsymbol{b}_1, \boldsymbol{b}_2, \boldsymbol{\pi}) = \left. \nabla_{\boldsymbol{\beta}}\, \mathbb{P}\left(\psi(Z ; \boldsymbol{\beta}, \boldsymbol{b}_1, \boldsymbol{b}_2, \boldsymbol{\pi})\right) \right|_{\boldsymbol{\beta} = \beta_0}$ is invertible such that the nuisance estimates satisfy satisfies $U(\boldsymbol{\beta}_0, \widehat{\boldsymbol{b}}_1, \widehat{\boldsymbol{b}}_2, \widehat{\boldsymbol{\pi}}) \overset{\mathbb{P}}{\to} U(\boldsymbol{\beta}_0, \boldsymbol{b}_1, \boldsymbol{b}_2, \boldsymbol{\pi})$.
    \end{enumerate}
    Then $\widehat{\boldsymbol{\beta}}_{\mathrm{mr}} - \boldsymbol{\beta}_0 = O_{\mathbb{P}}\left(\frac{1}{\sqrt{n}} + R_n\right)$,
where
\[R_n = \sum_{t = 2}^T \, \, \, \, \sum_{\boldsymbol{d}_{2, t} \in \overline{\mathcal{D}}_{2, t}} \bigg\{\lVert \widehat{\pi}_{t - 1} - \pi_{t - 1}\rVert\cdot \lVert \widehat{b}_2^{d_{t-1}, d_t} - b_2^{d_{t - 1}, d_t}\rVert 
    + \lVert \widehat{\pi}_{t} - \pi_{t}\rVert\cdot \lVert \widehat{b}_1^{ d_t} - b_1^{d_t}\rVert  \bigg\}.\]
If, in addition, $\lVert \widehat{\pi}_{t - 1} - \pi_{t - 1}\rVert\cdot \lVert \widehat{b}_2^{d_{t-1}, d_t} - b_2^{d_{t - 1}, d_t}\rVert 
    + \lVert \widehat{\pi}_{t} - \pi_{t}\rVert\cdot \lVert \widehat{b}_1^{ d_t} - b_1^{d_t}\rVert = o_{\mathbb{P}}(n^{-1/2})$, for all values of $t$ and $\boldsymbol{d}_{2, t} = (d_{t - 1}, d_t) \in \overline{D}_{2,t}$, then
    $\sqrt{n}(\widehat{\boldsymbol{\beta}}_{\mathrm{mr}} - \boldsymbol{\beta}_0) \overset{d}{\to} \mathcal{N}(\boldsymbol{0}, \boldsymbol{V}^*(\boldsymbol{\beta}_0))$,
    where we define $\boldsymbol{V}^*(\boldsymbol{\beta}) = U(\boldsymbol{\beta}, \boldsymbol{b}_1, \boldsymbol{b}_2, \boldsymbol{\pi})^{-1} \mathbb{P}\left\{\psi^{\otimes 2}(Z ; \boldsymbol{\beta}, \boldsymbol{b}_1, \boldsymbol{b}_2, \boldsymbol{\pi})\right\} U(\boldsymbol{\beta}, \boldsymbol{b}_1, \boldsymbol{b}_2, \boldsymbol{\pi})^{-1}$.
\end{theorem}

Theorem~\ref{thm:asymp-mr} describes the asymptotic properties of $\widehat{\boldsymbol{\beta}}_{\mathrm{mr}}$. Conditions (i)--(iii) in the result are similar to those listed in Theorem~\ref{thm:asymp}, and again are relatively mild; we expect they will hold for many MSM models $m$. We see that, in general, the convergence rate of $\widehat{\boldsymbol{\beta}}_{\mathrm{mr}}$ depends on products of nuisance function convergence rates (through $R_n$). For example, $o_{\mathbb{P}}(n^{-1/4})$ convergence for each component nuisance function is sufficient for $\sqrt{n}$-consistency of $\widehat{\boldsymbol{\beta}}_{\mathrm{mr}}$; such rates are achievable under structural conditions such as smoothness, sparsity, or additivity. Furthermore, in the latter case where $R_n = o_{\mathbb{P}}(n^{-1/2})$, $\widehat{\boldsymbol{\beta}}_{\mathrm{mr}}$ is also asymptotically normal and nonparametrically efficient (i.e., the asymptotic variance $\boldsymbol{V}^*(\boldsymbol{\beta}_0)$ is the lowest possible among all regular and asymptotically linear estimators~\citep{bickel1993efficient}), and asymptotically valid Wald-based CIs are given by $\widehat{\beta}_{\mathrm{mr},j} \pm z_{1 - \alpha / 2}\sqrt{\frac{\widehat{V}_{jj}^*}{n}}$, where $\widehat{V}_{jj}^*$ is the $j$-th diagonal element of 
{\footnotesize
\[\widehat{\boldsymbol{V}}^* = \mathbb{P}_n\left(\left.\nabla_{\boldsymbol{\beta}}\, \psi(Z ; \boldsymbol{\beta}, \widehat{\boldsymbol{b}}_1, \widehat{\boldsymbol{b}}_2, \widehat{\boldsymbol{\pi}})\right|_{\boldsymbol{\beta} = \widehat{\boldsymbol{\beta}}_{\mathrm{mr}}}\right)^{-1} 
\mathbb{P}_n\left(\psi^{\otimes 2}(Z ; \boldsymbol{\beta}, \widehat{\boldsymbol{b}}_1, \widehat{\boldsymbol{b}}_2, \widehat{\boldsymbol{\pi}})\right)
\mathbb{P}_n\left(\left.\nabla_{\boldsymbol{\beta}}\, \psi(Z ; \boldsymbol{\beta}, \widehat{\boldsymbol{b}}_1, \widehat{\boldsymbol{b}}_2, \widehat{\boldsymbol{\pi}})\right|_{\boldsymbol{\beta} = \widehat{\boldsymbol{\beta}}_{\mathrm{mr}}}\right)^{-1}.\]}%
CIs for linear combinations of the parameters are similarly straightforward to obtain.

\section{Proofs of Identification and Estimation Results}

\begin{proof}[Proof of 
Proposition~\ref{prop:sharp-null}]
    This is a straightforward consequence of consistency and the sharp null: for any $\boldsymbol{a}_{\Gamma, t} \in \{0,1\}^{\Gamma}$, $Y_t(\boldsymbol{a}_{\Gamma, t}) = Y_t(\boldsymbol{A}_{\Gamma, t}) \equiv Y_t$. Concretely, for any $\boldsymbol{d}_{\Gamma, t} \in \overline{\mathcal{D}}_{\Gamma, t}$,
    \[Y_t(\boldsymbol{d}_{\Gamma, t}) = Y_t(a_{t - \Gamma + 1} = d_{t - \Gamma + 1}(H_{t - \Gamma + 1}), \ldots, a_t = d_t(H_t(\boldsymbol{d}_{\Gamma - 1, t - 1}))) = Y_t,\]
    which implies the result.
\end{proof}

\begin{proof}[Proof of 
Proposition~\ref{prop:ident}]
    This result follows the usual $g$-formula identification argument \citep{robins1986}: defining $\boldsymbol{A}_{\Gamma, t} = (A_{t - \Gamma + 1}, \ldots, A_t)$, $\boldsymbol{d}_{\Gamma, t}(H_t)) = (d_{t + \Gamma + 1}(H_{t + \Gamma + 1}), \ldots, d_t(H_t))$,
    {\small
    \begin{align*}
        &\mathbb{E}(Y_t(\boldsymbol{d}_{\Gamma, t}) \mid V_{t - \Gamma + 1}) \\
        &= \mathbb{E}(\mathbb{E}(Y_t(\boldsymbol{d}_{\Gamma, t})) \mid H_{t - \Gamma + 1}) \mid V_{t - \Gamma + 1}) \\
        &= \mathbb{E}(\mathbb{E}(Y_t(\boldsymbol{d}_{\Gamma, t})) \mid H_{t - \Gamma + 1}, A_{t - \Gamma + 1} = d_{t - \Gamma + 1}(H_{t - \Gamma + 1})) \mid V_{t - \Gamma + 1}) \\
        &\ \ldots \\
        &= \mathbb{E}(\mathbb{E}( \cdots \mathbb{E}(Y_t(\boldsymbol{d}_{\Gamma, t}) \mid H_t, \boldsymbol{A}_{\Gamma, t} = \boldsymbol{d}_{\Gamma, t}(H_t)) \cdots \mid H_{t - \Gamma + 1}, A_{t - \Gamma + 1} = d_{t - \Gamma + 1}(H_{t - \Gamma + 1})) \mid V_{t - \Gamma + 1}),
    \end{align*}
    }%
    where we repeatedly invoke iterated expectations and Assumption~\ref{ass:NUC} (justified by Assumption~\ref{ass:positivity}), then use Assumption~\ref{ass:consistency} in the last equality. We can then rewrite this formula in an equivalent IPW form:
        {\small
    \begin{align*}
        &\mathbb{E}(\mathbb{E}( \cdots \mathbb{E}(Y_t(\boldsymbol{d}_{\Gamma, t}) \mid H_t, \boldsymbol{A}_{\Gamma, t} = \boldsymbol{d}_{\Gamma, t}(H_t)) \cdots \mid H_{t - \Gamma + 1}, A_{t - \Gamma + 1} = d_{t - \Gamma + 1}(H_{t - \Gamma + 1})) \mid V_{t - \Gamma + 1}) \\
        &= \mathbb{E}\left(\mathbb{E}\left(\frac{\mathds{1}(A_{t - \Gamma + 1} = d_{t - \Gamma + 1}(H_{t - \Gamma + 1}))}{\pi_{t - \Gamma + 1}(A_{t - \Gamma + 1}; H_{t - \Gamma + 1})} \cdots \mathbb{E}\left(\frac{\mathds{1}(A_t = d_t(H_t))}{\pi_t(A_t; H_t)}Y_t(\boldsymbol{d}_{\Gamma, t}) \mid H_t \right)\cdots \mid H_{t - \Gamma + 1}\right) \mid V_{t - \Gamma + 1}\right) \\
        &= \mathbb{E}\left(\prod_{j = t - \Gamma + 1}^t \frac{\mathds{1}(A_j = d_j(H_j))}{\pi_j(A_j; H_j)} Y_t \mid V_{t - \Gamma + 1}\right),
    \end{align*}
    }%
    where the last equality is achieved again by iterated expectations. The second statement in Proposition~\ref{prop:ident} is obtained by differentiating~\eqref{eq:msm-proj} with respect to $\boldsymbol{\beta}$, setting this to zero, then invoking the first statement of Proposition~\ref{prop:ident} (which we have just proved).
\end{proof}

\begin{proof}[Proof of Theorem~\ref{thm:asymp}]

This is an immediate application of Theorem 5.31 in \citep{vandervaart2000}.
\end{proof}

\begin{proof}[Proof of Theorem~\ref{thm:asymp-mr}]

This result follows from a ``master lemma'' for sample-split solutions to estimating equations with nuisance estimates plugged in: see Lemma 3 of \citet{kennedy2023}. From their general formulation, we must only verify that under our assumptions,
\[\mathbb{P}\left(\psi(Z; \boldsymbol{\beta}_0, \widehat{\boldsymbol{b}}_1, \widehat{\boldsymbol{b}}_2, \widehat{\boldsymbol{\pi}}) - \psi(Z; \boldsymbol{\beta}_0, \boldsymbol{b}_1, \boldsymbol{b}_2, \boldsymbol{\pi})\right) = O_{\mathbb{P}}(R_n),\]
where the remainder term $R_n$ is given in the statement of the theorem. To that end, observe that this asymptotic bias term equals
{
\begin{align*}
            &\mathbb{P}\bigg(\sum_{t = 2}^T \, \, \, \, \sum_{\boldsymbol{d}_{2, t} \in \overline{\mathcal{D}}_{2, t}} h(t, \boldsymbol{d}_{2, t})M(t, \boldsymbol{d}_{2, t}; \boldsymbol{\beta}_0)\bigg\{\widehat{b}_2^{d_{t-1}, d_t}(H_{t - 1}) - b_2^{d_{t-1}, d_t}(H_{t - 1}) \\
            & \quad \quad + \frac{\mathds{1}(A_{t - 1} = d_{t - 1}(H_{t - 1}))}{\widehat{\pi}_{t - 1}(A_{t - 1}; H_{t - 1})}\left(\widehat{b}_1^{d_t}(H_t) - \widehat{b}_2^{d_{t-1}, d_t}(H_{t - 1})\right)\\
            & \quad \quad +\frac{\mathds{1}(A_{t - 1} = d_{t - 1}(H_{t - 1}))\mathds{1}(A_{t} = d_{t}(H_{t}))}{\widehat{\pi}_{t - 1}(A_{t - 1}; H_{t - 1})\widehat{\pi}_{t}(A_{t}; H_{t})}\left(Y_t - \widehat{b}_1^{d_t}(H_t)\right)\bigg\}\bigg) \\
&= \mathbb{P}\bigg(\sum_{t = 2}^T \, \, \, \, \sum_{\boldsymbol{d}_{2, t} \in \overline{\mathcal{D}}_{2, t}} h(t, \boldsymbol{d}_{2, t})M(t, \boldsymbol{d}_{2, t}; \boldsymbol{\beta}_0)\bigg\{\widehat{b}_2^{d_{t-1}, d_t}(H_{t - 1}) - b_2^{d_{t-1}, d_t}(H_{t - 1}) \\
    & \quad \quad + \frac{\pi_{t - 1}(d_{t - 1}(H_{t - 1}); H_{t - 1})}{\widehat{\pi}_{t - 1}(d_{t - 1}(H_{t - 1}); H_{t - 1})}\left(\widehat{b}_1^{d_t}(H_t) - b_1^{d_t}(H_{t})\right)\\
    & \quad \quad + \frac{\mathds{1}(A_{t - 1} = d_{t - 1}(H_{t - 1}))}{\widehat{\pi}_{t - 1}(A_{t - 1}; H_{t - 1})}\left(b_2^{d_{t - 1},d_t}(H_{t - 1}) - \widehat{b}_2^{d_{t-1}, d_t}(H_{t - 1})\right)\\
    & \quad \quad +\frac{\mathds{1}(A_{t - 1} = d_{t - 1}(H_{t - 1}))\pi_t(d_t(H_t); H_t)}{\widehat{\pi}_{t - 1}(A_{t - 1}; H_{t - 1})\widehat{\pi}_{t}(d_{t}(H_t); H_{t})}\left(b_1^{d_t}(H_t) - \widehat{b}_1^{d_t}(H_t)\right)\bigg\}\bigg) \\
&= \mathbb{P}\bigg(\sum_{t = 2}^T \, \, \, \, \sum_{\boldsymbol{d}_{2, t} \in \overline{\mathcal{D}}_{2, t}} h(t, \boldsymbol{d}_{2, t})M(t, \boldsymbol{d}_{2, t}; \boldsymbol{\beta}_0)\bigg\{\left[\widehat{b}_2^{d_{t-1}, d_t} - b_2^{d_{t-1}, d_t}\right] \left[1 - \frac{\pi_{t - 1}(d_{t - 1}(A_{t - 1}); H_{t-1})}{\widehat{\pi}_{t - 1}(d_{t-1}(A_{t - 1}); H_{t - 1})}\right]\\
    & \quad \quad \quad \quad + \frac{\mathds{1}(A_{t - 1} = d_{t - 1}(H_{t - 1}))}{\widehat{\pi}_{t - 1}(A_{t - 1}; H_{t - 1})}\left[\widehat{b}_1^{d_t} - b_1^{d_t}\right]\left[1 - \frac{\pi_t(d_t(H_t); H_t)}{\widehat{\pi}_{t}(d_t(H_t); H_{t})}\right]\bigg\}\bigg) \\
&= O_{\mathbb{P}}\left(\sum_{t = 2}^T \, \, \, \, \sum_{\boldsymbol{d}_{2, t} \in \overline{\mathcal{D}}_{2, t}} \bigg\{\lVert \widehat{\pi}_{t - 1} - \pi_{t - 1}\rVert\cdot \lVert \widehat{b}_2^{d_{t-1}, d_t} - b_2^{d_{t - 1}, d_t}\rVert 
    + \lVert \widehat{\pi}_{t} - \pi_{t}\rVert\cdot \lVert \widehat{b}_1^{ d_t} - b_1^{d_t}\rVert  \bigg\}\right),
\end{align*}
}
under our assumptions, by Cauchy-Schwarz. Since the term in the $O_{\mathbb{P}}$ expression is $R_n$, we are done.
\end{proof}

\section{Properties of Availability-Respecting Estimands} \label{app:sharp-null}

\begin{proof}[Proof of Proposition~\ref{prop:sharp-null}]
    This is a straightforward consequence of consistency and the sharp null: for any $\boldsymbol{a}_{\Gamma, t} \in \{0,1\}^{\Gamma}$, $Y_t(\boldsymbol{a}_{\Gamma, t}) = Y_t(\boldsymbol{A}_{\Gamma, t}) \equiv Y_t$. Concretely, for any $\boldsymbol{d}_{\Gamma, t} \in \overline{\mathcal{D}}_{\Gamma, t}$,
    \[Y_t(\boldsymbol{d}_{\Gamma, t}) = Y_t(a_{t - \Gamma + 1} = d_{t - \Gamma + 1}(H_{t - \Gamma + 1}), \ldots, a_t = d_t(H_t(\boldsymbol{d}_{\Gamma - 1, t - 1}))) = Y_t,\]
    which implies the result.
\end{proof}

To provide insight into sequential excursion effects, we express the proposed contrasts in terms of observed variables.

\subsection{Special Case: (1,0) vs. (0,0)}
We start with the contrast (1,0) vs. (0,0): a contrast of (a) the deterministic sequence of a treatment opportunity on timepoint $t-1$ and no treatment opportunity on timepoint $t$ vs. (b) a treatment sequence with no treatment opportunity on either timepoint.
\begin{align*}
    &\mathbb{E} \left( Y_t(d_{t-1}^{(1)}, d_t^{(0)})  - Y_t(d_{t-1}^{(0)}, d_t^{(0)})  \right ) \\
    &= \mathbb{E} \left[ \mathbb{E} \left( Y_t(d_{t-1}^{(1)}, d_t^{(0)}) - Y_t(d_{t-1}^{(0)}, d_t^{(0)}) \mid I_{t-1} \right )  \right ] \\
    &= \sum_{i_{t-1}=0}^1 \mathbb{E} \left( Y_t(d_{t-1}^{(1)}(i_{t-1}), d_t^{(0)}) - Y_t(d_{t-1}^{(0)}(i_{t-1}), d_t^{(0)}) \mid I_{t-1} \right ) \mathbb{P} \left(I_{t-1} = i_{t-1} \right) \\
    &= \mathbb{E} \left( Y_t(1, d_t^{(0)}) - Y_t(0, d_t^{(0)}) \mid I_{t-1} =1\right ) \mathbb{P} \left(I_{t-1} = 1 \right) \\
    &=\mathbb{P} \left(I_{t-1} = 1 \right) \left[ \mathbb{E} \left( Y_t(d_t^{(0)}) \mid I_{t-1}=1, A_{t-1}=1 \right ) - \mathbb{E} \left( Y_t(d_t^{(0)}) \mid I_{t-1}=1, A_{t-1}=0 \right ) \right ] \tag{1}\\
    &=\mathbb{P} \left(I_{t-1} = 1 \right) \big[  \mathbb{E} \left\{ \mathbb{E} \left( Y_t \mid I_{t-1}=1, A_{t-1}=1, I_t, A_t=0 \right ) \mid  I_{t-1}=1, A_{t-1}=1 \right \} - \\
    & \quad \quad \quad \quad \quad \quad \quad \quad \mathbb{E} \left\{ \mathbb{E} \left( Y_t \mid I_{t-1}=1, A_{t-1}=0, I_t, A_t=0 \right ) \mid  I_{t-1}=1, A_{t-1}=0 \right \} \big ] \tag{2}
\end{align*}

The above derivation provides a few key insights into how availability-respecting estimands relate to both the effect of treatment on availability in subsequent timepoints $I_t$, and the outcome $Y_t$. Expression (1) shows how the estimand is designed to capture the deterministic regime that we ``wish we could compare'' if there were no positivity violations: the contrast compares a treated vs. untreated population on timepoint $t-1$, conditional on being available on timepoint $t-1$. Expression (2) further shows how treatment confounder feedback manifests in this estimand: the inner expectation marginalizes over the effect of $A_{t-1}$ (i.e., treatment on timepoint $t-1$) on availability at $t$, $I_t$. That is, the estimand captures the effect of treatment at timepoint $t-1$ on \textit{both} $Y_t$ and $I_t$. It is easy to see in expression (1) how this contrast would be 0 under the sharp null of no effect of treatment, itself (not just treatment opportunities).

\subsection{Special Case: (0,1) vs. (0,0)}
We now repeat the same steps for the contrast (0,1) vs. (0,0): a contrast of (a) the deterministic sequence of no treatment opportunity on timepoint $t-1$ and a treatment opportunity on timepoint $t$ vs. (b) a treatment sequence with no treatment opportunity on either timepoint.
\begin{align*}
    &\mathbb{E} \left( Y_t(d_{t-1}^{(0)}, d_t^{(1)})  - Y_t(d_{t-1}^{(0)}, d_t^{(0)})  \right ) \\
    &= \mathbb{E} \left[ \mathbb{E} \left( Y_t(d_{t-1}^{(0)}, d_t^{(1)}) - Y_t(d_{t-1}^{(0)}, d_t^{(0)}) \mid I_{t-1} \right )  \right ] \\
    &= \sum_{i_{t-1}=0}^1 \mathbb{E} \left( Y_t(d_{t-1}^{(0)}(i_{t-1}), d_t^{(1)}) - Y_t(d_{t-1}^{(0)}(i_{t-1}), d_t^{(0)}) \mid I_{t-1} \right ) \mathbb{P} \left(I_{t-1} = i_{t-1} \right) \\
    &= \mathbb{P} \left(I_{t-1} = 0 \right) \mathbb{E} \left[ Y_t(a_{t}=1) - Y_t(a_t=0) \mid I_t = 1, I_{t-1} =0, A_{t-1}=0\right ] \mathbb{P} \left(I_{t} = 1 \mid I_{t-1}=0, A_{t-1}=0 \right) + \\
    & \quad ~ \mathbb{P} \left(I_{t-1} = 1 \right) \mathbb{E} \left[ Y_t(a_{t}=1) - Y_t(a_t=0) \mid I_t = 1, I_{t-1} =1, A_{t-1}=0\right ]\mathbb{P} \left(I_{t} = 1 \mid I_{t-1}=1, A_{t-1}=0 \right )
\end{align*}
It is clear that each expectation in the final expression will be zero when the sharp null holds of no causal effect in any individual (i.e., $Y_t(a_{t}=1) = Y_t(a_t=0)$). Moreover, this illustrates how the estimand is a weighted sum of causal contrasts of the treatment itself. Under our assumptions, we can further express the expectations in terms of observed variables.
\begin{align*}
    &\mathbb{E} \left[ Y_t(a_{t}=1) - Y_t(a_t=0) \mid I_t = 1, I_{t-1} =0, A_{t-1}=0\right ]  \\
    &=\mathbb{E} \left[Y_t \mid I_{t-1}=0, A_{t-1}=0, I_t=1, A_t=1 \right] - \mathbb{E} \left[Y_t \mid I_{t-1}=0, A_{t-1}=0, I_t=1, A_t=0 \right]
\end{align*}
and 
\begin{align*}
    &\mathbb{E} \left[ Y_t(a_{t}=1) - Y_t(a_t=0) \mid I_t = 1, I_{t-1} =1, A_{t-1}=0\right ] \\
    &=\mathbb{E} \left[Y_t \mid I_{t-1}=1, A_{t-1}=0, I_t=1, A_t=1 \right] - \mathbb{E} \left[Y_t \mid I_{t-1}=1, A_{t-1}=0, I_t=1, A_t=0 \right]
\end{align*}


\section{Method Implementation}\label{app:implement}

We provide an implementation that builds the necessary dataset (with each observation copied once for every regime in $\overline{D}_{\Gamma, t}$), calculates the corresponding IP weights, and estimates the HR-MSM parameters $\boldsymbol{\beta}$ by solving the estimating equation in expression \ref{eq:msm-proj} using the \texttt{rootSolve} R package \citep{rootSolve2}. Our implementation provides variance estimates of the HR-MSM parameters using the sandwich estimator we derived, and/or the sample size-adjusted \texttt{HC} sandwich estimators (using the \texttt{sandwich} package in \texttt{R} \citep{sandwich_package}). The entire process takes roughly 10 seconds on a standard laptop, for $>100,000$ total (pre-copy) timepoints. This code is available at anonymous GitHub Repo: \url{https://anonymous.4open.science/r/causal_opto-52CD/README.md}.

\section{Application Details and Additional Results} \label{sec:app-secondary}
{\color{black} We fit the HR-MSM models below using a logit link. For binary data, \citet{Qian2021} instead target marginal relative risk parameters. These two approaches have different trade-offs: odds ratios are non-collapsible, while relative risks do not constrain the conditional probability of the outcome to lie in the range $[0,1]$. Collapsibility would be a problem in our application if (1) we believed the HR-MSM to be correct, and logistic, and (2) we also assumed some conditional outcome model (e.g., nuisance function) was logistic. In this case, these would be incompatible due to non-collapsibility. Neither of these apply to us. For (1), we are taking a projection parameter approach and thus we are not assuming the marginal structural model is correctly specified---see the detailed discussion at the end of Section~\ref{sec:hr-msm}. For (2), our proposed strategy either uses an IPW (no conditional outcome model) or a flexible doubly robust estimator where the outcome model need not be logistic. Nevertheless, some authors prefer the interpretation of relative risk parameters, and argue that it yields greater modeling ease. For this reason, we recommend considering both strategies when modeling optogenetics data.}

\subsection{Effects of a Single Dose $r$ timepoints Prior}\label{sec:lag-main-effects}

In Appendix Figure~\ref{fig:lag-main-effects}, we show estimates of the main effects of a single stimulation opportunity $r=1,2, \text{ or } 3$ timepoints prior to $Y_t$ from main text model \eqref{dose_group_int} (reshown below as model \eqref{eq:dose_group_int_trt}) in each treatment arm. These effects, $\widehat{\beta}_r$ with $r \in [3]$, are estimates of the log odds ratio comparing the mean counterfactual of $Y_t$ under a treatment sequence with a single dose (on $r = 1,2$, or $3$ timepoints prior)
vs. a treatment sequence with zero dose, within a given treatment arm. 

\begin{equation} \label{eq:dose_group_int_trt}
    \mbox{logit}\left ( \mathbb{E} [Y_t(\boldsymbol{d}_{\Gamma, t}) \mid  G = 0] \right ) = \beta_0 + \sum_{r=0}^2 \beta_{r+1}  \sigma_{t - r}(d_{t - r}).
\end{equation}

We fit an analogous model (shown in expression \eqref{eq:dose_group_int_ctrl}) to yield the same estimates among the optogenetics group.

\begin{equation} \label{eq:dose_group_int_ctrl}
    \mbox{logit}\left ( \mathbb{E} [Y_t(\boldsymbol{d}_{\Gamma, t}) \mid  G = 1] \right ) = \beta'_0 + \sum_{r=0}^2 \beta'_{r+1}  \sigma_{t - r}(d_{t - r}).
\end{equation}

\begin{figure}[H]
	\centering
\includegraphics[width=0.75\linewidth]{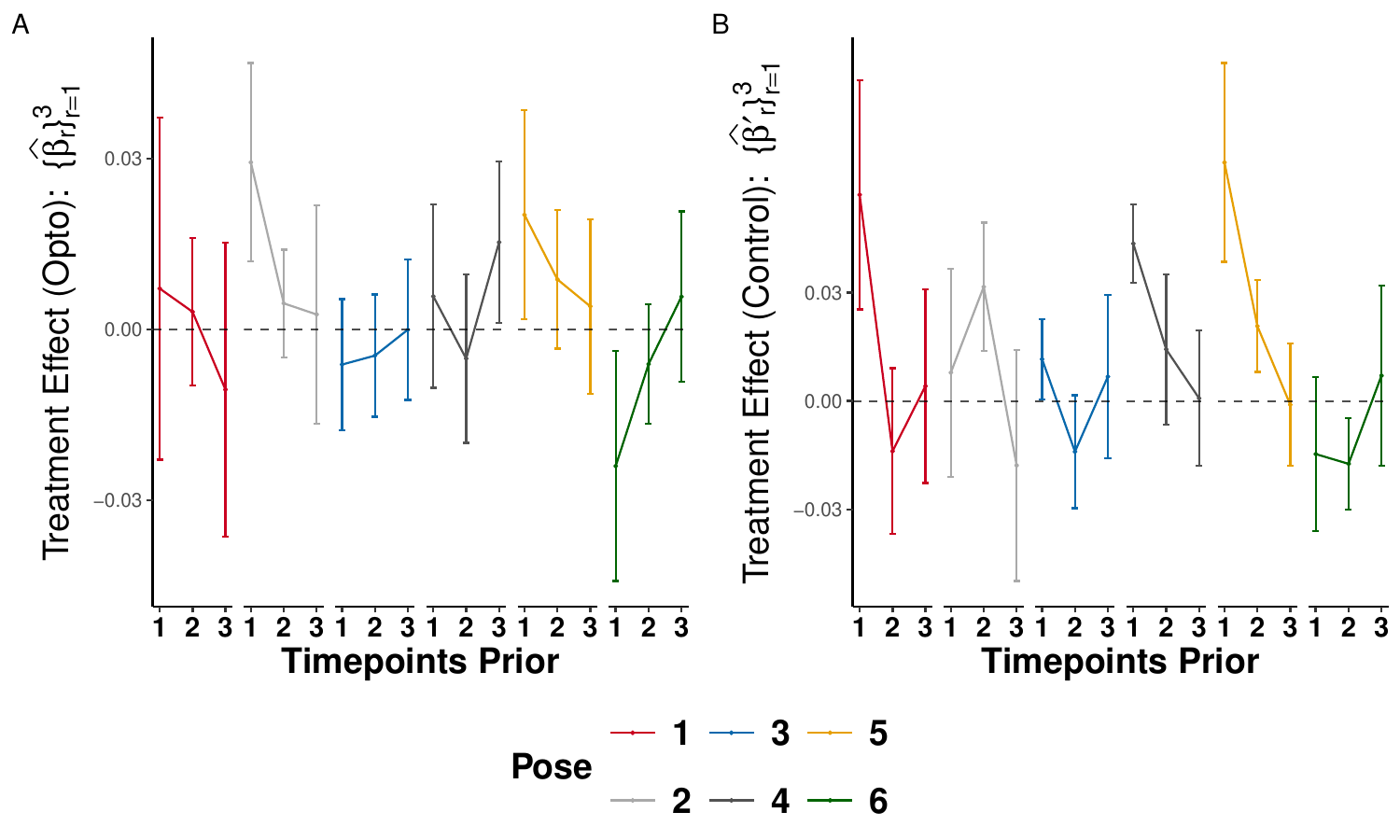}
\caption{\footnotesize\textbf{Our method enables estimation of effect dissipation.} Plots show coefficient estimates (error bars show 95\% CIs). [A] Main effects in the optogenetics group. $\widehat{\beta}_r$ with $r \in [3]$ from HR-MSM \eqref{dose_group_int} (shown above as expression \eqref{eq:dose_group_int_trt}). [B] Main effects in the optogenetics group. $\widehat{\beta}'_r$ with $r \in [3]$ from HR-MSM \eqref{eq:dose_group_int_ctrl}. } 		\label{fig:lag-main-effects}
\end{figure}

\subsection{Dose-Response Effects in Optogenetics Group}\label{app:dose-resp-opto}
Here we present results
from our analyses described in Section~\ref{sec:dose-resp}. In the left panel of Figure~\ref{fig:dose_optoDA}, we show results from the following HR-MSM:
\begin{equation} \label{eq:dose_trt_supp}
    \mbox{logit}\left ( \mathbb{E} [Y_t(\boldsymbol{d}_{\Gamma, t}) \mid  G=1] \right ) = \beta_0 + \sum_{r=1}^3 \beta_r \mathds{1} \bigg( \sum_{j = t - \Gamma + 1}^ t \sigma_j(d_j) = r \bigg),
    \end{equation}

\paragraph{Conditional Excursion Effect}
We next estimate an availability-conditional estimand \citep{Boruvka2018}, to determine if existing excursion effect methods have the capacity to reveal the effects identified with our method. We estimate this in the MSM

\begin{equation} \label{eq:dose_cond_ctrl}
    \mbox{logit}\left ( \mathbb{E} [Y_t(a_t) \mid I_{t} = 1, G=1] \right ) = \alpha_0 + \alpha_1 a_t.
    \end{equation}

Results are shown in the right panel of Figure~\ref{fig:dose_optoDA}.

\begin{figure}[H]
	\centering
\includegraphics[width=0.65\linewidth]{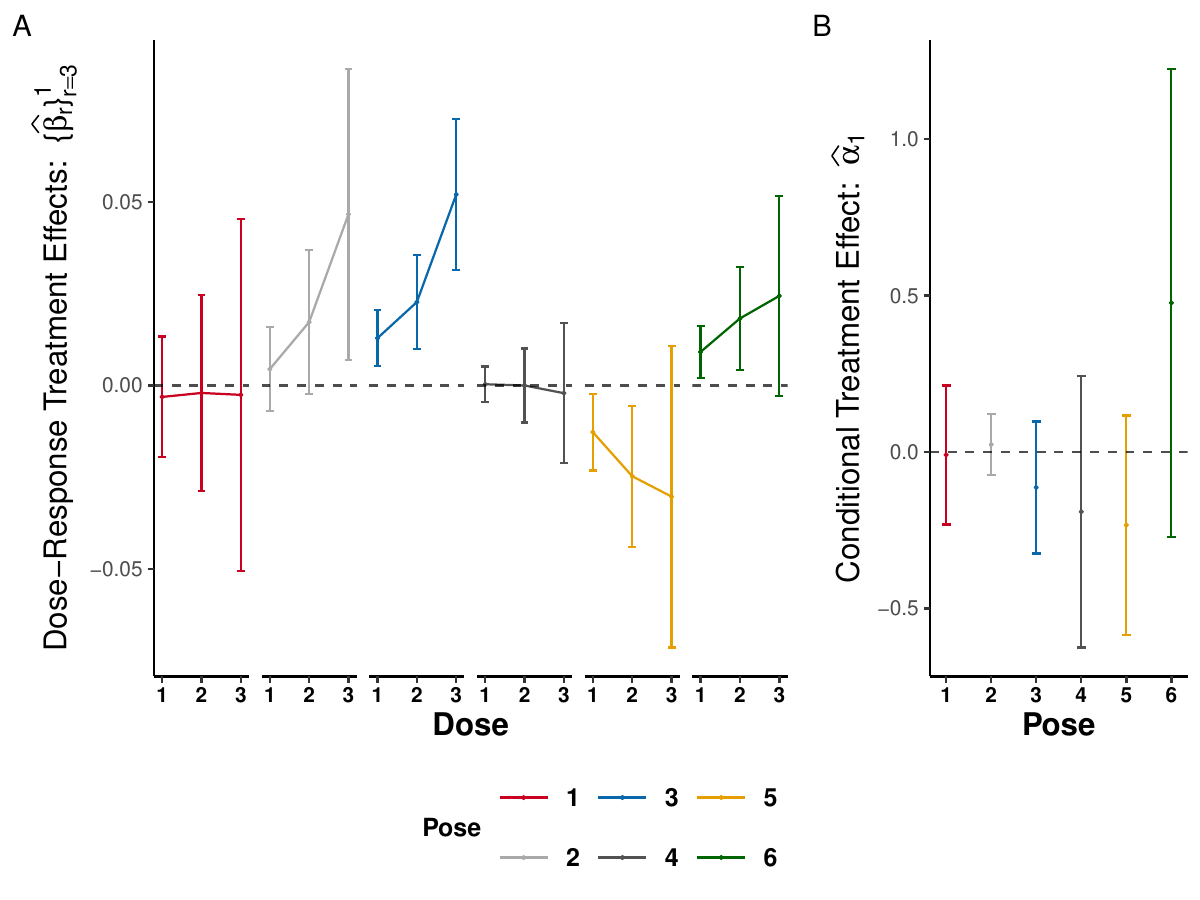}
\caption{\footnotesize\textbf{Our method enables estimation of dose effects in optogenetics arm.} Plots show coefficient estimates (error bars show 95\% CIs) as a function of dose. Columns and colors indicate the dose. [Left] Main effects of stimulation opportunity from HR-MSM \eqref{eq:dose_trt}. [Right] Availability-conditional effects of treatment estimated in MSM \eqref{eq:dose_cond}. } 		\label{fig:dose_optoDA}
\end{figure}

\subsection{Dose-Response Effects in Control Group}\label{sec:dose-resp-ctrl}
In the left panel of Figure~\ref{fig:dose_optoDA_ctrl}, we show results of the same analyses as those presented in Appendix~\ref{app:dose-resp-opto} but conducted in the control group ($G=0$). We fit an HR-MSM to estimate the causal effect of ``dose'', the number of treatment opportunities in the previous $\Gamma = 5$ timepoints:


\begin{equation} \label{eq:dose_ctrl}
    \mbox{logit}\left ( \mathbb{E} [Y_t(\boldsymbol{d}_{\Gamma, t}) \mid  G=0] \right ) = \beta_0 + \sum_{r=1}^3 \beta_r \mathds{1} \left ( \sum_{j = t - \Gamma + 1}^ t \sigma_j(d_j) = r \right ),
    \end{equation}

where $\sigma_j(d_j) = \mathds{1}(d_j = d_j^{(1)})$. The coefficient $\widehat{\beta}_r$ is an estimate of the log odds ratio comparing the mean counterfactual of $Y_t$ for a treatment sequence of dose $r \in [3]$ compared to a sequence of dose zero (see Figure \ref{fig:estimands}C for an illustration of the static regime analogue). A dose of three is the maximum feasible dose for $\Gamma = 5$ since the same pose cannot occur on two consecutive timepoints.

\paragraph{Conditional Excursion Effect}
We next estimate an availability-conditional estimand \citep{Boruvka2018}, to determine if existing excursion effect methods have the capacity to reveal the effects identified with our method. We estimate this in the MSM

\begin{equation} \label{eq:dose_cond_ctrl}
    \mbox{logit}\left ( \mathbb{E} [Y_t(a_t) \mid I_{t} = 1, G=0] \right ) = \alpha_0 + \alpha_1 a_t.
    \end{equation}

Results are shown in the right panel of Figure~\ref{fig:dose_optoDA_ctrl}.

\begin{figure}[H]
	\centering
\includegraphics[width=0.65\linewidth]{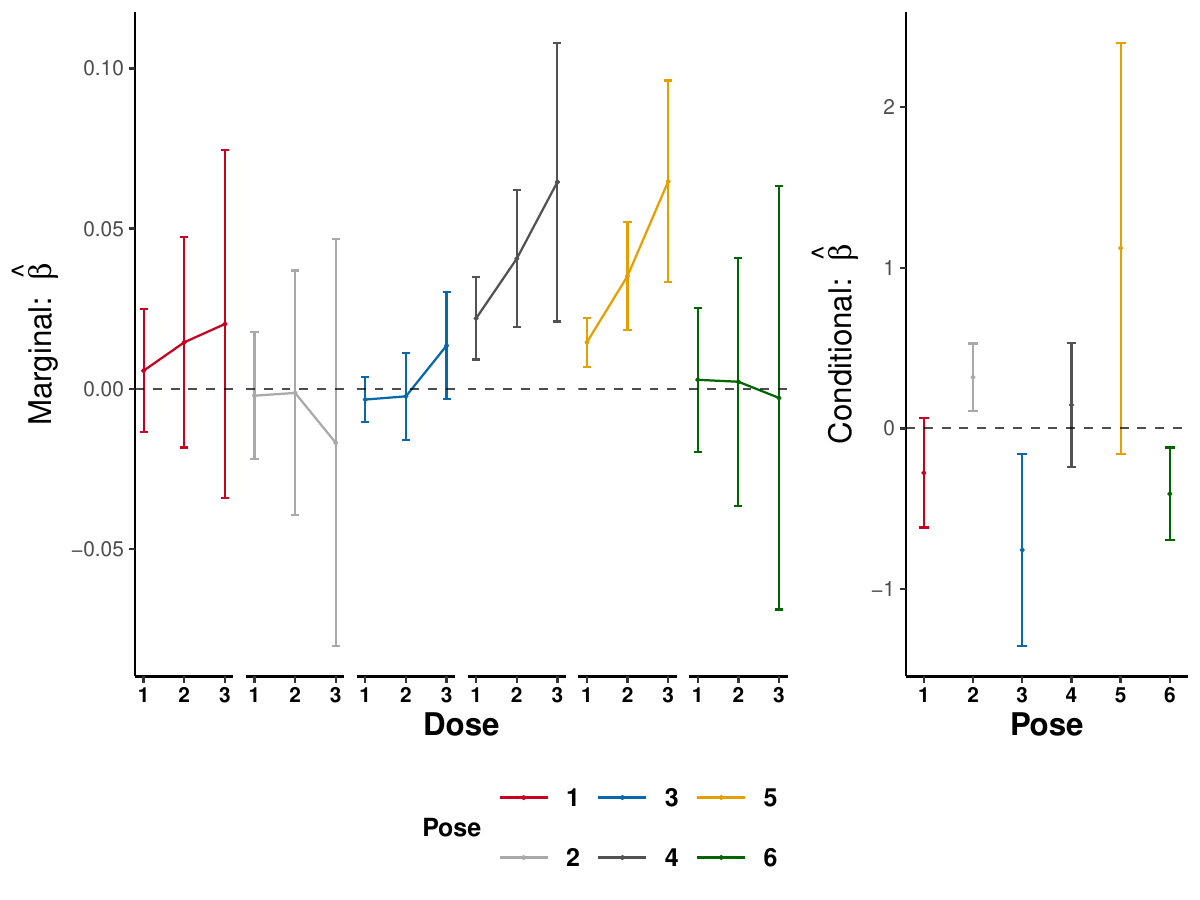}
\caption{\footnotesize\textbf{Our method enables estimation of dose effects in control arm.} Plots show coefficient estimates (error bars show 95\% CIs) as a function of dose. Columns and colors indicate the dose. [Left] Main effects of stimulation opportunity from HR-MSM \eqref{eq:dose_ctrl}. [Right] Availability-conditional effects of treatment estimated in MSM \eqref{eq:dose_cond_ctrl}. } 		\label{fig:dose_optoDA_ctrl}
\end{figure}
\subsection{Dose-Response Effects Over Timepoints}\label{sec:dose-resp-time}
We next show results from a follow-up to the analysis presented in Section~\ref{sec:dose-resp}. We fit the same model with the addition of an interaction with timepoint $t$ to test whether the causal effect of ``dose'' varied (linearly) across timepoints. Dose is defined as the number of treatment opportunities in the previous $\Gamma = 5$ timepoints:


\begin{equation*} \label{eq:dose_time}
\resizebox{\textwidth}{!}{
    $\mbox{logit}\left ( \mathbb{E} [Y_t(\boldsymbol{d}_{\Gamma, t}) \mid  G=0] \right ) = \beta_0 + \sum_{r=1}^3 \beta_r \mathds{1} \left ( \sum_{j = t - \Gamma + 1}^ t \sigma_j(d_j) = r \right ) + \beta_4 t + \sum_{r=1}^3 \beta_{r+4} \mathds{1} \left ( \sum_{j = t - \Gamma + 1}^ t \sigma_j(d_j) = r \right ) \times  t$},
    \end{equation*}

where $\sigma_j(d_j) = \mathds{1}(d_j = d_j^{(1)})$. The coefficient $\{\widehat{\beta}_r\}_{r=1}^3$ is an estimate of the log odds ratio comparing the mean counterfactual of $Y_t$ for a treatment sequence of dose $r \in [3]$ compared to a sequence of dose zero. Similarly, the coefficient estimates $\{\widehat{\beta}_r\}_{r=5}^7$ are the interaction terms of those doses with time $t$. A dose of three is the maximum feasible dose for $\Gamma = 5$ since the same pose cannot occur on two consecutive timepoints. Appendix Figure~\ref{fig:dose_time} shows that for some poses, the effects of dose do change (linearly) across timepoints. None of these are statistically significant however.

\begin{figure}[H]
	\centering
\includegraphics[width=0.65\linewidth]{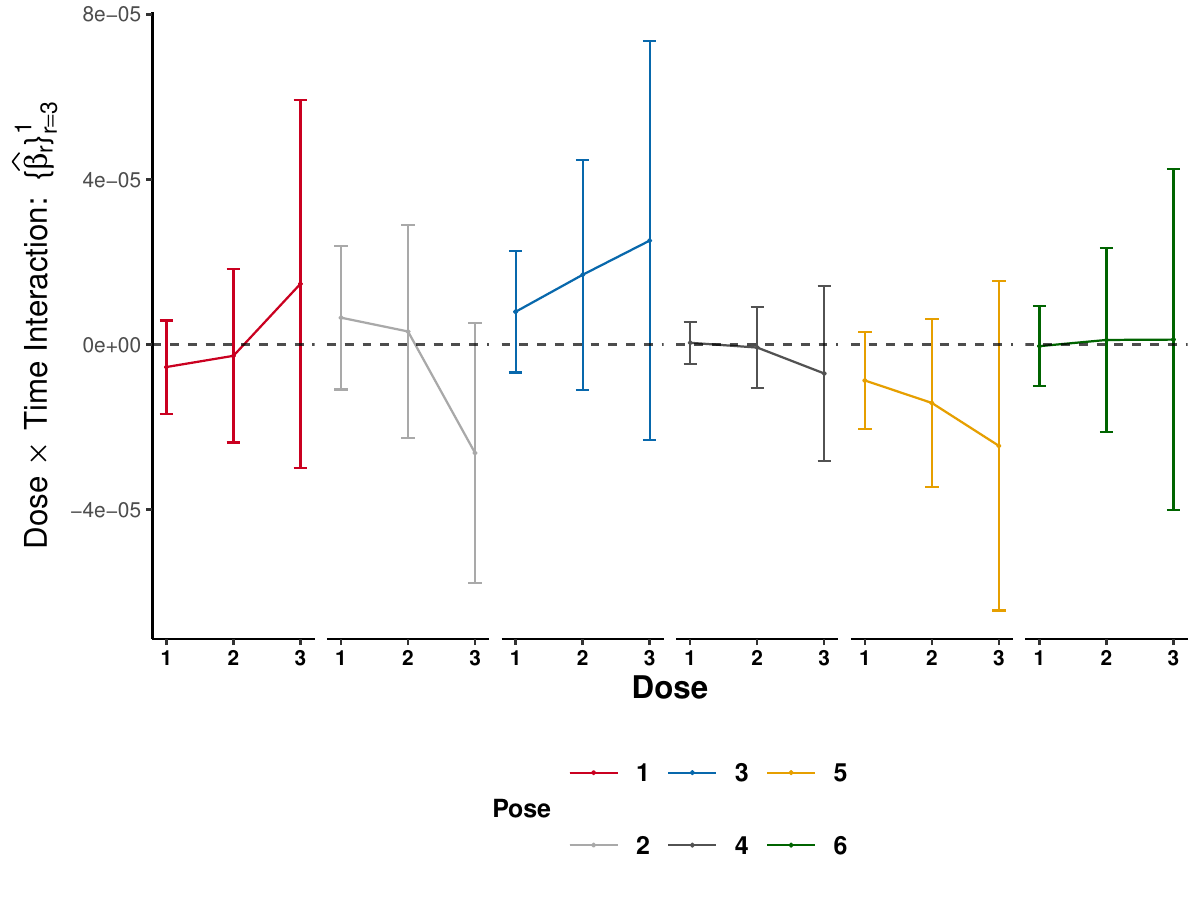}
\caption{\footnotesize\textbf{Effect of dose across timepoints.} Plots show coefficient estimates (error bars show 95\% CIs) as a function of dose. Columns and colors indicate the dose. Coefficients for interaction between dose and timepoints. } 		\label{fig:dose_time}
\end{figure}
\subsection{Sensitivity of Results to Choice of $\Gamma$: Effects of a Single Dose $r$ Timepoints Prior}\label{sec:gamma-sensitivity}

In Appendix Figure~\ref{fig:gamma-sensitive}, we show the results of a sensitivity analysis for the effect of a single stimulation opportunity $r=1,2, \ldots \text{ or } \Gamma$ timepoints prior to $Y_t$, for $\Gamma \in \{3,5,7\}$. We fit HR-MSMs with an analogous specification as model \eqref{dose_group_int} (reshown as model \eqref{eq:dose_group_int_trt} above), except with different $\Gamma$ values. We show the coefficient estimates corresponding to the interaction terms between treatment arm and the dissipation effects for each $r \in [\Gamma]$ at different $\Gamma$ values. The dissipation effect is the log odds ratio comparing the mean counterfactual of $Y_t$ under a treatment sequence with a single dose (on $r = 1,2,\ldots $, or $3$ timepoints prior)
vs. a treatment sequence with zero dose. Appendix Figure~\ref{fig:gamma-sensitive} shows that the corresponding effect estimates have similar values across different $\Gamma$ values demonstrating the results are relatively stable across different $\Gamma$ values. As expected, the $95\%$ CIs grow with $\Gamma$, since there are fewer observed sequences of timepoints that follow the corresponding policy as $\Gamma$ grows.

\begin{figure*}[t] 
	\centering
	\begin{subfigure}
		\centering
\includegraphics[width=0.45\linewidth]{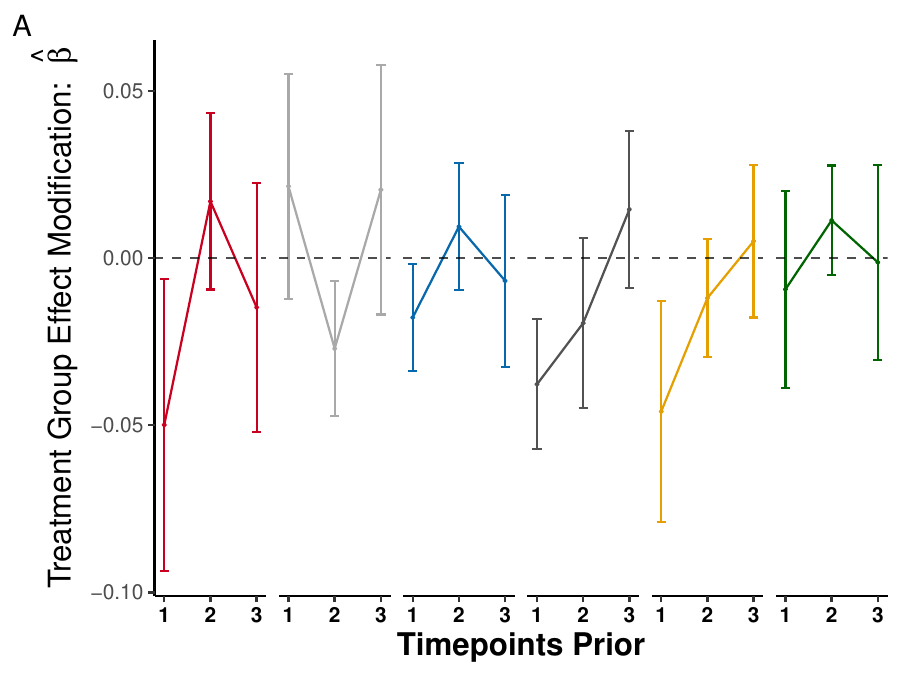}
	\end{subfigure}
    \begin{subfigure}
		\centering
\includegraphics[width=0.55\linewidth]{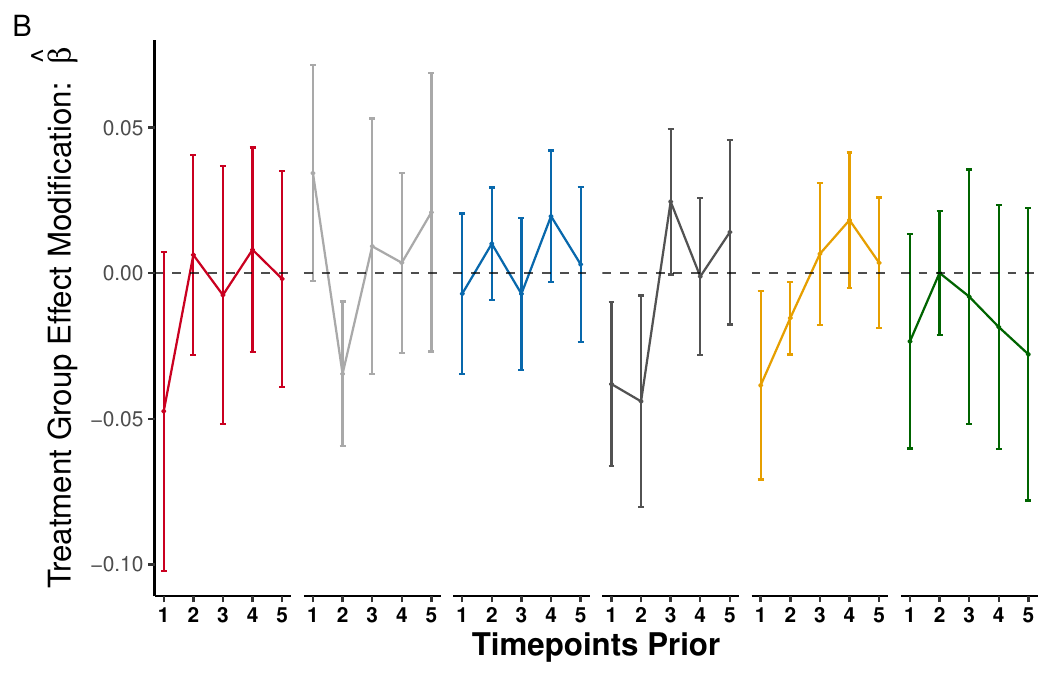}
	\end{subfigure}
    	\begin{subfigure}
		\centering
\includegraphics[width=0.7\linewidth]{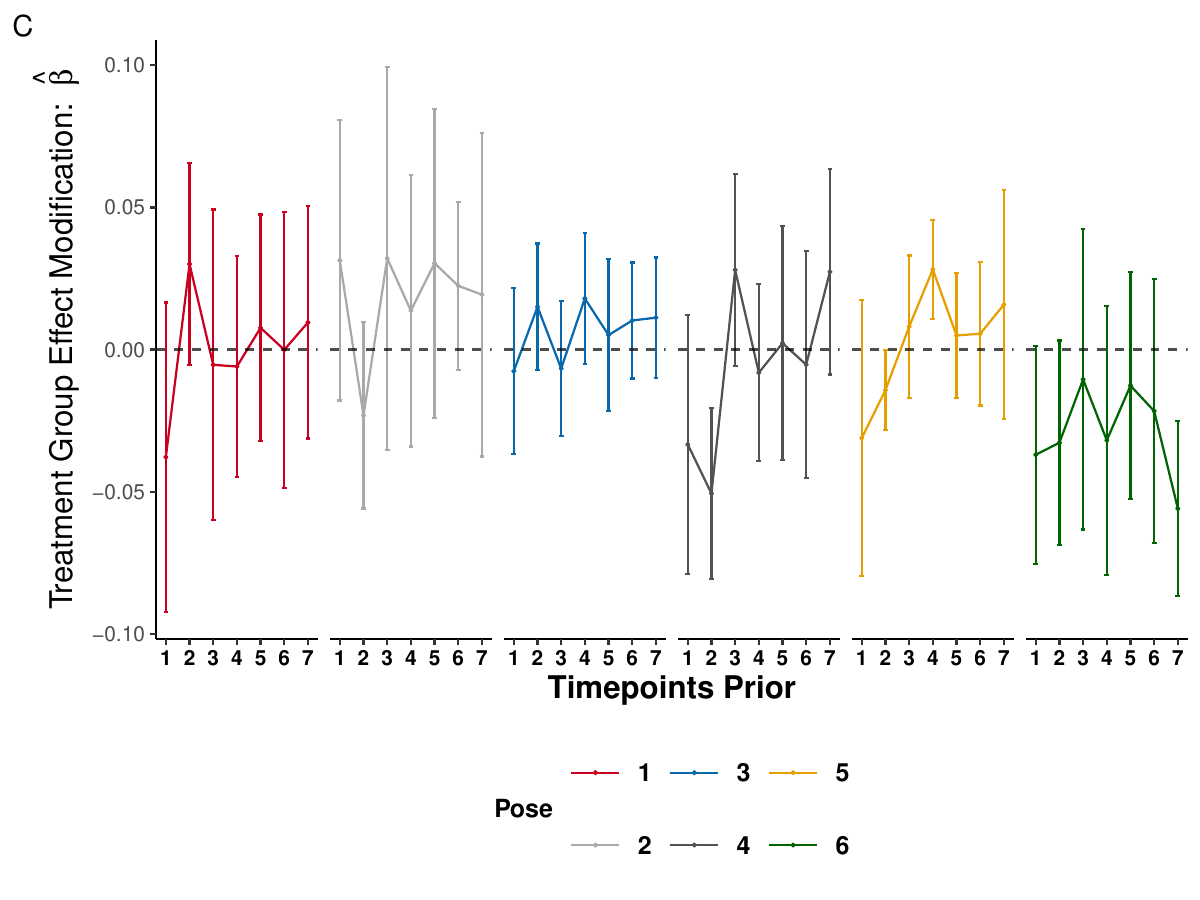}
	\end{subfigure}
	\caption{\footnotesize\textbf{Results are stable across $\Gamma$ values.} Plots show coefficient estimates (error bars show 95\% CIs) as a function of dose. Columns and colors indicate the dose. Coefficient estimates and $95\%$ CIs for interaction terms between treatment arm and the effect of a single stimulation on timepoint $t-r+1$ (for $r\in [\Gamma]$ doses prior)  for [A] $\Gamma=3$, [B] $\Gamma=5$, and [C] $\Gamma=7$.  } 	 
\label{fig:gamma-sensitive}
\end{figure*}
\clearpage
\subsection{Lagged}\label{lag_appendix}
We presented results in the main text for $\Delta = 1$, but show results for $\Delta = 5$ and $\Delta = 10$ here to illustrate that our framework can easily be extended for lagged variables or functions of outcome sequences for general $Y_t = f(Y_t, Y_{t+1},...,Y_{t+\Delta})$. In this application, we define $\tilde{Y}^{(\Delta)}_t = \max(Y_t, Y_{t+1},...,Y_{t+\Delta})$ and fit an MSM for each $\Delta$ with mean model

\begin{equation} \label{dose_trt_appendix}
    \mbox{logit}\left ( \mathbb{E} [\tilde{Y}^{(\Delta)}_t(\boldsymbol{d}_{\Gamma, t}) \mid  G=1] \right ) = \beta_0 + \sum_{r=1}^3 \beta_r \mathds{1} \left ( \sum_{j = t - \Gamma + 1}^ t \sigma_j(d_j) = r \right ).
    \end{equation}

Results are presented in Appendix Figure \ref{fig:dose_optoDA_full}. The results are qualitatively similar across values of $\Delta$, but some coefficient estimates change in magnitude and/or statistical significance. Together these results illustrate that be redefining the outcome, one can easily analyze lagged outcomes with our framework.

\begin{figure}[H]
	\centering
\includegraphics[width=0.65\linewidth]{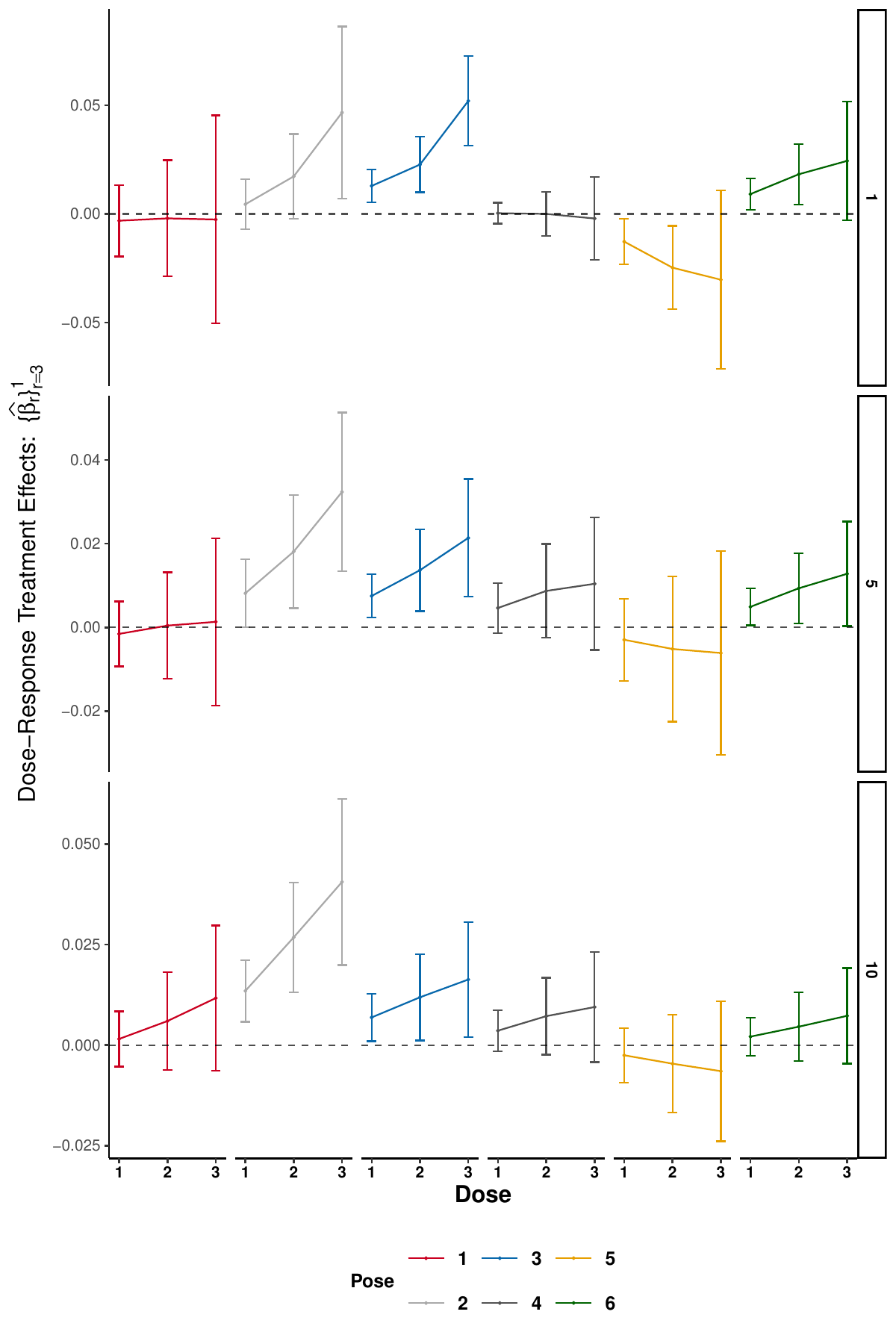}
\caption{\footnotesize\textbf{Our method enables estimation of dose effects.} Plots show coefficient estimates (error bars show 95\% CIs) as a function of dose. Columns and colors indicate the dose. Rows indicate the lag values, $\Delta$, in MSM \refeq{dose_trt_appendix}. Coefficient estimates represent main effects of stimulation with marginal HR-MSM (our approach), $\Gamma = 5$.
} 		\label{fig:dose_optoDA_full}
\end{figure}

\subsection{Replication of Original Author Analysis} \label{appendix_opto_replication}
Finally, Figure \ref{fig:lag_int}D shows that the ``standard'' macro summaries of individual poses exhibit no significant changes, emphasizing how estimands that marginalize over the stochastic dynamic (closed-loop) policies can obscure effects. We visualize the within-subject differences in target pose counts between treatment (opto) and baseline sessions in Figure~\ref{fig:raw_count}: $\sum_{t=1}^T Y_t - \sum_{t=1}^{T_0} Y_t^0$, where $Y_t, Y_t^0 \in \{0,1\}$ are indicators that the animal exhibited the target pose on timepoint $t$ of the treatment and baseline sessions, respectively; $T, T_0 \in \mathbb{Z}$ denote the total number of timepoints in treatment and baseline sessions, respectively. Because of the timepoint definition, the total number of timepoints usually differed within-subject between treatment and baseline sessions (i.e., $T \neq T_0$), but we found comparable analyses of $\frac{1}{T} \sum_{t=1}^T Y_t$ and $\frac{1}{T_0} \sum_{t=1}^{T_0} Y^0_t$ yielded qualitatively similar results to analyses of total counts $\sum_{t=1}^T Y_t - \sum_{t=1}^{T_0} Y_t^0$. For that reason, we present results in terms of total outcome counts to be consistent with the analyses presented in \cite{spont_da}. 

\begin{figure}[H]
	\centering
\includegraphics[width=0.6\linewidth]{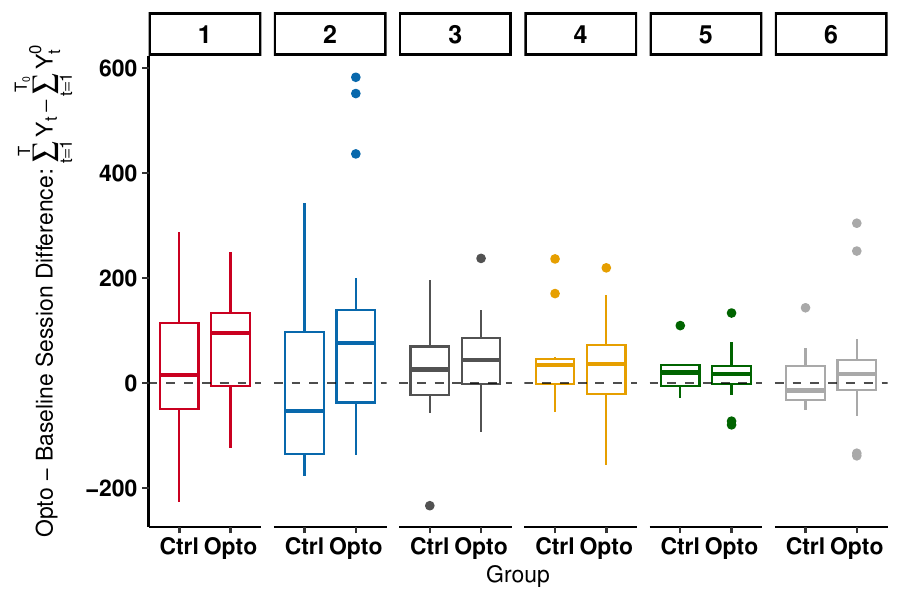} 
\caption{\footnotesize\textbf{Difference between target pose counts within-subject between baseline and treatment (opto) sessions.} Each point in the boxplot shows $\sum_{t=1}^T Y_t - \sum_{t=1}^{T_0} Y_t^0$, where $Y_t, Y_t^0 \in \{0,1\}$ are indicators that the animal exhibited the target pose on timepoint $t$ of the treatment and baseline sessions, respectively. $T, T_0 \in \mathbb{Z}$ were the total number of timepoints in treatment and baseline sessions, respectively. Columns and colors indicate target pose. \textit{Ctrl} and \textit{Opto} indicates negative-control and treatment group subjects, respectively.} 		
\label{fig:raw_count}
\end{figure}



To ensure the differences in results from our analyses and those of the \cite{spont_da} were not a result of the binning strategy they applied, we used the publicly available data, binned and pre-processed by the authors. We fit the following GEE model
\begin{equation} \label{macro_long_binned}
    \mbox{log} \left ( \mathbb{E} [\bar{\tilde{Y}}^s \mid G = g, S = s ] \right ) = \beta_0 + \beta_1 g + \beta_2 s + \beta_3 g \times s,
    \end{equation}

where $\bar{\tilde{Y}}^s = \sum_{b=1}^B \tilde{Y}^s_b$, and $\tilde{Y}^s_b$ is a count of the number of target poses the animal engaged in on bin $b$ of session $s$, and $B$ is the total number of timebins in the session. $S=1$ indicates optogenetics sessions, and $S=0$ indicates baseline sessions. We conducted inference based on the same sandwich variance estimator (``HC0'') we applied in our MSMs to be consistent. We adopted a Poisson likelihood working model for the GEE given that the outcome were counts. The estimate $\widehat{\beta}_3$ thus provides a hypothesis test of the opsin treatment effect. Figure \ref{fig:prePost_bin_fig} shows results similar to the unbinned analysis presented in the main text. This further confirms that the ``standard'' macro summary estimand obscures results.

\begin{figure}[H] 
	\centering
\includegraphics[width=0.50\linewidth]{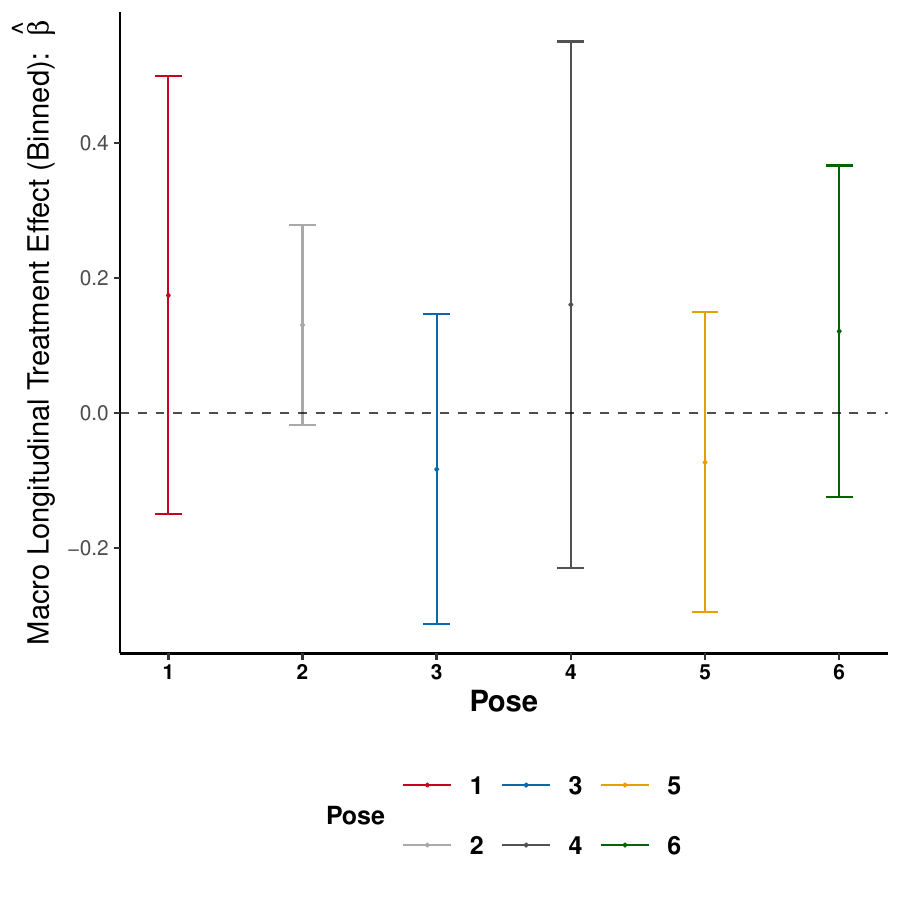}
	\caption{\footnotesize \textbf{Replication of authors results.} Plots show coefficient estimates (error bars show 95\% CIs) as a function of dose. Columns and colors indicate the pose. Analysis was conducted on the binned data provided by the \cite{spont_da}.} 
\label{fig:prePost_bin_fig}
\end{figure}


\subsection{Pre-processing} \label{sec:app-preprocess}
We downloaded the open-source application dataset from \cite{spont_da} from \url{https://zenodo.org/records/7274803}. We used the open-source pre-processing code provided by the authors on Github repo: \\ \url{https://github.com/dattalab/dopamine-reinforces-spontaneous-behavior}. Our analysis code is provided on anonymous GitHub Repo: \url{https://anonymous.4open.science/r/causal_opto-52CD/README.md}. We constructed timepoints as described in \cite{spont_da}. That is, we defined timepoints as consecutive timepoints when the animal was classified to be in a given pose. For our MSM analyses of the treatment (optogenetics)  sessions, we classified ``target pose'' timepoints only if they met the criteria of \cite{spont_da}, which required that the HMM predictions had sufficiently high forward algorithm probabilities of the latent states. This indicator was provided in the opto session dataset provided by the authors. The baseline session data did not, however, include this indicator since no optogenetic stimulation was applied. Thus when recreating the ``standard'' between-group (macro summary) analyses that compared baseline and optogenetics session data, we did not classify target pose timepoints based on whether it met this criteria: we classified the pose based on the most likely latent state prediction but did not require the forward algorithm probabilities met the threshold set by the authors (for either baseline or opto sessions to be consistent). We corresponded closely with the original authors to ensure we were pre-processing data correctly. 

There was a small percentage of timepoints that the authors described eliminating because they were deemed too short. We did not eliminate these timepoints because this created inconsistencies in the pattern of timepoints: it allowed two consecutive timepoints to be of the same timepoint type which broke with the pattern in the remainder of the dataset. This was a very small percentage of the dataset and appeared to have negligible effects on analyses.

To the best of our understanding, the original authors' hypothesis tests were conducted on further processed version of the data that first calculating the number of target pose occurrences in each 30 second bin of the experiment (period). From our understanding, this was done to provide a smoothed time-series of outcome frequency across the course of the opto sessions. We conducted similar analyses to make sure our pre-processing yielded comparable results, but we did not use these pre-processing steps in our MSM analyses or replication of the ``standard'' between-group (macro summary) analyses as it appeared to ``double-count'' target pose occurences that began before the end of one 30 second bin and ended after the start of the subsequent bin. 

As described in \citep{spont_da}, the experiment included two 30 minute replicates of both opto and baseline sessions. We constructed timepoints on each replicate separately (to account for the discontinuity in time between replicates) and then pooled the replicate datasets together to be consistent with the analysis procedures in \cite{spont_da}. We accounted for the longitudinal structure by using sandwich variance estimators in all of our hypothesis tests. 

Finally, we analyzed data from all animals available in the online repository. This included data from two groups of animals, each of which received a virus that causes expression of a different stimulating optogenetic protein. In personal communications with \cite{spont_da}, they felt this was entirely reasonable since both proteins cause neuronal excitation of the dopamine neurons targeted in their experiments (although some properties of the proteins may differ). Their analyses were, however, conducted on only one group of subjects. To avoid biasing our analyses, we completed all of our analyses without examining results on the subset of subjects that \cite{spont_da} analyzed.



\subsection{Treatment–Counfounder Feedback}\label{trt_conf_feedback}

Our reanalysis of the data from \cite{spont_da} provides an illustration of how causal estimands of the mean counterfactuals $\mathbb{E}[Y_t(\bar{A}_t)]$ under the \textit{observed} \textit{stochastic dynamic} regimes, $\bar{A}_t$, can obscure or exaggerate effects (when comparing to causal contrasts of mean counterfactuals under \textit{deterministic} regimes) because of ``treatment–counfounder feedback'' \citep{hernan_causal_2023}. 
This occurs, at least in part, because the intervention influences both the outcome $Y_t$ but also availability, $I_t$, which among other things, influences the total ``dose'' adminstered throughout an experiment (e.g., $\sum_{t=1}^T A_t$). For example, animals that exhibit a naturally high baseline frequency of the target pose, will be available for treatment more often, and thus will receive a higher ``total dose'' throughout the experiment on average. Suppose laser stimulation caused a \textit{reduction} in target pose frequency. Then the laser will cause a larger cumulative reduction in outcome probability among these ``high dose-receiving'' subjects. This could in turn dilute the differences between a treatment and control group if one analyzes \textit{mean} responding across an entire experiment, as done by \cite{spont_da}. 
Conversely, if the laser caused an \textit{increase} in target pose frequency, animals with high target pose frequency responding will tend to receive higher total doses, thereby exaggerating the effects. That is, standard analyses of the ``macro'' longitudinal effects, can yield diluted or exaggerated effects because they marginalize across a \textit{stochastic} dynamic policy. In contrast, sequential excursion effects based on \textit{deterministic} policies can undercover a rich set of causal effects partially because they are estimated by properly accounting for treatment-counfounder feedback.

\section{Further Simulation Details and Results}\label{sec:app_sims}

\subsection{HR-MSM for Simulation Data-Generating Mechanism}\label{app:hr_msm_sim}

In this section, we derive the form of the HR-MSM in~\eqref{eq:msm-sim}, and show that it is implied by the data-generating mechanism of the simulation study. First, observe that for $t \geq 2$,
\begin{align*}
    Y_t(d_{t - 1}, d_t) = \alpha_1 X_{t - 1} + \alpha_2 d_{t - 1}(X_{t - 1}) + \alpha_3 X_t(d_{t - 1}) + \alpha_4 d_{t}(X_t(d_{t - 1})) + \epsilon_t,
\end{align*}
for some exogenous $\epsilon_t \sim \mathcal{N}(0, \sigma_t^2)$, where $X_t(d_{t - 1})$ is the potential $X_t$ value under the intervention setting $A_{t -1}$ to $d_{t - 1}(X_{t - 1})$. Note that $d_{t - 1}(X_{t - 1}) = J_{t - 1}X_{t - 1}$, and by our structural equations, $X_t(d_{t - 1}) \sim \mathrm{Bernoulli}(0.4 + 0.4d_{t - 1}(X_{t - 1}))$, so that $\mathbb{E}(X_t(d_{t - 1})) = 0.4 + 0.2 J_{t - 1}$, recalling that $\mathbb{E}(X_{t - 1}) = 0.5$. Finally, $d_t(X_t(d_{t - 1})) = J_t \cdot \mathrm{Bernoulli}(0.4 + 0.4d_{t - 1}(X_{t - 1}))$, which gives $\mathbb{E}(d_t(X_t(d_{t - 1}))) = \{0.4 + 0.2 J_{t - 1}\}J_t$. Putting everything together, we obtain
\begin{align*}
    \mathbb{E}(Y_t(d_{t - 1}, d_t)) 
    &= 0.5 \alpha_1  + 0.5 J_{t - 1} \alpha_2 + \{0.4 + 0.2 J_{t - 1}\} \{\alpha_3 + J_t \alpha_4\} \\
    &= \{0.5 \alpha_1 + 0.4 \alpha_3\} + \{0.5 \alpha_2 + 0.2 \alpha_3\}J_{t - 1} + 0.4 \alpha_4\,J_{t} + 0.2 \alpha_4 \, J_{t - 1}J_t \\
    & \equiv \beta_{0} + \beta_1 J_{t - 1} + \beta_{2} J_t + \beta_3 J_{t - 1} J_t,
\end{align*}
as claimed.

\subsection{Nuisance Function Estimation for Multiply Robust Estimator}\label{app:nuisance-deriv}

Recall that, for any $(d_{t - 1}, d_t) \in \overline{\mathcal{D}}_{2, t}$, $b_1^{d_t}(H_t) = \mathbb{E}(Y_t \mid H_t, A_t = d_t(H_t))$ and $b_2^{d_{t-1}, d_t}(H_{t - 1}) = \mathbb{E}(b_1^{d_t}(H_t) \mid H_{t - 1}, A_{t - 1} = d_{t - 1}(H_{t - 1}))$. Observe that, as the data-generating mechanism defines $Y_t = \alpha_1 X_{t - 1} + \alpha_2 A_{t - 1}  + \alpha_3 X_t + \alpha_4 A_t + \epsilon_t$, where $\epsilon_t \sim \mathcal{N}(0, \sigma_t^2)$ is some independent auxiliary noise variable, we have
\begin{align*}
    b_1^{d_t}(H_t) &= \mathbb{E}(Y_t \mid H_t, A_t = d_t(H_t)) \\
    &= \mathbb{E}(Y_t \mid X_{t - 1}, A_{t - 1}, X_t, A_t = J_t \, X_t) \\
    &= \alpha_1 X_{t - 1} + \alpha_2 A_{t - 1} + \alpha_3 X_t + \alpha_4 J_t \, X_t.
\end{align*}

Further,
\begin{align*}
\mathbb{E}(b_1^{d_t}(H_t) \mid H_{t - 1}, A_{t - 1})
    &= \mathbb{E}(\alpha_1 X_{t - 1} + \alpha_2 A_{t - 1} + \alpha_3 X_t + \alpha_4 J_t \, X_t \mid X_{t - 2}, A_{t - 2}, X_{t - 1}, A_{t - 1}) \\
    &= \alpha_1 X_{t - 1} + \alpha_2 A_{t - 1} + (\alpha_3 + \alpha_4 J_t) \mathbb{E}(X_t \mid A_{t - 1}) \\
    &= \alpha_1 X_{t - 1} + \alpha_2 A_{t - 1} + (\alpha_3 + \alpha_4 J_t) \{0.4 + 0.4 A_{t - 1}\} \\
    &= 0.4 \alpha_3 + 0.4 \alpha_4 J_t + \alpha_1 X_{t - 1} + \left\{\alpha_2 + 0.4 \alpha_3 + 0.4 \alpha_4 J_t\right\}A_{t - 1}.
\end{align*}
As a consequence,
\begin{align*}
    b_2^{d_{t-1}, d_t}(H_{t - 1}) &= \mathbb{E}(b_1^{d_t}(H_t) \mid H_{t - 1}, A_{t - 1} = d_{t - 1}(H_{t - 1})) \\
    &= 0.4 \alpha_3 + 0.4 \alpha_4 J_t + \alpha_1 X_{t - 1} + \left\{\alpha_2 + 0.4 \alpha_3 + 0.4 \alpha_4 J_t\right\}J_{t - 1} \, X_{t - 1}.
\end{align*}
To consistently model these nuisance functions in our simulations, we use the following multi-step approach: (1) fit a linear model of $Y_t$ on $(X_{t - 1}, A_{t - 1}, X_{t}, A_{t})$ to obtain unbiased estimates $(\widehat{\alpha}_1, \widehat{\alpha}_2, \widehat{\alpha}_3, \widehat{\alpha}_4)$; (2) for $j \in \{0,1\}$, compute $\widehat{b}_1^{d_{t}^{(j)}}(H_t) = \widehat{\alpha}_1 X_{t - 1} + \widehat{\alpha}_2 A_{t - 1} + \widehat{\alpha}_3 X_t + \widehat{\alpha}_4 j \, X_t$, and regress this with a linear model on an intercept, $X_{t - 1}$, and $A_{t - 1}$ to obtain corresponding coefficients $(\widehat{\gamma}_{0,j}, \widehat{\gamma}_{1,j}, \widehat{\gamma}_{2,j})$; (3) for $j_1, j_2 \in \{0,1\}$, compute $\widehat{b}_2^{d_{t - 1}^{(j_1)}, d_{t}^{(j_2)}}(H_{t - 1}) = \widehat{\gamma}_{0,j_2} + \widehat{\gamma}_{1,j_2}X_{t - 1} + \widehat{\gamma}_{2,j_2} j_1 X_{t - 1}$.

\subsection{IPW Simulation Results}
In Appendix Figure~\ref{fig:sim_results_appendix1} we present the same IPW estimator results as in Figure~\ref{fig:sim_results} but with more sample sizes, $n$ and timepoints, $T$. We also present these same simulation results in terms of the $\boldsymbol{\beta}$ coefficients of HR-MSM \refeq{eq:msm-sim}. In the main text we presented results in terms of sequential excursion effect parameters, which are linear combinations of these HR-MSM regression coefficients. 

\begin{figure*}[t] 
	\centering
	\begin{subfigure} 
		\centering
\includegraphics[width=0.99\linewidth]{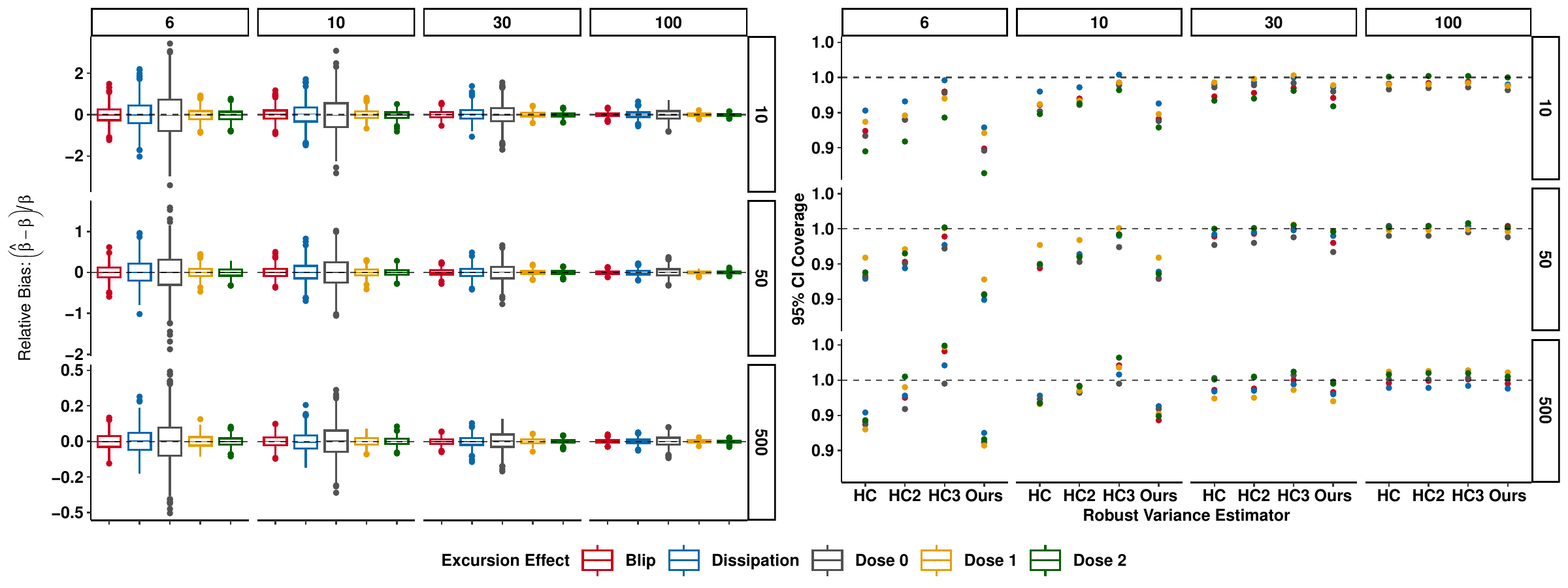}
	\end{subfigure}
	\caption{\footnotesize \textbf{IPW Simulation Results} Panel columns indicate sample sizes, $n$, and panel rows indicate number of timepoints, $T$ (cluster sizes). [Left] Relative bias associated with each sequential excursion effect. These results show that our estimator is consistent for the target parameters. [Right] 95\% Confidence interval (CI) coverage for the the sequential excursion effects. The coverage of 95\% CIs constructed using one of three established robust variance estimators and our robust variance estimator. The nominal coverage is reached for either large $n$ or large $t$ for all estimators.} 
\label{fig:sim_results_appendix1}
\end{figure*}

\begin{figure*}[t] 
	\centering
		\begin{subfigure} 
		\centering
\includegraphics[width=0.99\linewidth]{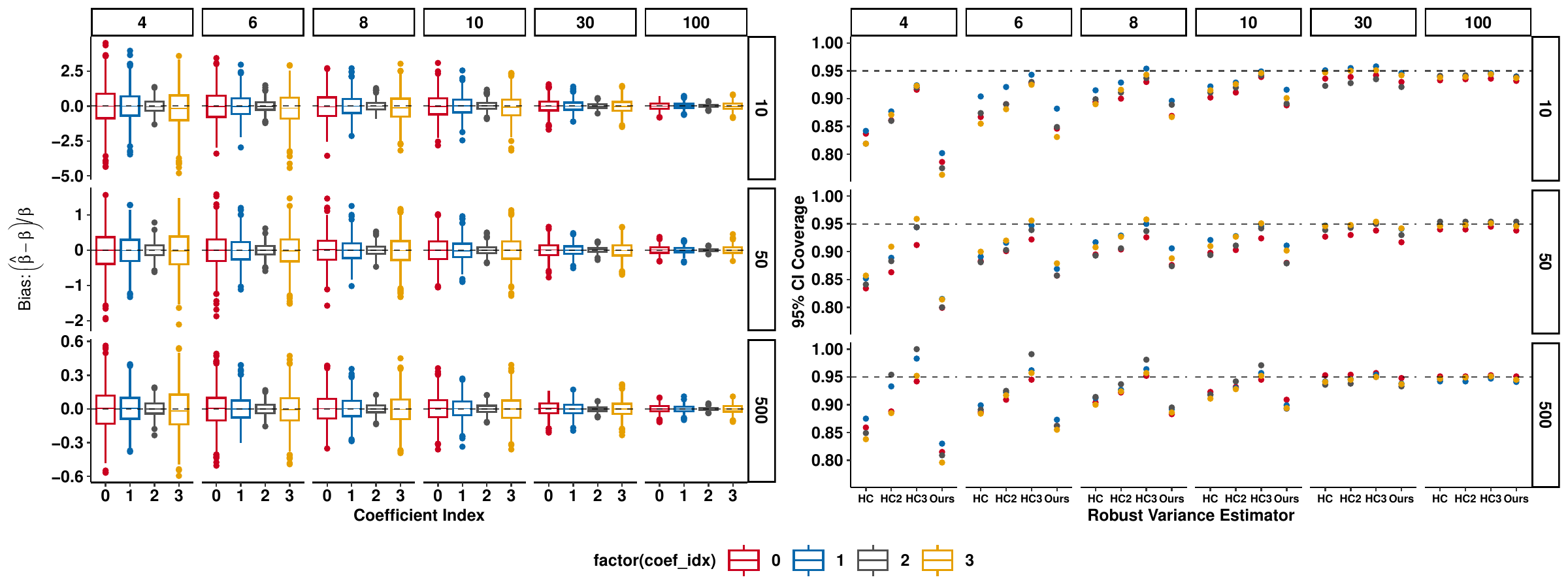}
	\end{subfigure}
	\caption{\footnotesize \textbf{IPW Simulation Results in Terms of MSM Coefficients} Relative bias and 95\% Confidence Interval (CI) coverage of regression coefficients of history-restricted marginal structural model (MSM) \refeq{eq:msm-sim}. Panel columns indicate sample sizes, $n$, and panel rows indicate number of timepoints, $T$ (cluster sizes). [Left] Relative bias associated with each HR-MSM regression coefficient. These results show that our estimator is consistent for the target parameters. [Right] 95\% CI coverage for the MSM coefficients. The coverage of 95\% CIs constructed using one of three established robust variance estimators and our robust variance estimator. The nominal coverage is reached for either large $n$ or large $t$ for all estimators.} 	\label{fig:sim_results_appendix2}
\end{figure*}

\subsection{Multiple Robust Simulation Results}\label{app:MR_sims}

In Appendix Figure~\ref{fig:sim_results_appendix1_mr} we present Multiply Robust estimator results from the simulation experiments presented in Figure~\ref{fig:sim_results}. 
We also present these same simulation results in terms of the $\boldsymbol{\beta}$ coefficients of HR-MSM \refeq{eq:msm-sim}. In the main text we presented results in terms of sequential excursion effect parameters, which are linear combinations of these HR-MSM regression coefficients. 

\begin{figure*}[t] 
	\centering
		\begin{subfigure} 
		\centering
\includegraphics[width=0.99\linewidth]{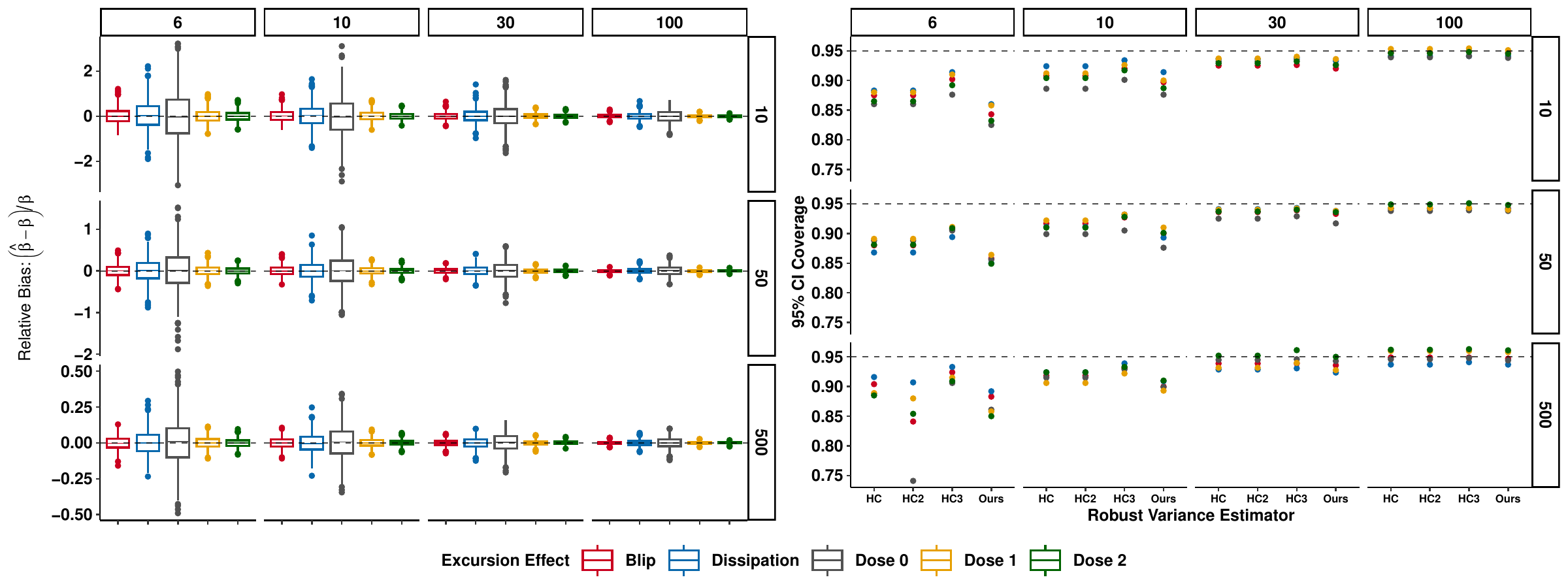}
	\end{subfigure}
	\caption{\footnotesize \textbf{Multiply Robust Simulation Results} Panel columns indicate sample sizes, $n$, and panel rows indicate number of timepoints, $T$ (cluster sizes). [Left] Relative bias associated with each sequential excursion effect. These results show that our estimator is consistent for the target parameters. [Right] 95\% Confidence interval (CI) coverage for the the sequential excursion effects. The coverage of 95\% CIs constructed using one of three established robust variance estimators and our robust variance estimator. The nominal coverage is reached for either large $n$ or large $t$ for all estimators.} 
\label{fig:sim_results_appendix1_mr}
\end{figure*}

\begin{figure*}[t] 
	\centering
		\begin{subfigure} 
		\centering
\includegraphics[width=0.99\linewidth]{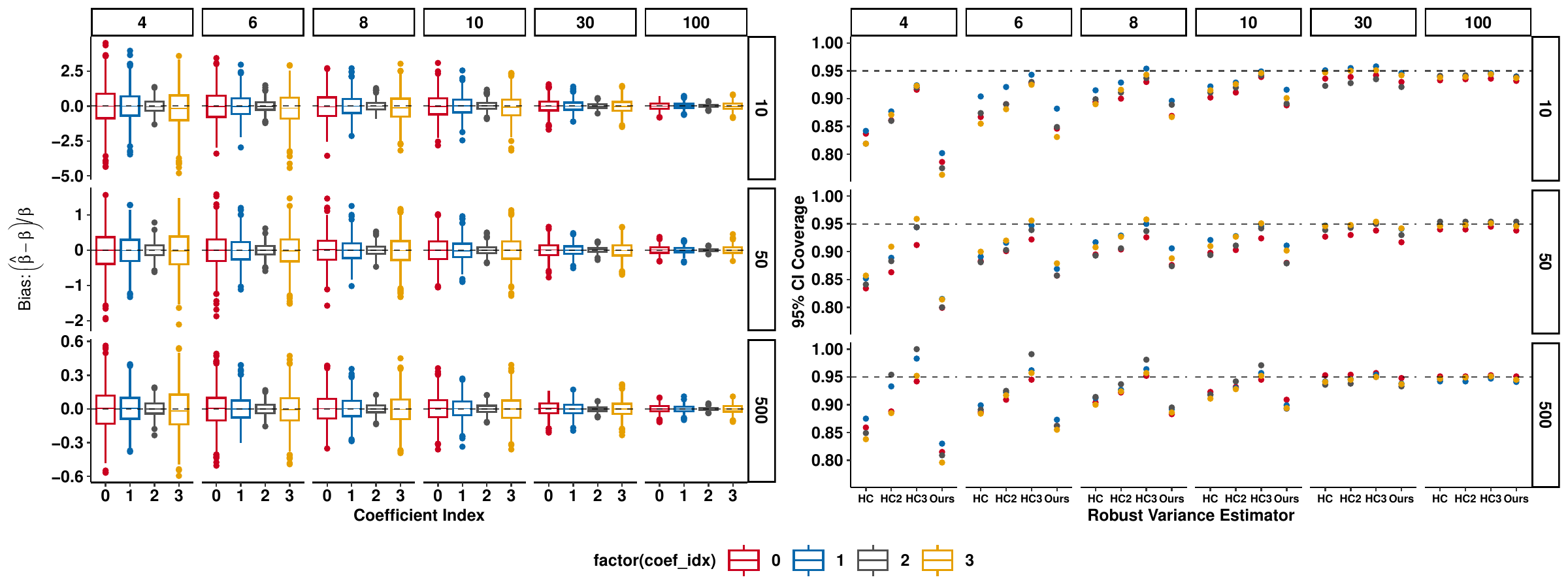}
	\end{subfigure}
	\caption{\footnotesize \textbf{Multiply Robust Simulation Results in Terms of MSM Coefficients} Relative bias and 95\% Confidence Interval (CI) coverage of regression coefficients of history-restricted marginal structural model (MSM) \refeq{eq:msm-sim}. Panel columns indicate sample sizes, $n$, and panel rows indicate number of timepoints, $T$ (cluster sizes). [Left] Relative bias associated with each HR-MSM regression coefficient. These results show that our estimator is consistent for the target parameters. [Right] 95\% CI coverage for the MSM coefficients. The coverage of 95\% CIs constructed using one of three established robust variance estimators and our robust variance estimator. The nominal coverage is reached for either large $n$ or large $t$ for all estimators.} 	\label{fig:sim_results_appendix2_mr}
\end{figure*}

\subsection{IPW and Multiply Robust Comparisons}
We next provide results from the same simulations reformatted to enable easy comparison of IPW and Multiply Robust estimator performance.

\begin{figure*} 
	\centering
		\begin{subfigure}
		\centering
\includegraphics[width=0.95\linewidth]{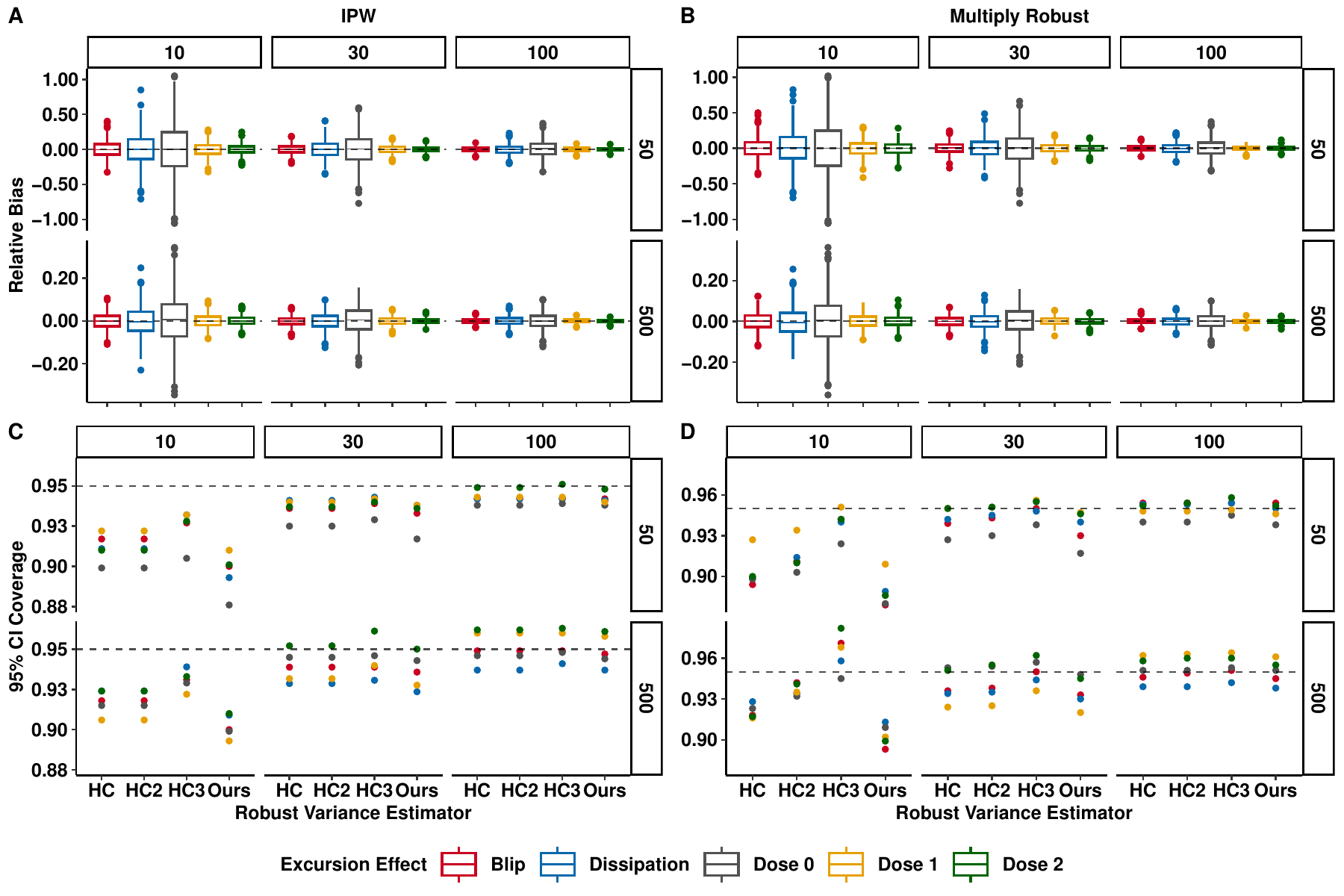}
	\end{subfigure}	\caption{\footnotesize \textbf{Simulation Results} Panel columns indicate sample sizes, $n$ (10, 30, or 100), and panel rows indicate number of timepoints, $T$ (cluster sizes, 50 or 500). [Left] IPW. [Right] Multiple Robust. [A]-[B] Relative bias, defined as $(\widehat{\beta} - \beta) / \beta$, associated with each sequential excursion effect. These results show that our estimator is consistent for the target parameters. [C]-[D] 95\% Confidence interval (CI) coverage for the sequential excursion effects. The coverage of 95\% CIs constructed using one of three established robust variance estimators and our robust variance estimator. The nominal coverage is reached for either large $n$ or large $t$ for all estimators.} 
		\label{fig:sim_results_joined}
\end{figure*}

\clearpage
\section{Sequential Excursion Effect Interpretation and Modeling Choices}\label{app:model_interp}

We hope this work encourages discussion about which causal estimands to target in optogenetics designs, given the drawbacks of analysis conventions that compare summaries of the outcome across groups. 
As we discussed, availability status is both a mediator of past exposure effects, and a confounder of the subsequent exposure-outcome relationship. As a result, estimating certain treatment effects with standard methods, like a regression (with or without including $I_t$ as a covariate), will be biased \citep{robins1986}. 

A natural question is, when is ``standard regression using treatment as a covariate'' a valid approach (e.g., see mixed model applied in \cite{sex_diff}). In cases when analysts fit regressions on the subset of timepoints when optogenetic stimulation is available (i.e., timepoints $\{t:~I_t = 1\}$), this is analogous to availability conditional excursion effects ($\Gamma = 1$) proposed in \cite{Boruvka2018}, albeit the regression coefficient point and variance estimators may differ. As emphasized above, such an analysis in closed-loop studies would yield biased estimates and would not have a causal interpretation in analyses of treatment sequences longer than one timepoint (i.e., $\Gamma > 1$).
However, in open-loop designs, treatment probabilities are marginally randomized and so standard unweighted regression 
(e.g., a GEE, possibly with a non-independence working correlation structure) of outcomes on the model $m$ yields valid estimates of the causal HR-MSM parameters, as discussed in Section~\ref{sec:ident-est}. Since mixed models are frequently employed in neuroscience, it is worth pointing out that these provide both \textit{marginal} and \textit{conditional} (on random effects) estimates of the same parameters in the case of a linear model, but may be biased for the marginal HR-MSM parameters when a non-linear link is used, due to non-collapsibility.\footnote{For an intuitive introduction to this equivalence between conditional and marginal interpretation of regression coefficients of linear mixed models, we recommend sections 2.2 and 7.4 of \cite{fitz_longitudinal}.} We caution, however, against using linear mixed models in closed-loop designs, as it may yield biased estimates for the marginal HR-MSM parameters due to the inherent estimation of the correlation structure. As mentioned in the main text, past work recommends that working independence correlation structures~\citep{liang1986} be used in MSMs, as alternative correlation structures may lead to biased MSM parameter estimates~\citep{tchetgen2012}. 
As discussed in the Section~\ref{sec:recs-opto},
one should never include time-varying variables occurring after timepoint $t-\Gamma+1$ as an effect modifier in the HR-MSM. For example, it would be invalid to include interactions between behavioral variables, measured on timepoints $t \in \{t-\Gamma + 1,...,t\}$, and treatment. Intuitively, this would condition on a mediator, potentially inducing bias. One could, however, include \textit{any} data (e.g., outcome, treatments, time-varying confounders), for all timepoints $t \in [t-\Gamma]$. 

\end{appendices}

\clearpage

\singlespacing

\singlespacing
\bibliographystyle{asa}
{\small \bibliography{refs}}

\end{document}